\documentclass[%
 superscriptaddress,
 nofootinbib,
 amsmath,amssymb,
 aps,
 twocolumn, 
 prl,
 floatfix,
]{revtex4-2}

\usepackage{mathtools}
\usepackage{MnSymbol}
\usepackage{bm}
\usepackage[utf8]{inputenc}
\usepackage[T1]{fontenc}
\usepackage{physics}
\usepackage{dsfont}
\usepackage{amsthm}
\usepackage{hyperref}
\usepackage{floatrow}
\usepackage{graphicx}
\usepackage[caption=false]{subfig}
\usepackage{algorithm}
\usepackage{algpseudocode}
\usepackage{mathrsfs}
\usepackage{amssymb}
\usepackage{qcircuit}
\usepackage{enumerate}

\def\<{\langle}\def\>{\rangle}
 
\def\rank{\mathsf{rank}}

\def\Choi{\mathsf{Choi}}
\def\dim{\mathsf{dim}}

\def\bbra#1{\llangle#1\rvert}
\def\kett#1{\lvert#1\rrangle}
\def\Conv{\mathsf{Conv}}
\def\Supp{\mathsf{Supp}}
\DeclareMathOperator*{\argmax}{arg\,max}
\DeclareMathOperator*{\argmin}{arg\,min}

\def\Range{\mathsf{Range}}
\newcommand{\RNum}[1]{\uppercase\expandafter{\romannumeral #1\relax}}

\theoremstyle{plain}
\newtheorem{theorem}{Theorem}
\newtheorem{lemma}{Lemma}

\theoremstyle{definition}
\newtheorem{definition}{Definition}

\theoremstyle{remark}

\makeatletter 
    
\renewcommand\onecolumngrid{% <<<<<<
\do@columngrid{one}{\@ne}%
\def\set@footnotewidth{\onecolumngrid}% <<<<<<<<<<<<<<<<
\def\footnoterule{\kern-6pt\hrule width 1.5in\kern6pt}%
}

\renewcommand\twocolumngrid{% <<<<<<
        \def\footnoterule{% restore rule
        \dimen@\skip\footins\divide\dimen@\thr@@
        \kern-\dimen@\hrule width.5in\kern\dimen@}
        \do@columngrid{mlt}{\tw@}
}%

\makeatother    
\begin{document}

\preprint{APS/123-QED}

\title{Optimal Strategies of Quantum Metrology with a Strict Hierarchy}
\author{Qiushi Liu}
\email{qsliu@cs.hku.hk}
\affiliation{
 QICI Quantum Information and Computation Initiative, Department of Computer Science, The University of Hong Kong, Pokfulam Road, Hong Kong, China
}
\author{Zihao Hu}
\email{zhhu@mae.cuhk.edu.hk}
\affiliation{
Department of Mechanical and Automation Engineering, The Chinese University of Hong Kong, Shatin, Hong Kong, China
}
\author{Haidong Yuan}
\email{hdyuan@mae.cuhk.edu.hk}
\affiliation{
Department of Mechanical and Automation Engineering, The Chinese University of Hong Kong, Shatin, Hong Kong, China
}
\author{Yuxiang Yang}
\email{yuxiang@cs.hku.hk}
\affiliation{
 QICI Quantum Information and Computation Initiative, Department of Computer Science, The University of Hong Kong, Pokfulam Road, Hong Kong, China
}
\date{\today}

\begin{abstract}
One of the main quests in quantum metrology is to attain the ultimate precision limit with given resources, where the resources are not only of the number of queries, but more importantly of the allowed strategies. With the same number of queries, the restrictions on the strategies constrain the achievable precision. In this work, we establish a systematic framework to identify the ultimate precision limit of different families of strategies, including the parallel, the sequential, and the indefinite-causal-order strategies, and provide an efficient algorithm that determines an optimal strategy within the family of strategies under consideration. With our framework, we show there exists a strict hierarchy of the precision limits for different families of strategies.
\end{abstract}

\maketitle

\emph{Introduction}.—Quantum metrology \cite{Giovannetti2011,Degen2017RMP} features a series of promising applications in the near future \cite{Martinis2015}. In the prototypical setting of quantum metrology, the goal is to estimate an unknown parameter carried by a quantum channel, given $N$ queries to it. A pivotal task is to design a \emph{strategy} that utilizes these $N$ queries to generate a quantum state with as much information about the unknown parameter as possible. This often involves, for example, preparing a suitable input probe state \cite{Lee2002JMP,Buzek1999PRL,Kitagawa1993PRA} and applying intermediate quantum control \cite{Demkowicz-Dobrzanski14PRL,Yuan2015PRL,Yuan2016PRL,Pang2017} as well as quantum error correction \cite{Duer2014PRL,Kessler2014PRL,Demkowicz-Dobrza2017PRX,Zhou2018}.   

In reality, the implementation of strategies is subject to physical restrictions. In particular, within the noisy and intermediate-scale quantum (NISQ) era \cite{Preskill2018quantumcomputingin}, we have to adjust the strategy to accommodate the limitations on the system. For example, for systems with short coherence time it might be favorable to adopt the parallel strategy [Fig.~\ref{fig:all strategies}\protect\subref{subfig:parallel strategy}], where multiple queries of the unknown channel are applied simultaneously on a multipartite entangled state \cite{Lee2002JMP}. When the system has longer coherence time and can be better controlled, one could choose
to query the channel sequentially [Fig.~\ref{fig:all strategies}\protect\subref{subfig:sequential strategy}], which may potentially enhance the precision. In addition to the parallel and sequential strategies, it was recently discovered that the quantum SWITCH \cite{Chiribella2013PRA}, a primitive where the order of making queries to the unknown channel is in a quantum superposition [Fig.~\ref{fig:all strategies}\protect\subref{subfig:quantum switch strategy}], can be employed to generate new strategies of quantum metrology \cite{Chapeau-Blondeau2021PRA,mukhopadhyay2018superposition,Zhao2020PRL} that may even break the Heisenberg limit \cite{Zhao2020PRL}. Moreover, indefinite causal structures beyond the quantum SWITCH \cite{Chiribella2013PRA,Oreshkov2012,Araujo2015} [Figs. \ref{fig:all strategies}\protect\subref{subfig:causal superposition strategy} and \protect\subref{subfig:general strategy}] have recently been shown to further boost the performance of certain information processing tasks \cite{Quintino2019PRL,Bavaresco2021PRL}. The ultimate performance of these strategies in quantum metrology, however, remains unknown. This is mainly due to the lack of a systematic method that optimizes the probe state, the control and other degrees of freedom in a strategy in a unified fashion which leads to the ultimate precision limit.

\begin{figure}[!htbp]
    \captionsetup{position=t}
    \centering
    \sidesubfloat[]{\includegraphics[height=0.235\linewidth]{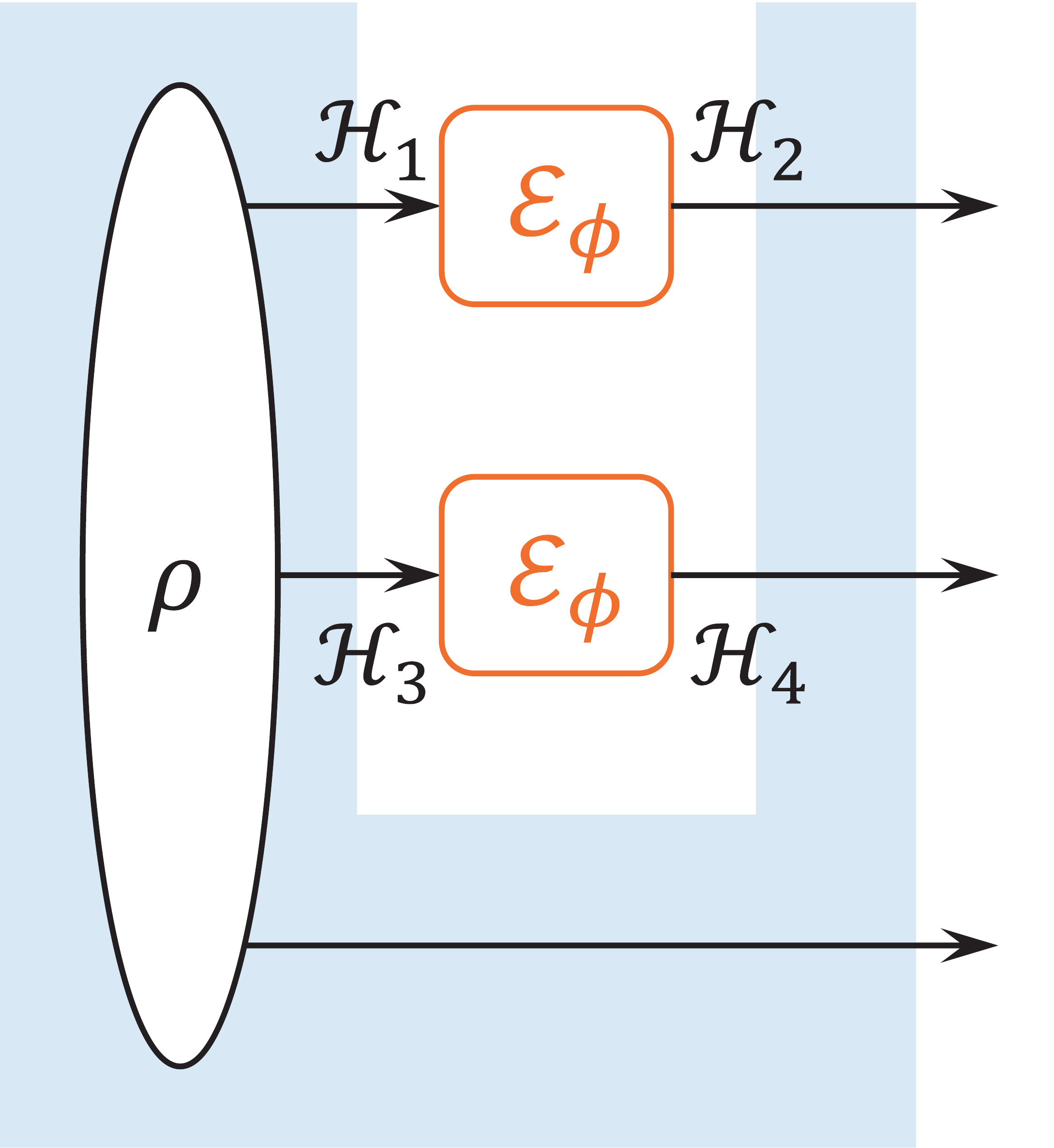}\label{subfig:parallel strategy}} 
    \sidesubfloat[]{\includegraphics[height=0.235\linewidth]{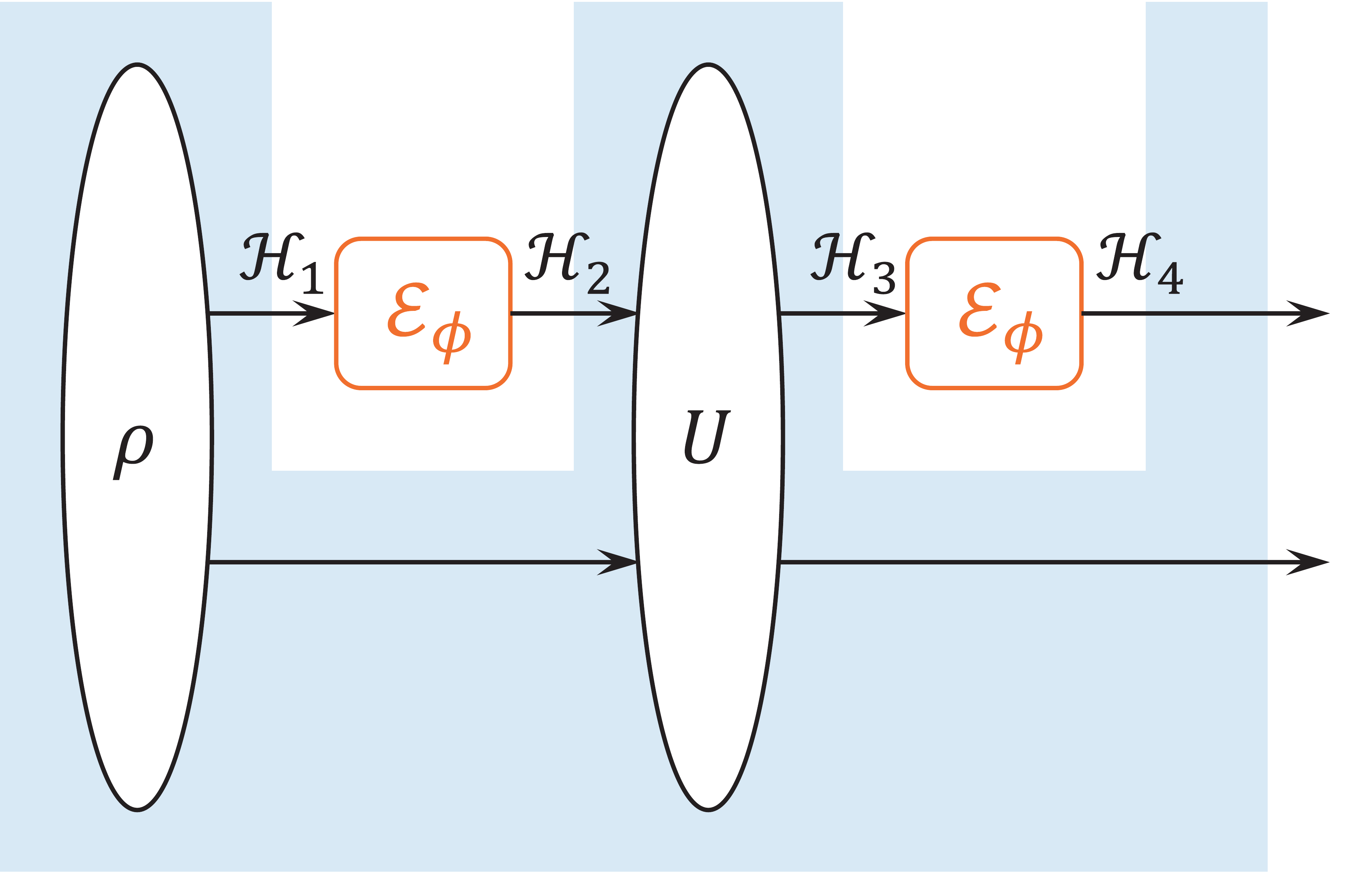}\label{subfig:sequential strategy}} \\ \vspace{0.5cm}
    \sidesubfloat[]{\includegraphics[height=0.235\linewidth]{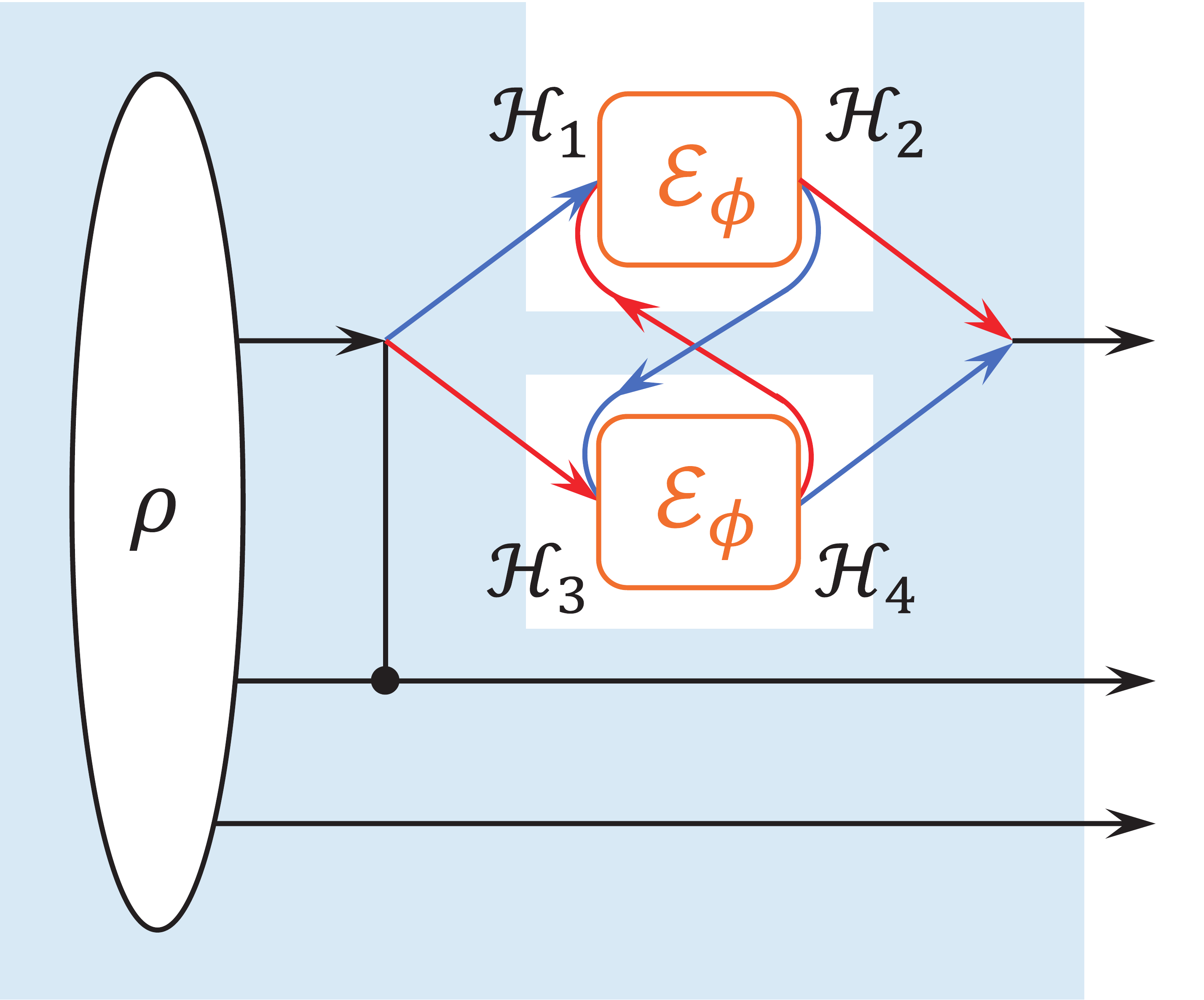} \label{subfig:quantum switch strategy}}
    \sidesubfloat[]{\includegraphics[height=0.235\linewidth]{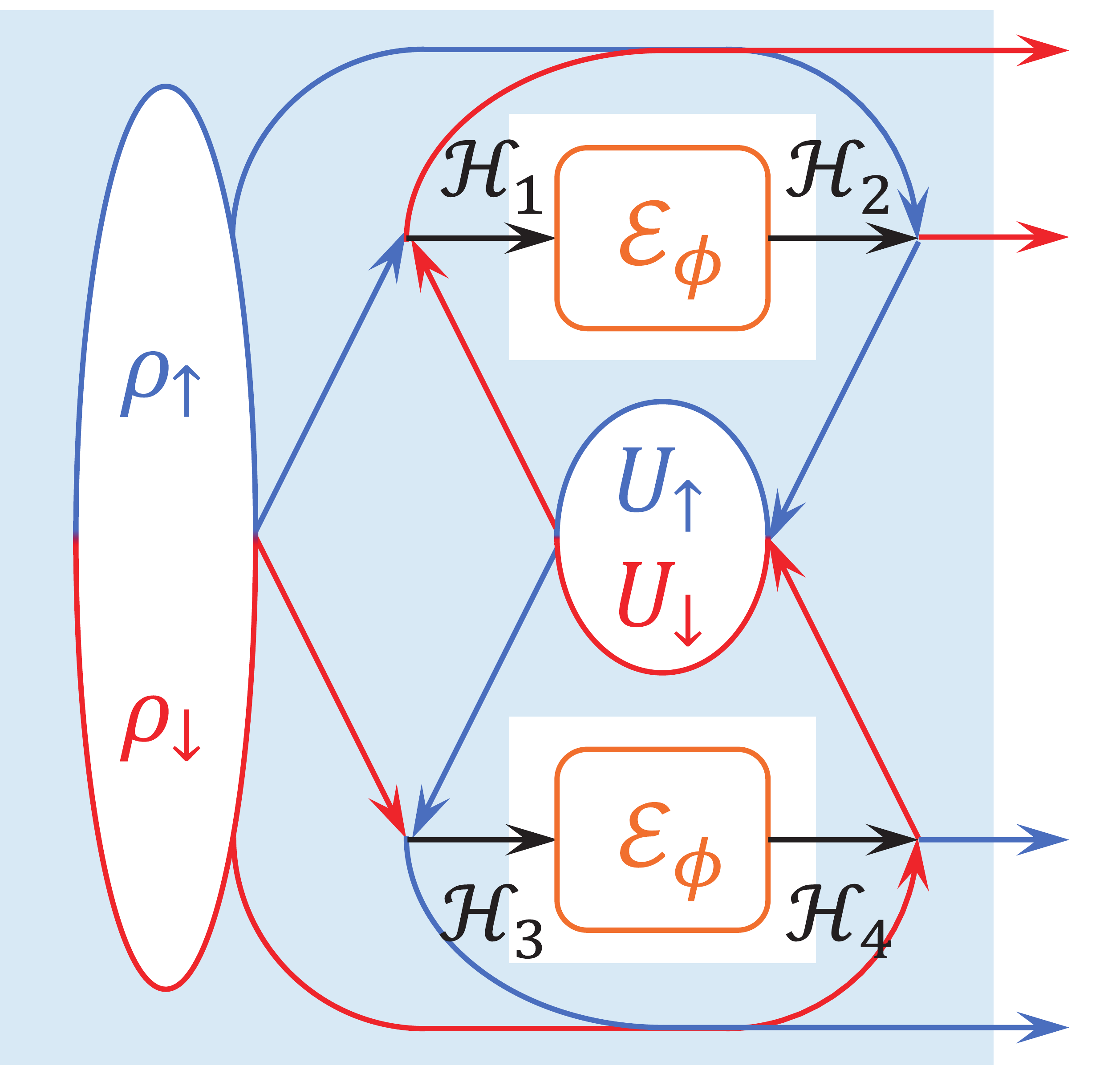}\label{subfig:causal superposition strategy}} 
    \sidesubfloat[]{\includegraphics[height=0.235\linewidth]{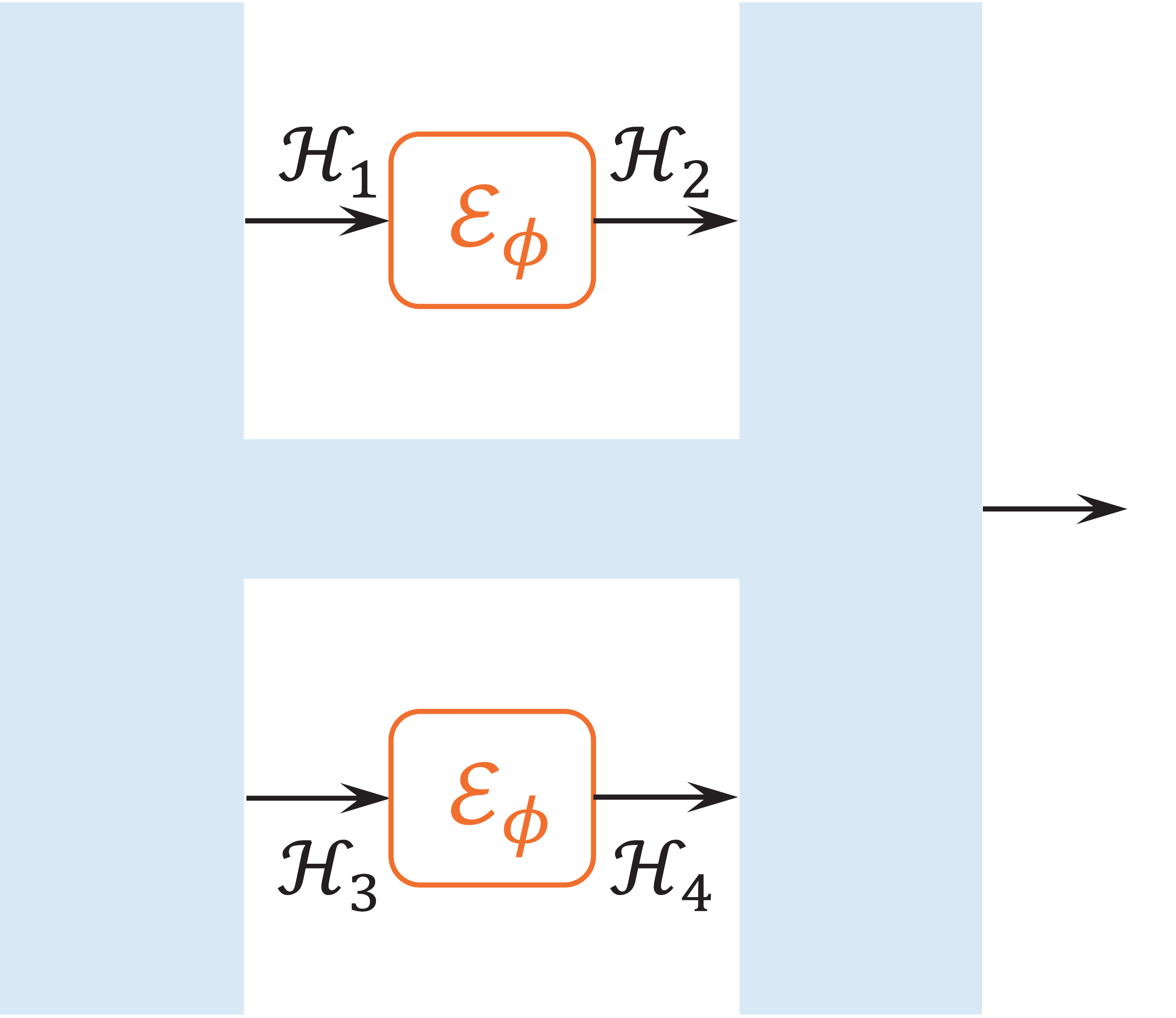}\label{subfig:general strategy}}
    \caption{\label{fig:all strategies}\textbf{Prototypical strategies of quantum metrology} (for the
    $N=2$ case). $\mathcal E_\phi$ is a quantum channel carrying an unknown parameter $\phi$, and the blue shaded area represents a \emph{strategy}. \protect\subref{subfig:parallel strategy} A parallel strategy. \protect\subref{subfig:sequential strategy} A sequential strategy, where $U$ is a control operation. \protect\subref{subfig:quantum switch strategy} A quantum SWITCH strategy. The blue and red lines, respectively, correspond to two different execution orders entangled with a control qubit. \protect\subref{subfig:causal superposition strategy} A causal superposition strategy. Two sequential strategies, plotted in blue and in red, respectively, are entangled with a control qubit (not shown in the figure) and the output will be measured with the control qubit collectively. \protect\subref{subfig:general strategy} A general indefinite-causal-order strategy.}
\end{figure} 

In this work, we develop a semidefinite programming (SDP) method of evaluating the optimal precision of single-parameter quantum metrology for finite $N$ (which we call \emph{the nonasymptotic regime}) over a family of admissible strategies. With this method, we show a strict hierarchy (see Fig. \ref{fig:qfi all strategies}) of the optimal performances under different families of strategies, which include the parallel, the sequential, and the indefinite-causal-order \cite{Chiribella2013PRA,Oreshkov2012,Araujo2015} ones (see Fig.~\ref{fig:all strategies}). In conjunction, we design an algorithm to obtain an optimal strategy achieving the highest precision. For the strategy set that admits a symmetric structure, we develop a method of reducing the complexity of our algorithms by an exponential factor.

\emph{Quantum Fisher information}.—The uncertainty $\delta \hat{\phi}$ of estimating an unknown parameter $\phi$ encoded in a quantum state $\rho_\phi$, for any unbiased estimator $\hat{\phi}$, can be determined via the quantum Cramér-Rao bound (QCRB) as  $\delta \hat{\phi} \ge 1/\sqrt{\nu J_Q(\rho_\phi)}$\cite{helstrom1976quantum,holevo2011probabilistic,Braunstein1994PRL}, where $J_Q(\rho_\phi)$ is the quantum Fisher information (QFI) of the state $\rho_\phi$ and $\nu$ is the number of repeated measurements \cite{Asymptotic_note}. For single-parameter estimation, the QCRB is achievable, and the QFI thus quantifies the amount of information that can be extracted from the quantum state. 
One way to compute the QFI is \cite{Fujiwara2008,Escher2011} 
\begin{equation} \label{eq:qfi_state}
    J_Q(\rho_\phi) = 4\min_{\Tr_A{(\dyad{\Psi_\phi})} = \rho_\phi}\braket{\dot{\Psi}_\phi},
\end{equation}
where $\ket{\Psi_\phi}$ is the purification of $\rho_\phi$ with an ancillary space $\mathcal H_A$, $\Tr_A$ denotes the partial trace over $\mathcal H_A$, and $\dot{\Psi}:=\partial \Psi /\partial \phi$. 
When the parameter is carried by a quantum channel $\mathcal E_\phi$, i.e., a completely positive trace-preserving (CPTP) map, the channel QFI can be defined as the maximal QFI of output states using the optimal input assisted by arbitrary ancillae \cite{Fujiwara2008,Demkowicz-Dobrzanski2012,yuan2017fidelity,yuan2017quantum}: $
    J_Q^{(\mathrm{chan})}(\mathcal E_\phi) = \max_{\rho_{\mathrm{in}} \in \mathcal S(\mathcal H_S \otimes \mathcal H_A)}J_Q[(\mathcal E_\phi \otimes \mathcal I_A)(\rho_{\mathrm{in}})]$, 
where $\mathcal S(\mathcal H)$ denotes the space of density operators on the Hilbert space $\mathcal H$, $\mathcal H_{S/A}$ denotes the Hilbert space of the system or ancillae, and $\mathcal I_A$ is the identity on $\mathcal H_A$.

We denote by $\mathcal L(\mathcal H)$ the set of linear operators on the finite-dimensional Hilbert space $\mathcal H$, and $\mathcal L[\mathcal L(\mathcal H_1), \mathcal L(\mathcal H_2)]$ denotes the set of linear maps from $\mathcal L(\mathcal H_1)$ to $\mathcal L(\mathcal H_2)$. By the Choi-Jamiołkowski (CJ) isomorphism, a parametrized quantum channel $\mathcal E_\phi \in \mathcal L[\mathcal L(\mathcal H_{2i-1}), \mathcal L(\mathcal H_{2i})]$ (for $1 \le i \le N$) can be represented by a positive semidefinite operator (called the CJ operator) $E_\phi = \Choi(\mathcal E_\phi) = \mathcal E_\phi \otimes \mathcal I(\kett I \bbra I)$, where $\kett I=\sum_i \ket{i} \ket{i}$. The CJ operator of $N$ identical quantum channels is $N_\phi = E_\phi^{\otimes N} \in \mathcal L\left(\otimes_{i=1}^{2N} \mathcal H_i\right)$.

\emph{Strategy set in quantum metrology}.—A \emph{strategy} is an arrangement of physical processes (the blue shaded area in Fig.~\ref{fig:all strategies}) which, when concatenated with given queries to $\mathcal{E}_{\phi}$, generates an output quantum state carrying the information about $\phi$. A strategy can be described by a CJ operator on $\mathcal L\left(\mathcal H_F\otimes_{i=1}^{2N} \mathcal H_i\right)$, where $\mathcal H_F$ denotes the output Hilbert space of the concatenation, referred to as the \emph{global future} space. The concatenation of two processes is characterized by the link product \cite{Chiribella2008PRL,Chiribella2009PRA} of two corresponding CJ operators $A\in\mathcal L\left(\otimes_{a\in \mathcal A}\mathcal H_a\right)$ and $B\in\mathcal L\left(\otimes_{b\in \mathcal B}\mathcal H_b\right)$ as
\begin{equation}
    A*B := \Tr_{\mathcal A\cap\mathcal B}\left[\left(\mathds{1}_{\mathcal B\backslash\mathcal A}\otimes A^{T_{\mathcal A\cap\mathcal B}}\right)\left(B\otimes\mathds{1}_{\mathcal A\backslash\mathcal B}\right)\right],
\end{equation}
where $T_i$ denotes the partial transpose on $\mathcal H_i$, and $\mathcal H_{\mathcal A/\mathcal B}$ denotes $\otimes_{i\in \mathcal A/\mathcal B} \mathcal H_{i}$. The output state lies in the global future $F$, which should not affect any state in the past. 

Following the above formalism, given a sufficiently large ancillary Hilbert space, a \emph{strategy set} determined by the relevant causal constraints is described by a subset $\mathsf{P}$ of
\begin{equation}
    \mathsf{Strat}:=\left\{ P \in \mathcal L\left(\mathcal H_F\otimes_{i=1}^{2N} \mathcal H_i\right)\middle|P\ge0,\rank(P)=1 \right\}.
\end{equation}
Here, without loss of generality \cite{purification_note}, we have restricted $P$ to pure processes (rank-1 operators) due to the monotonicity of QFI \cite{petz2008quantum}. Our goal is to identify the ultimate precision limit of parameter estimation characterized by the QFI within such constraints: 
\begin{definition} \label{def:qfi}
The QFI of $N$ quantum channels $\mathcal E_\phi$ \cite{Yang2019PRL} given a strategy set $\mathsf P$ is 
\begin{equation} 
    J^{(\mathsf P)}(N_\phi) := \max_{P \in \mathsf P} J_Q(P * N_\phi),
\end{equation}
where $J_Q(\rho)$ is the QFI of the state $\rho$, and $N_\phi$ is the CJ operator of $N$ channels.
\end{definition}

In general we can write the ensemble decomposition \cite{Fujiwara2008} of the CJ operator $N_\phi$ as $N_\phi = \sum_{i=1}^{r}\dyad{N_{\phi,i}} = \mathbf N_\phi\mathbf N_\phi^\dagger$, where $\mathbf N_\phi := \left(\ket{N_{\phi,1}},\dots,\ket{N_{\phi,r}}\right)$ and $r := \max_\phi\rank({N_\phi})$. We also define $\dot{\tilde{\mathbf N}}_\phi:=\dot{\mathbf N}_\phi-\mathrm i \mathbf N_\phi h$ for $h \in \mathbb H_r$, where $\mathbb H_r$ denotes the set of $r\times r$ Hermitian matrices, and the performance operator \cite{Altherr2021PRL}
\begin{equation}
    \Omega_\phi(h) := 4\left(\dot{\tilde{\mathbf N}}_\phi\dot{\tilde{\mathbf N}}_\phi^\dagger\right)^T.
\end{equation}

With these notions, we can show (see Appendix \ref{app:lemma qfi}, which is analogous to the approach in \cite{Altherr2021PRL}) that the QFI admits the form
\begin{equation} \label{eq:qfi_primal}
    J^{(\mathsf P)}(N_\phi)  
    = \max_{\tilde{P} \in \tilde{\mathsf P}} \min_{h \in \mathbb H_r} \Tr\left[\tilde{P}\Omega_\phi(h)\right],
\end{equation}
with 
\begin{equation} \label{eq:marginal strategy set}
    \tilde{\mathsf P} := \left\{\tilde{P} = \Tr_FP \middle| P \in \mathsf P\right\}.
\end{equation}

To evaluate the QFI, we first exchange $\max_{\tilde{P}}$ and $\min_h$ without changing the optimal QFI, assured by the minimax theorem \cite{Fan1953, rockafellar1970convex} since the objective function is concave on $\tilde{P}$ and convex on $h$ \cite{minimax_note}. Hence, the problem is cast into
\begin{equation} \label{eq:primal qfi min max}
    J^{(\mathsf P)}(N_\phi) = \min_h\max_{\tilde P\in\tilde{\mathsf P}}\Tr\left[\tilde P\Omega_\phi(h)\right].
\end{equation}
Then we fix $h$ and formulate the dual problem of maximization over the set $\tilde{\mathsf P}$. Finally we further optimize the value of $h$. To simplify the calculation of QFI we require that
\begin{equation} \label{eq:condition of thm}
    \tilde{\mathsf P} = \Conv\left\{\bigcup_{i=1}^{K}\left\{S^i\ge0\middle|S^i\in\mathsf S^i\right\}\right\}, 
\end{equation}
where $\Conv\{\cdot\}$ denotes the convex hull, and each $\mathsf S^i$ for $i=1,\dots,K$ is an affine space of Hermitian operators. Adopting the above-mentioned method, we get
\begin{theorem} \label{thm:qfi}
Given an arbitrary strategy set $\mathsf{P}$ such that $\tilde{\mathsf P}$ given by Eq.~(\ref{eq:marginal strategy set}) satisfies the condition Eq.~(\ref{eq:condition of thm}), the QFI of $N$ quantum channels $\mathcal E_\phi$ can be expressed as the following optimization problem:
\begin{equation} \label{eq:thm}
    \begin{aligned}
    J^{(\mathsf P)}(N_\phi) =& \min_{\lambda,Q^i,h} \lambda, \\
    \mathrm{s.t.}\ &\,\lambda Q^i \ge \Omega_\phi(h),\ Q^i\in \overline{\mathsf S}^i,\ i=1,\dots,K,
    \end{aligned} 
\end{equation}
where $\overline{\mathsf {S}}^i:=\left\{Q\mathrm{\ is\ Hermitian}\middle|\Tr(QS)=1,S\in\mathsf S^i\right\}$ is the dual affine space of $\mathsf S^i$.
\end{theorem}

The proof can be found in Appendix \ref{app:thm qfi}. We remark that similar optimization ideas have been applied to other tasks, such as quantum Bayesian estimation \cite{Chiribella_2012}, quantum network optimization \cite{Chiribella_2016}, non-Markovian quantum metrology \cite{Altherr2021PRL}, and quantum channel discrimination \cite{Bavaresco2021PRL}. The minimization problem in Theorem \ref{thm:qfi} can be further written in the form of SDP and solved efficiently, with detailed numerically solvable forms given in Appendix \ref{app:qfi all strategies}, where the constraints in Eq. (\ref{eq:thm}) can be further simplified in some cases.

\emph{Optimal strategies}.—By itself, the QFI does not reveal how to implement the optimal strategy achieving the highest precision. Here, in addition to Theorem \ref{thm:qfi}, we design an algorithm that yields a strategy attaining the optimal QFI for any strategy set satisfying Eq.~(\ref{eq:condition of thm}). The method, which generalizes the method of finding an optimal probe state for a single channel \cite{PhysRevResearch.2.013235,PRXQuantum.2.010343}, is summarized as Algorithm \ref{alg:optimal probe} (see Appendix \ref{app:alg optimal probe} for its derivation).

\begin{algorithm} [!htbp]
\caption{Find an optimal strategy in the set $\mathsf P$.}\label{alg:optimal probe}
\begin{enumerate}[{(i)}]
    \item Given $N_\phi$ the CJ operator of $N$ channels, solve for an optimal value $h=h^{(\mathrm{opt})}$ in Eq. (\ref{eq:thm}) of Theorem \ref{thm:qfi} via SDP.
    \item Fixing $h=h^{(\mathrm{opt})}$, solve for an optimal value $\tilde P^{(\mathrm{opt})}$ of $\tilde P\in\mathsf{\tilde P}$ in Eq. (\ref{eq:qfi_primal}) via SDP such that
    \begin{equation} \label{eq:optimal condition}
    \begin{aligned}
        \Re \left\{\Tr \left\{\tilde P^{(\mathrm{opt})}\left[-\mathrm i \mathbf N_\phi \mathscr H\left(\dot{\mathbf N}_\phi-\mathrm i \mathbf N_\phi h^{(\mathrm{opt})}\right)^{\dagger}\right]^T\right\} \right\} \\
        =0\ \mathrm{for}\ \mathrm{all}\ \mathscr H \in \mathbb H_r,
    \end{aligned}
    \end{equation}
    where $\mathbf N_\phi := \left(\ket{N_{\phi,1}},\dots,\ket{N_{\phi,r}}\right)$. An optimal strategy $P^{(\mathrm{opt})} \in \mathsf P$ can be taken as a purification of $\tilde P^{(\mathrm{opt})}$.
\end{enumerate}
\end{algorithm}

By Algorithm \ref{alg:optimal probe} we obtain the CJ operator of a strategy that attains the optimal QFI. For strategies following definite causal order, there exists an operational method of mapping the CJ operator of the strategy to a probe state and a sequence of in-between control operations with minimal memory space \cite{Bisio2011PRA}. For causal order superposition strategies (see the strategy set $\mathsf{Sup}$), we show that they can always be implemented by controlling the order of operations in a circuit with a quantum SWITCH (see Appendix \ref{app:implementation of strategies}). In this way, we obtain a systematic method to identify optimal sequential and causal superposition strategies, one of the key problems in quantum metrology.    

\emph{Strategy sets}.—We consider the evaluation of QFI for five different families of strategies. In all the following definitions the subscript $i$ of an operator denotes the Hilbert space $\mathcal H_i$ it acts on. 
 
The family of parallel strategies [see Fig.~\ref{fig:all strategies}\protect\subref{subfig:parallel strategy}] is the first and one of the most successful examples of quantum-enhanced metrology, featuring the usage of entanglement to achieve precision beyond the classical limit \cite{Giovannetti2006PRL}. By making parallel use of $N$ quantum channels together with ancillae, we can regard these $N$ channels as one single channel from $\mathcal L\left(\otimes_{i=1}^{N}\mathcal H_{2i-1}\right)$ to $\mathcal L\left(\otimes_{i=1}^{N}\mathcal H_{2i}\right)$. A parallel strategy set $\mathsf{Par}$ is defined as the collection of $P \in \mathsf{Strat}$ such that \cite{Chiribella2009PRA}
\begin{equation}
     \Tr_F{P}= \mathds{1}_{2,4,\dots,2N} \otimes P^{(1)},\ \Tr P^{(1)}=1.
\end{equation}
Note that the optimal QFI of parallel strategies can also be evaluated using the method in \cite{Fujiwara2008,Demkowicz-Dobrzanski2012}.

A more general protocol is to allow for sequential use of $N$ channels assisted by ancillae, where only the output of the former channel can affect the input of the latter channel, and any control gates can be inserted between channels [see Fig.~\ref{fig:all strategies}\protect\subref{subfig:sequential strategy}]. A sequential strategy set $\mathsf{Seq}$ is defined as the collection of $P \in \mathsf{Strat}$ such that \cite{Chiribella2009PRA}
\begin{equation}
    \begin{aligned}
    &\Tr_F P= \mathds{1}_{2N} \otimes P^{(N)},\ \Tr P^{(1)}=1,\\
    &\Tr_{2k-1} P^{(k)} =  \mathds{1}_{2k-2} \otimes P^{(k-1)},\ k=2,\dots, N.
    \end{aligned}
\end{equation}
Unlike the case of parallel strategies, there is no existing way of evaluating the exact QFI using sequential strategies.

We also consider families of strategies involving indefinite causal order. The first one, denoted by $\mathsf{SWI}$, takes advantage of the (generalized) quantum SWITCH \cite{COLNAGHI20122940,Araujo2014PRL}, where the execution order of $N$ channels is entangled with the state of an $N!$-dimensional control system  [see Fig.~\ref{fig:all strategies}\protect\subref{subfig:quantum switch strategy}]. See Appendix \ref{app:qfi all strategies} for the formal definition.

More generally, we consider the quantum superposition of multiple sequential orders, each with a unique order of querying the $N$ channels [see Fig.~\ref{fig:all strategies}\protect\subref{subfig:causal superposition strategy}]. This can be implemented by entangling $N!$ definite causal orders with a quantum control system \cite{Wechs2021PRXQuantum}. If $N=2$ and the control system is traced out, this notion is equivalent to causal separability \cite{Oreshkov2012,Araujo2015}. A causal superposition strategy set $\mathsf{Sup}$ is defined as the collection of $P \in \mathsf{Strat}$ such that
\begin{equation}
    \begin{aligned}
    &\Tr_FP= \sum_\pi q^\pi P^{\pi},\ \sum_{\pi\in S_N} q^\pi=1, \\
    &\,P^{\pi} \in \mathsf{Seq}^{\pi},\ q^\pi \ge 0,\ \pi\in S_N, 
    \end{aligned}
\end{equation}
where each permutation $\pi$ is an element of the symmetric group $S_N$ of degree $N$, and each $\mathsf{Seq}^{\pi}$ denotes a sequential strategy set whose execution order of $N$ channels is $\mathcal E_\phi^{\pi(1)}\rightarrow\mathcal E_\phi^{\pi(2)}\rightarrow\cdots\rightarrow\mathcal E_\phi^{\pi(N)}$, having denoted by $\mathcal E_\phi^k$ the channel from $\mathcal L(\mathcal H_{2k-1})$ to $\mathcal L(\mathcal H_{2k})$. Note that $\mathsf{SWI}$  is a subset of $\mathsf{Sup}$, where the intermediate control is trivial. There are other strategies, such as quantum circuits with quantum controlled casual order (QC-QCs) and probabilistic QC-QCs \cite{Wechs2021PRXQuantum,Purves2021}, which we will not discuss here.

Finally, we introduce the family of general indefinite-causal-order strategy [see Fig.~\ref{fig:all strategies}\protect\subref{subfig:general strategy}], which is the most general strategy set considered in this work. Here the only requirement is that the concatenation of the strategy $P$ with $N$ arbitrary channels results in a legitimate quantum state. The causal relations in this case  \cite{Araujo2015} are a bit cumbersome, but for our purpose what matters is the dual affine space (see Theorem \ref{thm:qfi}), which is simply the space of no-signaling channels \cite{Chiribella2013PRA,Chiribella_2016}. A general indefinite-causal-order strategy set $\mathsf{ICO}$ is defined as the collection of $P \in \mathsf{Strat}$ such that
\begin{equation}
    \rho_F = P*\left(\otimes_{j=1}^N E^j\right),\ \rho_F \ge 0,\ \Tr\rho_F = 1,
\end{equation}
for any $E^j \in \mathcal L(\mathcal H_{2j-1}\otimes\mathcal H_{2j}\otimes\mathcal H_{A_j})$ that denotes the CJ operator of an arbitrary quantum channel with an arbitrary ancillary space $\mathcal H_{A_j}$. 

We note that, unlike the previous strategies that can always be physically realized, the physical realization of the general $\mathsf{ICO}$ is untraceable \cite{Purves2021,Wechs2021PRXQuantum}. The optimal value obtained with general $\mathsf{ICO}$ nevertheless serves as a useful tool that can gauge the performances of different strategies. For example, as we will show, in some cases the optimal QFI $J^{(\mathsf{Sup})}$ and $J^{(\mathsf{ICO})}$ are equal or nearly equal. This then shows that the physically realizable strategy obtained from the set $\mathsf{Sup}$ is already optimal or nearly optimal among all possible strategies, which we will not be able to tell without $J^{(\mathsf{ICO})}$.

\emph{Symmetry reduced programs for optimal metrology}.—The complexity of the original optimization problems in Theorem \ref{thm:qfi} and Algorithm \ref{alg:optimal probe} can be reduced by exploiting the permutation symmetry. In Appendix \ref{app:symmetry optimization}, we prove that we can choose a permutation-invariant matrix $h$ for Theorem \ref{thm:qfi} and solve for a permutation-invariant optimal strategy \cite{invariant_strategy_note} by Algorithm \ref{alg:optimal probe} based on this choice, if any permutation $\pi\in S_N$ bijectively maps each affine space $\mathsf S^i$ [in Eq.~(\ref{eq:condition of thm})] to some affine space $\mathsf S^j$. That is, for any $\pi\in S_N$ and any $i$, there exists a $j$ such that the mapping $S\mapsto G_\pi SG_\pi^\dagger$ on $\mathsf S^i$ is a bijective function from  $\mathsf S^i$ to  $\mathsf S^j$, where $G_\pi$ is a unitary representation of $\pi$. Furthermore, if each space $\mathsf S^i$ itself is permutation invariant, we can restrict each $Q^i\in \overline{\mathsf{S}}^i$ to be permutation invariant, further reducing the complexity of optimization. For both optimization problems we can apply the technique of group-invariant SDP to reduce the size as there exists an isomorphism which preserves positive semidefiniteness, from the permutation-invariant subspace to the space of block-diagonal matrices \cite[Theorem 9.1]{Bachoc2012}. Table \ref{tab:complexity} compares the number of variables involved in QFI evaluation with and without exploiting the symmetry (see Appendix \ref{app:qfi all strategies} for its derivation as well as Appendix \ref{app:complexity} for the complexity of Algorithm \ref{alg:optimal probe}, where by group-invariant SDP we also numerically evaluate the growth of QFI $J^{(\mathsf{ICO})}$ up to $N=5$).

\begin{table}[!htbp]
\caption{\label{tab:complexity} \textbf{Complexity of QFI evaluation for each family of strategies} (with repect to $N$). The asymptotic numbers of variables in optimization are compared between the original (Ori.) and group-invariant (Inv.) SDP. We denote $d:=\dim(\mathcal H_1)\dim(\mathcal H_2)$ and $s:=\max_\phi\rank(E_\phi)\le d$.}
\begin{ruledtabular}
\begin{tabular}{cccccc}
SDP & $\mathsf{Par}$ & $\mathsf{Seq}$ & $\mathsf{SWI}$ & $\mathsf{Sup}$ & $\mathsf{ICO}$ \\
\colrule
Ori. & $O\left(s^N\right)$ & $O\left(d^N\right)$ & $O\left(s^N\right)$ & $O\left(N!\,d^N\right)$ & $O\left(d^N\right)$ \\
Inv. & $O\left(N^{d^2-1}\right)$ & $O\left(d^N\right)$ & $O\left(N^{s^2-1}\right)$ & $O\left(d^N\right)$ & $O\left(N^{d^2-1}\right)$\\
\end{tabular}
\end{ruledtabular}
\end{table}

\emph{Hierarchy of strategies}.—By substituting the definitions of different strategy sets into Theorem \ref{thm:qfi}, we obtain the \emph{exact} values of the optimal QFI. We find that a strict hierarchy of QFI exists quite prevalently. For demonstration purposes, here we show only the result for the amplitude damping channel for $N=2$ and supplement our findings with bountiful numerical results in Appendix \ref{app:supplementary numerical results}. In this case, the process encoding $\phi$ is a $z$ rotation $U_z(\phi) = e^{-\mathrm i\phi t\sigma_z/2}$, where $t$ is the evolution time, followed by an amplitude damping channel described by two Kraus operators: $K_1^{(\mathrm{AD})} = \dyad{0} + \sqrt{1-p} \dyad{1}$ and $K_2^{(\mathrm{AD})} = \sqrt{p}\ketbra{0}{1}$, with the decay parameter $p$. 

\begin{figure} [!bhtp]
    \centering
    \includegraphics[width=\linewidth]{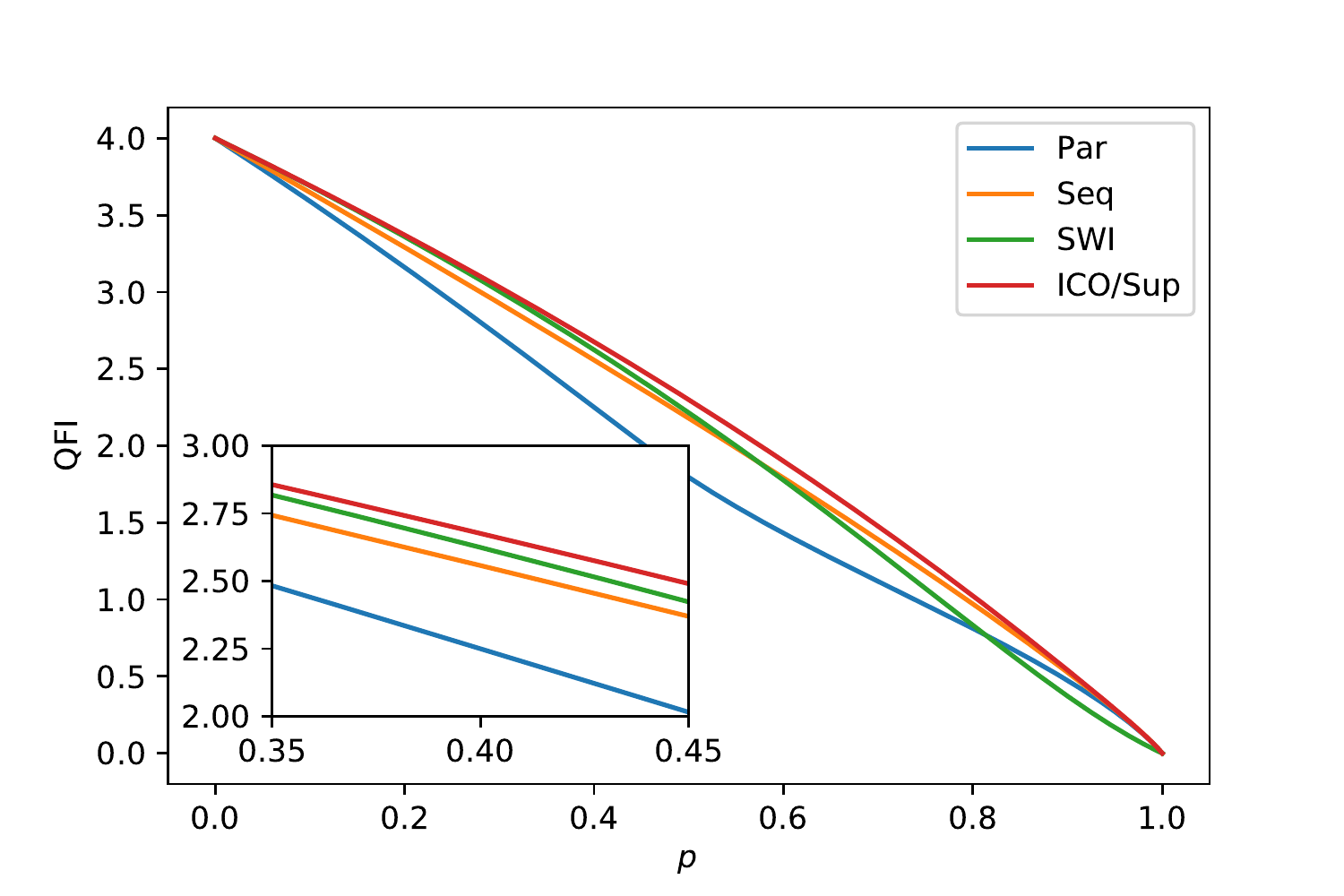}
    \caption{\textbf{Hierarchy of QFI using parallel, sequential, and indefinite-causal-order strategies.} We take $N=2$ and $\phi=1.0$, fix $t=1.0$ and vary the decay parameter $p$. The gaps can be seen more clearly by zooming in on the interval $[0.35,0.45]$ of the value of $p$. $J^{(\mathsf{Sup})}=J^{(\mathsf{ICO})}$ with an error tolerance of no more than $10^{-8}$ in this case, but the gap between $J^{(\mathsf{Sup})}$ and $J^{(\mathsf{ICO})}$ could exist  for larger $N$ or for other types of noise, which can be observed by randomly sampling noise channels.}
    \label{fig:qfi all strategies}
\end{figure}

In Fig.~\ref{fig:qfi all strategies} we plot the QFI versus $p$ for the amplitude damping noise with all 5 strategy sets for $N=2$.  A strict hierarchy of $\mathsf{Par}$, $\mathsf{Seq}$ and $\mathsf{ICO}$ holds if $p$ is neither 1 nor 0, i.e., $J^{(\mathsf{Par})}<J^{(\mathsf{Seq})}<J^{(\mathsf{ICO})}$. This is in contrast to the asymptotic regime of $N\to\infty$, where the relative difference between $J^{(\mathsf{Seq})}$ and $J^{(\mathsf{Par})}$ vanishes for this channel \cite{PRXQuantum.2.010343}. Besides, in this case general $\mathsf{ICO}$ cannot strictly outperform $\mathsf{Sup}$, implying that causally superposing two sequential strategies is sufficient to achieve the general optimality in this particular scenario. The gap between $J^{(\mathsf{Sup})}$ and $J^{(\mathsf{ICO})}$, however, could be observed for the same channel with larger $N$ or for other channels when $N=2$ (see Appendix \ref{app:supplementary numerical results}). In fact, by randomly sampling noise channels from CPTP channel ensembles, we find that for 984 of 1000 random channels, a strict hierarchy $J^{(\mathsf{Par})}<J^{(\mathsf{Seq})}<J^{(\mathsf{Sup})}<J^{(\mathsf{ICO})}$ holds for $N=2$, implying that there exist more powerful strategies than causal superposition strategies in these cases. We note that a strict hierarchy of strategies has been found for channel discrimination in  \cite{Bavaresco2021PRL}, but much less is known in quantum metrology until our work.

Our method can also test the tightness of existing QFI bounds in the nonasymptotic regime, which has seldom been done until this work. Here we take the commonly used, asymptotically tight \cite{PRXQuantum.2.010343} upper bound for parallel strategies [see \cite[Theorem 4 and Eq.~(17)]{Fujiwara2008} or \cite[Eq.~(16)]{Demkowicz-Dobrzanski2012}]. For $p=0.5$, our result shows that the exact parallel QFI $J^{(\mathsf{Par})} = 1.795$ is $32.7\%$ lower than the asymptotically tight parallel upper bound $2.667$, and even the exact sequential QFI $J^{(\mathsf{Seq})} = 2.179$ is $18.3\%$ lower than this parallel upper bound \cite{seq_bound_note}. Similar phenomena are observed in other noise models and for different $N$ (see Appendix \ref{app:compar asymp}). 

With Algorithm \ref{alg:optimal probe} we can also construct strategies to achieve the optimal QFI. Remarkably, we find that a simple strategy of applying a quantum SWITCH using a control qubit $\ket{\plus}_C:=\left(\ket{0}_C+\ket{1}_C\right)/\sqrt{2}$ (without any additional control operations on the probe) beats any sequential strategies (which can involve complex control) in certain cases (e.g. $p<0.5$). To our best knowledge, this is the only instance of noisy quantum metrology so far, where the advantage of indefinite causal orders is established rigorously. In Appendix \ref{app:implementation of strategies} we also present two explicit examples of implementing optimal sequential and causal superposition strategies, obtained by first applying Algorithm \ref{alg:optimal probe} and converting the CJ operators into quantum circuit consisting of single-qubit rotations and CNOT gates (as well as a quantum SWITCH for the case of $\mathsf{Sup}$). For optimal causal superposition strategies, the permutation symmetry allows us to only control the execution order of channels while fixing state preparation and intermediate control [$\rho_\uparrow = \rho_\downarrow$ and $U_\uparrow = U_\downarrow$ in Fig.~\ref{fig:all strategies}\protect\subref{subfig:causal superposition strategy}], which can be implemented by a $(2N-1)$-quantum SWITCH of $N$ channels $\mathcal E_\phi$ and $N-1$ intermediate operations.

Our result serves as a versatile tool for the demonstration of optimal quantum metrology and the design of optimal quantum sensors, especially in the context of control optimization \cite{Hou2019PRL,Hou2021PRL} and indefinite causal orders \cite{Oreshkov2012,Chiribella2013PRA,Araujo2015,mukhopadhyay2018superposition,Quintino2019PRL,Zhao2020PRL,Chapeau-Blondeau2021PRA,Bavaresco2021PRL}.

The code accompanying the paper is openly available \cite{Code_note}.

We thank Cyril Branciard, Alastair A. Abbott, and Raphael Mothe for helpful discussions and comments on our first manuscript. This work is supported  by Guangdong Natural Science Fund—General Programme via Project 2022A1515010340, by HKU Seed Fund for Basic Research for New Staff via Project 202107185045, and by the Research Grants Council of Hong Kong through the Grant No. 14307420.

\emph{Note added}.—Recently, it has been shown that neither sequential nor causal superposition strategies provide any advantage over parallel strategies asymptotically \cite{Kurdzialek22}.

%\nocite{*}

%\bibliography{reference}% Produces the bibliography via BibTeX.

%\clearpage
\onecolumngrid

\appendix

\section{Proof of Eq.~(\ref{eq:qfi_primal}) of the main text} \label{app:lemma qfi}
The formalism of this proof has been developed for strategies following a definite causal order in \cite{Altherr2021PRL}, and the generalization to indefinite-causal-order strategies considered here is straightforward.

In \cite{Fujiwara2008}, the QFI of a quantum state is expressed as a minimization problem
\begin{equation} 
    J_Q(\rho_\phi) = 4\min_{\ket{\psi_{\phi,i}}}\sum_{i=1}^q\Tr\left(\dyad{\dot{\psi}_{\phi,i}}\right),
\end{equation}
for any integer $q \ge \max_\phi \rank(\rho_\phi)$, where $\left\{\ket{\psi_{\phi,i}}\right\}$ is a set of unnormalized vectors such that $\rho_\phi = \sum_i\dyad{\psi_{\phi,i}}$\footnote{We assume $\left\{\ket{\psi_{\phi,i}}\right\}$ is continuously differentiable with respect to $\phi$.}. In the main text, the QFI of $N$ quantum channels $\mathcal E_\phi$ is defined as the QFI of the output state obtained from the concatenation of the CJ operator $N_\phi$ of $N$ quantum channels and an optimal strategy $P$ in a given strategy set $\mathsf P$:
\begin{equation} 
    J^{(\mathsf P)}(N_\phi) := \max_{P \in \mathsf P} J_Q(P * N_\phi).
\end{equation}
Due to the monotonicity \cite{petz2008quantum} of the QFI under CPTP maps (e.g., the partial trace operation over ancillary space in this case), by choosing a proper global future space $\mathcal H_F$ an optimal $P$ can be taken as a pure process (rank-1 operator) denoted by $\dyad{P}$. For a fixed $P=\dyad{P}$, minimization over decompositions of $P*N_\phi$ is equivalent to minimization over decompositions of $N_\phi$, as $\rank(P*N_\phi) \le \rank(N_\phi)$.

As a positive semidefinite operator, $N_\phi$ has a decomposition as:
\begin{equation}
    N_\phi = \sum_{i=1}^{r}\dyad{N_{\phi,i}},
\end{equation}
where $r := \max_\phi\rank(N_\phi)$\footnote{We also assume $\left\{\ket{N_{\phi,i}}\right\}$ is continuously differentiable with respect to $\phi$.}. Note that the decomposition is nonunique. Defining $\mathbf N_\phi := \left(\ket{N_{\phi,1}},\dots,\ket{N_{\phi,r}}\right)$, an arbitrary alternative decomposition $\tilde{\mathbf N}_\phi := \left(\ket{\tilde N_{\phi,1}},\dots,\ket{\tilde N_{\phi,r}}\right)$ can be related to $\mathbf N_\phi$ by $\tilde{\mathbf N}_\phi = \mathbf N_\phi V_\phi$, where $V_\phi$ is an $r\times r$ unitary matrix. Then the QFI of $N$ channels can be expressed as
\begin{equation}
    J^{(\mathsf P)}(N_\phi) = \max_{P \in \mathsf P} \min_{V_\phi}\Tr\left[P\left(\mathds{1}_F\otimes \Omega_\phi\right)\right],
\end{equation}
having defined the performance operator $\Omega_\phi := 4\left(\dot{\tilde{\mathbf N}}_\phi\dot{\tilde{\mathbf N}}_\phi^\dagger\right)^T = 4\sum_{i=1}^r \left(\dyad{\dot{\tilde{N}}_{\phi,i}} \right)^T$, where the superscript $T$ denotes the transpose. We can further define an $r\times r$ Hermitian matrix $h:=\mathrm i \dot V_\phi V_\phi^{\dagger}$ to take care of the freedom of choice for the decomposition $\tilde{\mathbf N}_\phi$, by noting that
$\dot{\tilde{\mathbf N}}_\phi \dot{\tilde{\mathbf N}}_\phi^\dagger = \left(\dot{\mathbf N}_\phi-\mathrm i \mathbf N_\phi h\right)\left(\dot{\mathbf N}_\phi-\mathrm i \mathbf N_\phi h\right)^\dagger$. Now by redefining $\dot{\tilde{\mathbf N}}_\phi := \left(\dot{\mathbf N}_\phi-\mathrm i \mathbf N_\phi h\right)$, we rewrite $\Omega_\phi$ as $\Omega_\phi(h)$ to explicitly manifest its dependence on $h$. Hence, we finally arrive at Eq.~(\ref{eq:qfi_primal}) of the main text:
\begin{equation}
    \begin{aligned}
    J^{(\mathsf P)}(N_\phi) &= \max_{P \in \mathsf P} \min_{h \in \mathbb{H}_r} \Tr\left\{P\left[\mathds{1}_F\otimes\Omega_\phi(h)\right]\right\} \\
    &= \max_{\tilde{P} \in \tilde{\mathsf P}} \min_{h \in \mathbb H_r} \Tr\left[\tilde{P}\Omega_\phi(h)\right],
    \end{aligned}
\end{equation}
where $\tilde{\mathsf P} := \left\{\tilde{P} = \Tr_FP \middle|P \in \mathsf P\right\}$. \qed

\section{Proof of Theorem \ref{thm:qfi}} \label{app:thm qfi}
Starting from Eq.~(\ref{eq:qfi_primal}) of the main text, we exchange the order of minimization and maximization thanks to Fan's minimax theorem \cite{Fan1953}, since $\Tr \left[\tilde P \Omega_\phi(h) \right]$ is concave on $\tilde P$ and convex on $h$, and $\tilde{\mathsf{P}}$ is assumed to be a compact set. Then the problem of QFI evaluation can be rewritten as 
\begin{equation} \label{app:eq min max}
    J^{(\mathsf P)}(N_\phi) = \min_h\max_{\tilde P \in \tilde{\mathsf P}}\Tr\left[\tilde P\Omega_\phi(h)\right].
\end{equation}
Reformulating the condition of Theorem \ref{thm:qfi}, we require that each operator $\tilde P\in\tilde{\mathsf P}$ can be written as a convex combination of positive semidefinite operators $S^i$, $i=1,\dots,K$:
\begin{equation}
    \tilde P = \sum_{i=1}^Kq^iS^i,\ \mathrm{for}\ \sum_{i=1}^Kq^i=1,\ q^i\ge0,\ S^i\ge0,\ S^i\in\mathsf S^i,\ i=1,\dots,K,
\end{equation}
where each $\mathsf S^i$ is an affine space of Hermitian operators. Thus Eq. (\ref{app:eq min max}) can be reformulated as
\begin{equation}
    \begin{aligned}
    J^{(\mathsf P)}(N_\phi) =& \min_h\max_{\tilde P}\Tr\left[\tilde P\Omega_\phi(h)\right], \\
    \mathrm{s.t.}\ &\,\tilde P=\sum_{i=1}^Kq^iS^i, \\
    &\sum_{i=1}^Kq^i=1, \\
    &\,q^i\ge0,\ S^i\ge0,\ S^i\in\mathsf S^i,\ i=1,\dots,K.
    \end{aligned}
\end{equation}

For now we fix $h$ and consider the dual problem of maximization over $\tilde{\mathsf P}$.  For each affine space $\mathsf S^i$ we have defined its dual affine space $\overline{\mathsf S}^i$, whose dual affine space in turn is exactly $\mathsf S^i$ \cite{Chiribella_2016}. Choose an affine basis $\{Q^{i,j}\}_{j=1}^{L_i}$ for $\overline{\mathsf S}^i$, and the maximization problem is further expressed as
\begin{equation}
    \begin{aligned}
    \max_{\tilde P}&\Tr\left[\tilde P\Omega_\phi(h)\right], \\
    \mathrm{s.t.}\ &\,\tilde P=\sum_{i=1}^Kq^iS^i, \\
    &\sum_{i=1}^Kq^i=1, \\
    &\,q^i\ge0,\ S^i\ge0,\ \Tr\left(S^iQ^{i,j}\right)=1,\ i=1,\dots,K,\ j=1,\dots,L_i.
    \end{aligned}
\end{equation}
Defining $P^i:=q^iS^i$ to avoid the product of variables in optimization, we have
\begin{equation}
    \begin{aligned}
    \max&\Tr\left[\tilde P\Omega_\phi(h)\right], \\
    \mathrm{s.t.}\ &\,\tilde P=\sum_{i=1}^KP^i, \\
    &\sum_{i=1}^Kq^i=1, \\
    &\,P^i\ge0,\ \Tr\left(P^iQ^{i,j}\right)=q^i,\ i=1,\dots,K,\ j=1,\dots,L_i,
    \end{aligned}
\end{equation}
where the constraints $q^i\ge0$ can be safely removed, since $\Tr S^i=\prod_{j=1}^{N}d_{2j}$, implying that $\overline{\mathsf S}^i$ includes a positive operator proportional to identity for any $i=1,\dots,K$, having denoted $d_j := \dim(\mathcal H_j)$ for simplicity. The Lagrangian of the problem is given by 
\begin{equation}
    \begin{aligned}
    L &= \sum_i\Tr\left[P^i\Omega_\phi(h)\right] + \left(1-\sum_iq^i\right)\lambda + \sum_i\Tr\left(P^i\tilde Q^i\right) + \sum_{i,j}\left[q^i-\Tr\left(P^iQ^{i,j}\right)\right]\lambda^{i,j} \\
    &= \lambda + \sum_i\Tr\left\{P^i\left[\Omega_\phi(h)+\tilde Q^i-\sum_j\lambda^{i,j}Q^{i,j}\right]\right\} + \sum_i\left[q^i\left(\sum_j\lambda^{i,j}-\lambda\right)\right],
    \end{aligned}
\end{equation}
for $\tilde Q^i\ge0$. Hence, by removing $\tilde Q^i$ the dual problem is written as
\begin{equation}
    \begin{aligned}
    \min&\,\lambda, \\
    \mathrm{s.t.}\ &\sum_j\lambda^{i,j}Q^{i,j}\ge\Omega_\phi(h),\ \lambda=\sum_{j}\lambda^{i,j},\ i=1,\dots,K,\ j=1,\dots,L_i.
    \end{aligned}
\end{equation}
We define $Q^i := \sum_j\lambda^{i,j}Q^{i,j}/\lambda$ if $\lambda\neq0$ ($\lambda=0$ corresponds to a trivial case where the QFI is zero), and clearly $Q^i$ is an arbitrary operator in the set $\overline{\mathsf S}^i$. Therefore, we cast the dual problem into
\begin{equation}
    \begin{aligned}
    \min&\,\lambda, \\
    \mathrm{s.t.}\ &\,\lambda Q^i \ge \Omega_\phi(h),\ Q^i\in \overline{\mathsf S}^i,\ i=1,\dots,K.
    \end{aligned}
\end{equation}
Slater's theorem \cite{watrous2018theory} implies that the strong duality holds, since the QFI is finite and the inequality constraints can be strictly satisfied for a positive semidefinite operator $\Omega_\phi(h)$, by choosing $\lambda Q^i=\mu\lVert\Omega_\phi(h)\rVert\mathds{1}_{1,2,\dots,2N}$ for $\mu>1$ and any $i=1,\dots,K$, having denoted the operator norm by $\lVert \cdot \rVert$. Finally, by optimizing the choice of $h$ we derive the result of Theorem \ref{thm:qfi}. \qed

\section{Proof of the validity of Algorithm \ref{alg:optimal probe}} \label{app:alg optimal probe}
We first recall the minimax theorem:
\begin{equation} \label{eq:minimax}
   \min_x \max_y f(x,y) = \max_y \min_x f(x,y) 
\end{equation}
for a function $f(x,y)$ convex on $x$ and concave on $y$. Assume $(x_0, y_1)$ is a solution for the L.H.S. of Eq. (\ref{eq:minimax}) and $(x_1,y_0)$ is a solution for the R.H.S. of Eq. (\ref{eq:minimax}). It is easy to see that
\begin{equation}
    f(x_0,y_1) \ge f(x_0,y_0) \ge f(x_1,y_0).
\end{equation}
In view of Eq. (\ref{eq:minimax}) both equalities hold. Therefore, $(x_0,y_0)$ is a saddle point of $f(x,y)$, i.e., $x_0=\argmin_x{f(x,y_0)}$ and $y_0 = \argmax_y{f(x_0,y)}$. Since the objective function $\Tr\left[\tilde P\Omega_\phi(h)\right]$ in the primal problem of estimating QFI is convex on $h$ and concave on $\tilde P$, we can substitute $x=h$ and $y=\tilde P$. Obviously, $h^{(\mathrm{opt})}$ is an optimal solution for $\min_h \max_{\tilde P} \Tr\left[\tilde P\Omega_\phi(h)\right]$ and thus corresponds to a saddle point. Then $x_0=\argmin_x{f(x,y_0)}$ can be satisfied by requiring $\partial_h f\left(h,\tilde P^{(\mathrm{opt})}\right)|_{h=h^{(\mathrm{opt})}} = 0$, resulting in Eq. (\ref{eq:optimal condition}) in the main text:
\begin{equation} \label{app:eq:optimal strategy constraint}
    \Re \left\{\Tr \left\{\tilde P^{(\mathrm{opt})}\left[-\mathrm i \mathbf N_\phi \mathscr H\left(\dot{\mathbf N}_\phi-\mathrm i \mathbf N_\phi h^{(\mathrm{opt})}\right)^{\dagger}\right]^T\right\} \right\}
        =0\ \mathrm{for}\ \mathrm{all}\ \mathscr H \in \mathbb H_r.
\end{equation}
Meanwhile $y_0 = \argmax_y{f(x_0,y)}$ corresponds to the requirement that $\tilde P^{(\mathrm{opt})}$ is an optimal solution for fixed $h=h^{(\mathrm{opt})}$. Therefore, $\left(h^{(\mathrm{opt})},\tilde P^{(\mathrm{opt})}\right)$ is a saddle point and an optimal solution for $\max_{\tilde P} \min_h \Tr\left[\tilde P\Omega_\phi(h)\right]$. By definition a purification of $\tilde P^{(\mathrm{opt})}$ is an optimal physically implemented strategy, i.e., we can choose a strategy $P^{(\mathrm{opt})}$ such that $\Tr_F P^{(\mathrm{opt})} = \tilde P^{(\mathrm{opt})}$. \qed 

We remark that Eq. (\ref{app:eq:optimal strategy constraint}) can be reformulated as a set of linear constraints by choosing a basis $\{\mathscr H^i\}_{i=1}^{r^2}$ for the space $\mathbb H_r$ of Hemitian matrices. For example, denoting by $E_{ij}$ the $r\times r$ matrix of which only the $(i,j)$-th element is $1$ and all other elements are $0$, we can choose $\mathscr H =E_{jj}$ for $j=1,\dots,r$, $\mathscr H =E_{jk}+E_{kj}$ and $\mathscr H =\mathrm i\left(E_{jk}-E_{kj}\right)$ for $k=1,\dots,r$ and $k\neq j$, and obtain a series of constraints, which are equivalent to Eq.(\ref{app:eq:optimal strategy constraint}).

\section{Symmetry reduced optimization} \label{app:symmetry optimization}
This section demonstrates how to reduce the size of the optimization problems concerned in Theorem \ref{thm:qfi} and Algorithm \ref{alg:optimal probe}, by exploiting the permutation symmetry as applicable. 

Let us begin with some notations. We consider the action of $S_N$, the symmetric group of degree $N$, on the finite-dimensional representation space $\otimes_{i=1}^{N}\mathcal W_i$. $G_\pi$ is a unitary (and orthogonal) operator on $\otimes_{i=1}^{N}\mathcal W_i$ corresponding to the permutation $\pi \in S_N$: $G_\pi=\sum_{\mathbf{i}=(i_1,\dots,i_N)}\left(\otimes_{j=1}^N\ket{i_{\pi(j)}}\right)\otimes_{k=1}^N \bra{i_k}$, where $\{\ket{i_j}\}_i$ denotes an orthonormal basis of $\mathcal W_j$. Then an operator $X$ on $\otimes_{i=1}^{N}\mathcal W_i$ is said to be permutation invariant iff $G_\pi X G_\pi^\dagger=X$ for all $\pi \in S_N$. Analogously, a space $\mathcal X$ is permutation invariant iff $G_\pi X G_\pi^\dagger \in \mathcal X$ for any $X\in\mathcal X$ and any $\pi\in S_N$.   

We now explicitly express the components of the performance operator $\Omega_\phi(h)$. Given $N$ identical quantum channels $\mathcal E_\phi(\rho) = \sum_i K_{\phi,i}^{\dagger} \rho K_{\phi,i}$, we decompose the CJ operator $E_\phi = \sum_{i}\dyad{E_{\phi,i}}$ corresponding to each channel. Note that $\ket{E_{\phi,i}}$ is the vectorization of the Kraus operator $K_{\phi,i}$, i.e., $\ket{E_{\phi,i}} = \sum_{m,n} \matrixel{m}{K_{\phi,i}}{n}\ket{m}\ket{n}$. The CJ operator of $N$ identical quantum channels is $N_\phi = E_\phi^{\otimes N} = \sum_{\mathbf i} \dyad{N_{\phi,\mathbf i}}$, where we use the notation $\ket{N_{\phi, \mathbf i=(i_1,\dots,i_N)}} = \otimes_{n=1}^N \ket{E _{\phi,i_n}}$. Taking the derivative results in (we omit the subscript $\phi$ in $\ket{E_{\phi,i}}$) $\ket{\dot N_{\phi, \mathbf i}} = \sum_{j=1}^N\ket{E_{i_1}} \cdots \ket{E_{i_{j-1}}} \ket{\dot E_{i_j}} \ket{E_{i_{j+1}}} \cdots \ket{E_{i_N}}$, from which we can obtain the performance operator $\Omega_\phi(h) = 4\left(\dot{\tilde{\mathbf N}}_\phi\dot{\tilde{\mathbf N}}_\phi^\dagger\right)^T = 4\sum_{\mathbf i} \left(\dyad{\dot{\tilde{N}}_{\phi,\mathbf i}} \right)^T$, where $\ket{\dot{\tilde{N}}_{\phi,\mathbf j}} = \ket{\dot{N}_{\phi,\mathbf j}}-\mathrm i \sum_{\mathbf k}\ket{N_{\phi,\mathbf k}}h_{\mathbf k \mathbf j}$.

\subsection{Symmetry reduced QFI evaluation}
With these notations, we first consider the optimization in QFI evaluation and have the following:
\begin{lemma} \label{app:lem:qfi h symmetry}
%In the optimization problem of Theorem \ref{thm:qfi}, if the union of dual affine spaces $\bigcup_i \overline{\mathsf {S}}^i$ is permutation-invariant, then there must exist a permutation-invariant $h$ as a feasible solution.
In the optimization problem of Theorem \ref{thm:qfi}, if, for any $\pi\in S_N$ and any $i$, there exists a $j$ such that the mapping $S\mapsto G_\pi SG_\pi^\dagger$ on $\mathsf S^i$ is a bijective function from $\mathsf S^i$ to $\mathsf S^j$, then there must exist a permutation-invariant $h$ as a feasible solution.
%any set of $K$ operators $\{Q^i\}_{i=1}^K$ for $Q^i \in \overline{\mathsf {S}}^i$ is mapped to a set $\{\tilde Q^i\}_{i=1}^K$ such that $\tilde Q^i \in \overline{\mathsf {S}}^i$ under any permutation operation $Q^i\rightarrow G_\pi Q^i G_\pi^\dagger$, then there must exist a permutation-invariant $h$ as a feasible solution.
\end{lemma}
\begin{proof}
%if for any $i\neq j$, any $Q^i \in \overline{\mathsf {S}}^i$, any $Q^j \in \overline{\mathsf {S}}^j$ and any $\pi \in S_N$, we have $G_\pi Q^i G_\pi^\dagger \in \overline{\mathsf {S}}^k$ and $G_\pi Q^j G_\pi^\dagger \in \overline{\mathsf {S}}^l$ for $k\neq l$,
Without loss of generality we assume all $\mathsf S^i$ are distinct spaces; otherwise, we just remove the duplicate ones. We first prove that, under any permutation operation $Q\mapsto G_\pi^\dagger Q G_\pi$, for any $i\in\{1,\dots,K\}$ there exists a unique $j\in\{1,\dots,K\}$ such that the dual affine space $\overline{\mathsf S}^j$ is bijectively mapped to $\overline{\mathsf S}^i$. The condition of the lemma implies that, for any $G_\pi$ and any $i$ we can find $j_i$ such that $G_\pi S^i G_\pi^\dagger \in \mathsf S^{j_i}$ for all $S^i\in \mathsf S^i$. Due to the bijectivity, $\mathsf S^i$ and the corresponding $\mathsf S^{j_i}$ are isomorphic, with different $j_i$ for different $i$. Apparently $\overline{\mathsf S}^i$ and $\overline{\mathsf S}^{j_i}$ are also isomorphic. If we choose any $Q^{j_i} \in \overline{\mathsf S}^{j_i}$, i.e., $\Tr(Q^{j_i}S^{j_i}) = 1$ for all $S^{j_i}\in \mathsf S^{j_i}$, then we have $\Tr(G_\pi^\dagger Q^{j_i}G_\pi S^i)=\Tr(Q^{j_i}G_\pi S^i G_\pi^\dagger)=1$ for all $S^i\in \mathsf S^i$. Therefore, $G_\pi^\dagger Q^{j_i} G_\pi \in \overline{\mathsf S}^i$ for all $Q^{j_i} \in \overline{\mathsf S}^{j_i}$. Furthermore, the permutation operation $Q\mapsto G_\pi^\dagger Q G_\pi$ from $\overline{\mathsf S}^{j_i}$ to $\overline{\mathsf S}^i$ is bijective as $\overline{\mathsf S}^{j_i}$ and $\overline{\mathsf S}^i$ are isomorphic. In particular, any set of $K$ operators $\{Q^i\}_{i=1}^K$ for $Q^i \in \overline{\mathsf {S}}^i$ is mapped to a set $\{\tilde Q^i\}_{i=1}^K$ such that $\tilde Q^i \in \overline{\mathsf {S}}^i$.

We then prove that there exists a permutation-invariant performance operator $\Omega_\phi(h)$ as a feasible solution. The group action is characterized by the permutation operator $G_\pi$ on the representation space $\otimes_i\mathcal W_i = \otimes_i\left(\mathcal H_{2i-1} \otimes \mathcal H_{2i}\right)$. Suppose $h^{(\mathrm{opt})}$ is an optimal solution, i.e., for any $i\in\{1,\dots,K\}$, there exist optimal values of $\lambda$ and $Q^i$ such that $\lambda Q^i \ge \Omega_\phi\left(h^{(\mathrm{opt})}\right)$ for $Q^i\in \overline{\mathsf S}^i$. Then for any permutation $\pi\in S_N$ we have
\begin{equation} \label{app:eq:dual constraint permutation}
    \lambda G_\pi Q^i G_\pi^\dagger \ge G_\pi \Omega_\phi\left(h^{(\mathrm{opt})}\right) G_\pi^\dagger,\ Q^i\in \overline{\mathsf S}^i,\ i=1,\dots,K.
\end{equation} 
Since $G_\pi$ maps $\{Q^i\}_i$ to a set $\{\tilde Q^i\}_i$ such that $\tilde Q^i \in \overline{\mathsf {S}}^i$, the constraint Eq. (\ref{app:eq:dual constraint permutation}) becomes
\begin{equation}
    \lambda \tilde Q^i \ge G_\pi \Omega_\phi\left(h^{(\mathrm{opt})}\right) G_\pi^\dagger,\ \tilde Q^i\in \overline{\mathsf S}^i,\ i=1,\dots,K,
\end{equation}
which implies that $G_\pi \Omega_\phi(h^{(\mathrm{opt})}) G_\pi^\dagger$ is also a feasible solution for any $\pi$. Furthermore, the permutation-invariant solution of the performance operator $\frac{1}{N!}\sum_{\pi\in S_N} G_\pi \Omega_\phi(h^{(\mathrm{opt})}) G_\pi^\dagger$ is feasible.

Next, we show that we can choose a permutation-invariant $h$ such that the performance operator $\Omega_\phi(h) = \frac{1}{N!}\sum_{\pi\in S_N} G_\pi \Omega_\phi\left(h^{(\mathrm{opt})}\right) G_\pi^\dagger$. $h$ is an operator on $\otimes_{i=1}^N \mathbb C^s$, where $s:=\max_\phi \rank(E_\phi)\le d$. For distinction we denote the group action on $h$ by $G_\pi' h G_\pi^{\prime\dagger}$ and the group action on $\Omega_\phi(h)$ by $G_\pi \Omega_\phi(h) G_\pi^\dagger$.
Writing $h$ in components, we have 
\begin{equation} \label{app:eq:h component transform}
    \begin{aligned}
    \sum_{\mathbf i,\mathbf j} G_\pi' h_{\mathbf i\mathbf j}\ketbra{\mathbf i}{\mathbf j} G_\pi^{\prime\dagger} &= \sum_{\mathbf i,\mathbf j} h_{\mathbf i\mathbf j}\ketbra{\pi(\mathbf i)}{\pi(\mathbf j)} \\
    &= \sum_{\pi^{-1}(\mathbf i),\pi^{-1}(\mathbf j)} h_{\pi^{-1}(\mathbf i) \pi^{-1}(\mathbf j)}\ketbra{\mathbf i}{\mathbf j} \\
    &= \sum_{\mathbf i,\mathbf j} h_{\pi^{-1}(\mathbf i) \pi^{-1}(\mathbf j)}\ketbra{\mathbf i}{\mathbf j},
    \end{aligned}
\end{equation}
where $\pi(\mathbf i):=(i_{\pi(1)},\dots,i_{\pi(N)})$. Note that $G_\pi \ket{N_{\phi,\mathbf i}} = \ket{N_{\phi,\pi(\mathbf i)}}$ and
\begin{equation} 
    \begin{aligned}
    G_\pi \ket{\dot{N}_{\phi,\mathbf i}} &= \sum_{j=1}^N\ket{E_{i_1}} \cdots \ket{E_{i_{j-1}}} \ket{\dot E_{i_j}} \ket{E_{i_{j+1}}} \cdots \ket{E_{i_N}} \\
    &= \sum_{j=1}^N \ket{E_{i_{\pi(1)}}} \cdots \ket{E_{i_{\pi[\pi^{-1}(j)-1]}}} \ket{\dot E_{i_j}} \ket{E_{i_{\pi[\pi^{-1}(j)+1]}}} \cdots \ket{E_{i_{\pi(N)}}} \\
    &= \sum_{j=1}^N \ket{E_{i_{\pi(1)}}} \cdots \ket{E_{i_{\pi(j-1)}}} \ket{\dot E_{i_{\pi(j)}}} \ket{E_{i_{\pi(j+1)}}} \cdots \ket{E_{i_{\pi(N)}}} \\
    &= \ket{\dot{N}_{\phi,\pi(\mathbf i)}},
    \end{aligned}
\end{equation}
which results in
\begin{equation} \label{app:eq:group action on performance operator and h}
    \begin{aligned}
    G_\pi \Omega_\phi(h) G_\pi^\dagger &= 4 \sum_{\mathbf j} G_\pi \left(\dyad{\dot{\tilde{N}}_{\phi,\mathbf j}} \right)^T G_\pi^\dagger \\ 
    &= 4 \sum_{\mathbf j} G_\pi \dyad{\dot{\tilde{N}}_{\phi,\mathbf j}^*} G_\pi^\dagger \\
    &= 4 \sum_{\mathbf j} G_\pi \left(\ket{\dot N_{\phi,\mathbf j}^*} + \mathrm i \sum_{\mathbf k}\ket{N_{\phi,\mathbf k}^*}h_{\mathbf k \mathbf j}^*\right) \left(\bra{\dot N_{\phi,\mathbf j}^*} - \mathrm i \sum_{\mathbf l}h_{\mathbf l \mathbf j} \bra{N_{\phi,\mathbf l}^*}\right) G_\pi^\dagger \\
    &= 4 \sum_{\mathbf j} \left(\ket{\dot N_{\phi,\pi(\mathbf j)}^*} + \mathrm i \sum_{\mathbf k}\ket{N_{\phi,\pi(\mathbf k)}^*}h_{\mathbf k \mathbf j}^*\right) \left(\bra{\dot N_{\phi,\pi(\mathbf j)}^*} - \mathrm i \sum_{\mathbf l}h_{\mathbf l \mathbf j} \bra{N_{\phi,\pi(\mathbf l)}^*}\right) \\
    &= 4 \sum_{\mathbf j} \left(\ket{\dot N_{\phi,\mathbf j}^*} + \mathrm i \sum_{\mathbf k}\ket{N_{\phi,\mathbf k}^*}h_{\pi^{-1}(\mathbf k) \pi^{-1}(\mathbf j)}^*\right) \left(\bra{\dot N_{\phi,\mathbf j}^*} - \mathrm i \sum_{\mathbf l}h_{\pi^{-1}(\mathbf l) \pi^{-1}(\mathbf j)} \bra{N_{\phi,\mathbf l}^*}\right) \\
    &= \Omega_\phi\left(G_\pi' h G_\pi^{\prime\dagger}\right),
    \end{aligned}
\end{equation}
where in the last equation we have used Eq. (\ref{app:eq:h component transform}). Therefore, by choosing the permutation-invariant solution $h=\frac{1}{N!}\sum_{\pi\in S_N} G_\pi' h^{(\mathrm{opt})} G_\pi^{\prime\dagger}$ we obtain $\Omega_\phi(h) = \frac{1}{N!}\sum_{\pi\in S_N} G_\pi \Omega_\phi(h^{(\mathrm{opt})}) G_\pi^\dagger$, and this permutation-invariant choice of $h$ is also a feasible solution.
\end{proof}
If a stronger condition holds, then not only can we choose a permutation-invariant $h$, but we can also restrict $Q^i\in\overline{\mathsf {S}}^i$ to be permutation invariant. In this case all matrix variables concerned in the optimization are permutation invariant.
\begin{lemma} \label{app:lem:qfi Q symmetry}
In the optimization problem of Theorem \ref{thm:qfi}, if each affine space $\mathsf S^i$ is permutation invariant for any $i=1,\dots,K$, then there must exist a permutation-invariant $Q^i\in\overline{\mathsf {S}}^i$ as a feasible solution for each $i$.
\end{lemma}
\begin{proof}
By Lemma \ref{app:lem:qfi h symmetry} we can restrict $h$ to be permutation invariant in optimization, and therefore the performance operator $\Omega_\phi(h)$ is permutation invariant. It is easy to see that each dual affine space $\overline{\mathsf S}^i$ is also permutation invariant. Then for any $i=1,\dots,K$, for $Q^{i,(\mathrm{opt})}\in \overline{\mathsf {S}}^i$ satisfying the constraint $\lambda Q^i \ge \Omega_\phi(h)$, we have the permutation-invariant solution $Q^i = \frac{1}{N!}\sum_{\pi\in S_N} G_\pi Q^{i,(\mathrm{opt})} G_\pi^\dagger \in \overline{\mathsf {S}}^i$ which also satisfies the same constraint. Hence, this permutation-invariant choice of $Q^i$ is also a feasible solution.
\end{proof}

Now since the optimization is restricted to the invariant subspace, we can reduce the matrix sizes by \emph{block diagonalization}. Generally, for a group-invariant space of complex matrices $\mathbb C^{n\times n,(\mathrm{inv})}$, there exists an isomorphism $\varphi$ preserving positive semidefiniteness between $\mathbb C^{n\times n,(\mathrm{inv})}$ and a direct sum of complex matrix spaces \cite[Theorem 9.1]{Bachoc2012}:
\begin{equation}
    \varphi:\mathbb C^{n\times n,(\mathrm{inv})} \rightarrow \bigoplus_{k=1}^I \mathbb C^{m_k\times m_k},
\end{equation}
where $m_k$ is the multiplicity of the $k$-th inequivalent irreducible representation of the group, and $I$ is the number of inequivalent irreducible representations. To be more specific, for the symmetric group $S_N$, the representation space $\mathcal W=(\mathbb C^d)^{\otimes N}$ can be decomposed into $\bigoplus_{|\mu|=N} \mathcal W^\mu$, where each partition $\mu:=(\mu_1,\dots,\mu_d)$ (with nonnegative integers $\mu_i\ge \mu_j$ for any $i<j$) corresponds to a \emph{Young diagram}, and $|\mu|:=\sum_i \mu_i$. Each \emph{isotypic component} $\mathcal W^\mu$ can be further decomposed into a direct sum of equivalent irreducible subspaces: $\mathcal W^\mu = \bigoplus_{i=1}^{m_\mu} \mathcal W^{\mu,i}$. We define the dimension of the irreducible representation $d_\mu:=\dim\left(\mathcal W^{\mu,i}\right)$. It is worth mentioning that the first decomposition into isotypic components is unique while the second decomposition into equivalent irreducible representation spaces is not (see, e.g., \cite{fulton2013representation}). In terms of the unitary group action $G_\pi$ on $\mathcal W$, there exists a unitary transformation of basis $U$ such that for any $\pi$ we have
\begin{equation} \label{app:eq:block diagonalization of group action}
    U^\dagger G_\pi U = \bigoplus_{|\mu|=N} G_\pi^\mu \otimes \mathds{1}\left(m_\mu\right),
\end{equation}
where $G_\pi^\mu$ is a unitary operator on the irreducible representation space associated with the Young diagram label $\mu$, $m_\mu$ is the corresponding multiplicity, and $\mathds{1}\left(m_\mu\right)$ is an $m_\mu\times m_\mu$ identity matrix acting on the multiplicity subspace. By Schur's lemma, for any group-invariant operator $X$ on $\mathcal W$, i.e., $X$ commuting with all $G_\pi$, we have
\begin{equation} \label{app:eq:block diagonalization}
    U^\dagger X U = \bigoplus_{|\mu|=N} \mathds{1}\left(d_\mu\right) \otimes X^\mu,
\end{equation}
for any $\pi$, where $X^\mu$ is an $m_\mu \times m_\mu$ matrix. With such block diagonalization of the permutation-invariant operator, we reduce the dimension from $d^{2N}$ to \cite[Eq. (57)]{schur1901ueber}
\begin{equation}
    \sum_{|\mu|=N} m_\mu^2 = \binom{N+d^2-1}{d^2-1} \le (N+1)^{d^2-1},
\end{equation}
where the upper bound is obtained straightforwardly from the definition of the binomial coefficient. Specifically, if $X$ is further restricted to be a Hermitian matrix variable, then the number of \emph{real} scalar variables contained in all elements of $X$ is reduced from $d^{2N}$ to $\binom{N+d^2-1}{d^2-1}$. 

Now let us turn to the optimization problem of QFI evaluation. If the Hermitian matrix $h$ is taken to be permutation invariant by Lemma \ref{app:lem:qfi h symmetry}, the number of variables concerned in $h$ is reduced from $s^{2N}$ to $\binom{N+s^2-1}{s^2-1}$. Similarly, if further by Lemma \ref{app:lem:qfi Q symmetry} each Hermitian matrix $Q^i$ is permutation invariant, the number of variables in $Q^i$ is also reduced from $d^{2N}$ to $\binom{N+d^2-1}{d^2-1}$, where we redefine $d:=\dim(\mathcal H_1)\dim(\mathcal H_2)$. By this reduction the complexity gets polynomial rather than exponential, with respect to $N$. 

We can then reformulate the optimization in Theorem \ref{thm:qfi} with the reduced form. We consider two cases relevant to the strategy sets mentioned in the main text. First, under certain circumstances we can reduce the number of constraints in optimization as follows:
\begin{theorem}[Symmetry reduced Theorem \ref{thm:qfi}, first case] \label{app:thm:qfi h symmetry}
If, for any $\pi\in S_N$ and any $i$, there exists a $j$ such that the mapping $S\mapsto G_\pi SG_\pi^\dagger$ on $\mathsf S^i$ is a bijective function from $\mathsf S^i$ to $\mathsf S^j$, and meanwhile for any $i,j$ there exists some permutation operation such that $\mathsf S^i$ is bijectively mapped to $\mathsf S^j$, then the QFI of $N$ quantum channels $\mathcal E_\phi$ can be expressed as:
\begin{equation} \label{eq:thm:qfi h symmetry}
    \begin{aligned}
    J^{(\mathsf P)}(N_\phi) =& \min_{\lambda,Q^1,h^\mu} \lambda, \\
    \mathrm{s.t.}\ &\,\lambda Q^1 \ge \Omega_\phi(h),\ Q^1\in \overline{\mathsf S}^1,
    \end{aligned} 
\end{equation}
where $h=U' \left[\bigoplus_{|\mu|=N} \mathds{1}\left(d_\mu\right) \otimes h^\mu \right] U^{\prime\dagger}$ with each $h^\mu$ as an $m_\mu'\times m_\mu'$ matrix variable and $U'$ as a unitary transformation of basis.
\end{theorem}
\begin{proof}
We first prove that, for any $i,j$ there exists some permutation operation $Q\mapsto G_\pi Q G_\pi^\dagger$ such that $\overline{\mathsf S}^i$ is bijectively mapped to $\overline{\mathsf S}^j$. Given any $i$ and $j$, if we choose an arbitrary $Q^i \in \overline{\mathsf S}^i$, i.e., $\Tr(Q^i S^i) = 1$ for all $S^i\in \mathsf S^i$, then the condition of the theorem implies that there exists $G_\pi$ such that $G_\pi S^i G_\pi^\dagger \in \mathsf S^j$. As the map is bijective, in fact $G_\pi^\dagger S^j G_\pi \in \mathsf S^i$ for all $S^j\in \mathsf S^j$. Then we have $\Tr(G_\pi Q^i G_\pi^\dagger S^j)=\Tr(Q^i G_\pi^\dagger S^j G_\pi)=1$ for all $S^j\in \mathsf S^j$. Therefore, following the same argument in the proof of Lemma \ref{app:lem:qfi h symmetry}, under $Q \mapsto G_\pi Q G_\pi^\dagger$, $\overline{\mathsf S}^i$ is bijectively mapped to $\overline{\mathsf S}^j$.

By Lemma \ref{app:lem:qfi h symmetry} we can take $h$ to be permutation invariant and apply the block diagonalization to $h$ given by Eq. (\ref{app:eq:block diagonalization}). Then for any $Q^1\in \overline{\mathsf S}^1$ satisfying the constraint $\lambda Q^1 \ge \Omega_\phi(h)$, there exists $G_\pi$ such that $Q^i = G_\pi Q^1 G_\pi^\dagger \in \overline{\mathsf S}^i$ satisfying $\lambda Q^i \ge  G_\pi \Omega_\phi(h)  G_\pi^\dagger = \Omega_\phi(h)$ for any $i=2,\dots,K$. Therefore, all the constraints except for $\lambda Q^1 \ge \Omega_\phi(h)$ are redundant and can be removed.
\end{proof}
Now we consider the second case. If by Lemmas \ref{app:lem:qfi h symmetry} and \ref{app:lem:qfi Q symmetry} $h$ and each $Q^i$ are permutation invariant, we can then reformulate the constraints in optimization with reduced matrix dimensions. To relate the group representation on $\otimes_i\left(\mathcal H_{2i-1} \otimes \mathcal H_{2i}\right)$ to the representation on $\mathbb (C^s)^{\otimes N}$, we choose a unitary transformation of basis $U'$ for $h$, decompose $h=U' \left[\bigoplus_{|\mu|=N} \mathds{1}\left(d_\mu\right) \otimes h^\mu \right] U^{\prime\dagger}$ with $h^\mu$ as an $m_\mu'\times m_\mu'$ matrix, and divide $U' = \left( U^{\prime,\mu^1}, \dots, U^{\prime,\mu^I} \right)$ into blocks, where $\mu^i$ is the label of the $i$-th Young diagram, and $U^{\prime,\mu^i}$ has $d_{\mu^i} m_{\mu^i}'$ columns. Thus by defining the matrix
\begin{equation} \label{app:eq:N mu}
    \dot{\tilde{\mathbf N}}_\phi^\mu = \dot{\mathbf N}_\phi U^{\prime,\mu}-\mathrm i \mathbf N_\phi U^{\prime,\mu} \left[\mathds{1}\left(d_\mu\right) \otimes h^\mu \right], 
\end{equation}
we have the permutation-invariant performance operator
\begin{equation}
    \Omega_\phi(h) = 4\sum_{|\mu|=N}\left(\dot{\tilde{\mathbf N}}_\phi^\mu\dot{\tilde{\mathbf N}}_\phi^{\mu\dagger}\right)^T = 4\sum_{|\mu|=N}\dot{\tilde{\mathbf N}}_\phi^{\mu*}\dot{\tilde{\mathbf N}}_\phi^{\mu*\dagger}.
\end{equation}
Analogously, we apply the block diagonalization using a change of basis $U$ to 
\begin{equation} \label{app:eq:Qi}
    Q^i =U \left[\bigoplus_{|\mu|=N} \mathds{1}\left(d_{\mu}\right) \otimes Q^{i,\mu} \right] U^\dagger \in \overline{\mathsf S}^i,
\end{equation}
where $Q^{i,\mu}$ is an $m_\mu \times m_\mu$ matrix variable. The unitary transformation $U = \left(U^{\mu^1}, \dots, U^{\mu^I}\right)$ is first divided into blocks, and then each $U^{\mu^i}=\left(U^{\mu^i,1}, \dots, U^{\mu^i,d_{\mu^i}}\right)$ is divided into blocks for $i=1,\dots,I$, where $\mu^i$ is the label of the $i$-th Young diagram, and $U^{\mu^i,j}$ has $m_{\mu^i}$ columns for $j=1,\dots,d_{\mu^i}$. Note that $U$ also gives the block diagonalization of $\Omega_\phi(h)$.

Before proceeding further we prove that the range of $\dot{\tilde{\mathbf N}}_\phi^{\mu*}$ is exactly contained in the irreducible representation space corresponding to $\mu$:
\begin{lemma} \label{app:lem:young}
If we decompose the representation space $\mathcal W = \otimes_{i=1}^N\left(\mathcal H_{2i-1} \otimes \mathcal H_{2i}\right)  = \bigoplus_{|\mu|=N} \mathcal W^\mu$, then $\Range\left(\dot{\tilde{\mathbf N}}_\phi^{\mu*}\right) \subseteq \mathcal W^\mu$ for any Young diagram label $\mu$.
\end{lemma}
\begin{proof}
To characterize the representation space $\mathcal W^\mu$ we introduce the key notion of \emph{Young symmetrizer} as well as other related concepts very briefly (see, e.g., \cite{fulton2013representation} for mathematical details). A \emph{Young tableau} is a filling into the boxes of the Young diagram with positive integers weakly increasing along each row and strictly increasing along each column. For a \emph{standard} Young tableau labeled by $\nu$, i.e., a Young tableau filled with the entries $1,\dots,N$, we define two permutation subgroups 
\begin{equation}
    P_\nu := \left\{\sigma \in S_N \middle| \sigma\ \mathrm{preserves}\ \mathrm{each}\ \mathrm{row}\right\} 
\end{equation}
and 
\begin{equation}
    Q_\nu := \left\{\sigma \in S_N \middle| \sigma\ \mathrm{preserves}\ \mathrm{each}\ \mathrm{column}\right\}. 
\end{equation}
In the group algebra $\mathbb C S_N$ we define two elements $a_\nu := \sum_{\sigma \in P_\nu} e_\sigma$ and $b_\nu := \sum_{\sigma \in Q_\nu} \mathrm{sgn}(\sigma)e_\sigma$, where $e_\sigma$ is the unit vector corresponding to $\sigma$ and $\mathrm{sgn}(\cdot)$ denotes the parity of the permutation. Then the Young symmetrizer is defined by
\begin{equation}
    c_\nu := a_\nu b_\nu = \sum_{\sigma \in P_\nu} \sum_{\pi \in Q_\nu} \mathrm{sgn}(\pi)e_{\sigma\pi} \in \mathbb C S_N.
\end{equation}
It is known that a Young diagram of $\mu$ corresponds to $d_\mu$ standard Young tableaux, with each Young tableau of $\nu$ characterizing an irreducible representation space, given by the image of $c_\nu$ on $\otimes_i \mathcal W_i$ under the natural group algebra representation $\mathbb C S_N \rightarrow \mathrm{End}\left(\otimes_i \mathcal W_i\right)$,
where $\mathrm{End}(V)$ denotes the set of endomorphisms on $V$. 

%From Eq. (\ref{app:eq:group action on performance operator and h}) we know $G_\pi \Omega_\phi(h) G_\pi^\dagger = \Omega_\phi\left(G_\pi' h G_\pi^{\prime\dagger}\right)$ for any $\pi\in S_N$.
With these notions we now have an explicit characterization of the representation space. Note that we can prove $\Range\left(\dot{\tilde{\mathbf N}}_\phi^{\mu*}\right) \subseteq \mathcal W^\mu$ if $\Range\left(\dot{\mathbf N}_\phi^* U^{\prime,\mu*}\right)\subseteq \mathcal W^\mu$ and $\Range\left(\mathbf N_\phi^* U^{\prime,\mu*}\right)\subseteq \mathcal W^\mu$. In the proof of Lemma \ref{app:lem:qfi h symmetry} we have seen $G_\pi \ket{\dot N_{\phi,\mathbf i}} = \ket{\dot N_{\phi,\pi(\mathbf i)}}$ and $G_\pi \ket{N_{\phi,\mathbf i}} = \ket{N_{\phi,\pi(\mathbf i)}}$, from which it follows that
\begin{equation} \label{app:eq:group action on dot N phi}
    G_\pi \dot{\mathbf N}_\phi = \dot{\mathbf N}_\phi G_\pi'
\end{equation} 
and $G_\pi \mathbf N_\phi = \mathbf N_\phi G_\pi'$.

Now we prove that $\Range\left(\dot{\mathbf N}_\phi^* U^{\prime,\mu*}\right)\subseteq \mathcal W^\mu$. From the discussions above we know that
\begin{equation} \label{app:eq:young symmetrizer U h mu}
    \Range\left(U^{\prime,\mu}\right)=\bigoplus_\nu c_\nu \left[\otimes_i\mathbb C^s\right],
\end{equation}
for all Young tableau labels $\nu$ corresponding to the Young diagram of $\mu$. Explicitly, we have
\begin{equation} \label{app:eq:young symmetrizer U h mu explicit}
    c_\nu \left[\otimes_i\mathbb C^s\right] = \Range\left[\sum_{\sigma \in P_\nu} \sum_{\pi \in Q_\nu} \mathrm{sgn}(\pi)G_{\sigma\pi}'\right].
\end{equation} 
Note that there always exists a unitary transformation of basis $V$ such that we can obtain a real matrix $U_{(\mathrm{real})}^{\prime,\mu}=U^{\prime,\mu}V^\dagger$, which leads to
\begin{equation}
    \Range\left(U^{\prime,\mu*}\right)=\Range\left(U_{(\mathrm{real})}^{\prime,\mu*}V^*\right)=\Range\left(U_{(\mathrm{real})}^{\prime,\mu}V^*\right)=\Range\left(U^{\prime,\mu}\right).
\end{equation}
Thus $\Range\left(\dot{\mathbf N}_\phi^* U^{\prime,\mu*}\right)=\Range\left(\dot{\mathbf N}_\phi^* U^{\prime,\mu}\right)$. Furthermore, from Eqs. (\ref{app:eq:young symmetrizer U h mu}) and (\ref{app:eq:young symmetrizer U h mu explicit}) we have $\Range\left(\dot{\mathbf N}_\phi^* U^{\prime,\mu*}\right) \subseteq \mathcal W^\mu$ if
\begin{equation} \label{app:eq:range in representation space}
    \Range\left[\dot{\mathbf N}_\phi^*\sum_{\sigma \in P_\nu} \sum_{\pi \in Q_\nu} \mathrm{sgn}(\pi)G_{\sigma\pi}'\right] \subseteq \mathcal W^\mu
\end{equation}
for all Young tableau labels $\nu$ corresponding to the Young diagram of $\mu$. To show Eq. (\ref{app:eq:range in representation space}), note that by Eq. (\ref{app:eq:group action on dot N phi}) we have
\begin{equation} \label{app:eq:identical range by group action}
    \Range\left[\dot{\mathbf N}_\phi^*\sum_{\sigma \in P_\nu} \sum_{\pi \in Q_\nu} \mathrm{sgn}(\pi)G_{\sigma\pi}'\right] = \Range\left[\sum_{\sigma \in P_\nu} \sum_{\pi \in Q_\nu} \mathrm{sgn}(\pi) G_{\sigma\pi}\dot{\mathbf N}_\phi^*\right].
\end{equation}
Since the R.H.S. of Eq. (\ref{app:eq:identical range by group action}) is the image of the Young symmetrizer $c_\nu$ on $\Range\left(\dot{\mathbf N}_\phi^*\right)$, it is a subset of $\mathcal W^\mu$. Therefore, $\Range\left(\dot{\mathbf N}_\phi^* U^{\prime,\mu*}\right)\subseteq \mathcal W^\mu$. In the same way we can also show that $\Range\left(\mathbf N_\phi^* U^{\prime,\mu*}\right)\subseteq \mathcal W^\mu$ and thus complete the proof.
\end{proof}
Therefore, $U^\dagger \Omega_\phi(h) U$ is block diagonal with each block on the space $\Range\left(\dot{\tilde{\mathbf N}}_\phi^\mu\right)$. This results in:
\begin{theorem}[Symmetry reduced Theorem \ref{thm:qfi}, second case] \label{app:thm:qfi Q symmetry}
If each affine space $\mathsf S^i$ is permutation invariant for any $i=1,\dots,K$, then the QFI of $N$ quantum channels $\mathcal E_\phi$ can be expressed as:
\begin{equation} \label{eq:thm:qfi Q symmetry}
    \begin{aligned}
    J^{(\mathsf P)}(N_\phi) =& \min_{\lambda,Q^{i,\mu},h^\mu} \lambda, \\
    \mathrm{s.t.}\ &\,\lambda Q^{i,\mu} \ge 4U^{\mu,1\dagger} \left(\dot{\tilde{\mathbf N}}_\phi^\mu\dot{\tilde{\mathbf N}}_\phi^{\mu\dagger}\right)^T U^{\mu,1},\ i=1,\dots,K,\ |\mu|=N,
    \end{aligned} 
\end{equation}
where $Q^{i,\mu}$ is given by Eq. (\ref{app:eq:Qi}).
\end{theorem}
\begin{proof}
By Lemmas \ref{app:lem:qfi h symmetry} and \ref{app:lem:qfi Q symmetry}, $h$, $\Omega_\phi(h)$ and all $Q^i$ are taken to be permutation invariant. The original constraint $\lambda Q^i \ge \Omega_\phi(h)$ on the permutation-invariant space is reformulated as
\begin{equation} \label{app:eq:symmetry reduced constraint}
    \lambda \bigoplus_{|\mu|=N} \mathds{1}\left(d_\mu\right) \otimes Q^{i,\mu} \ge \bigoplus_{|\mu|=N} 4U^{\mu\dagger}\sum_{|\nu|=N}\dot{\tilde{\mathbf N}}_\phi^{\nu*}\dot{\tilde{\mathbf N}}_\phi^{\nu*\dagger}U^\mu.
\end{equation}
Since by Lemma \ref{app:lem:young} $\Range\left(\dot{\tilde{\mathbf N}}_\phi^{\mu*}\right) \subseteq \mathcal W^\mu=\Range\left(U^\mu\right)$ for any $\mu$, we have $U^{\mu\dagger}\sum_{|\nu|=N}\dot{\tilde{\mathbf N}}_\phi^{\nu*}\dot{\tilde{\mathbf N}}_\phi^{\nu*\dagger}U^\mu = U^{\mu\dagger}\dot{\tilde{\mathbf N}}_\phi^{\mu*}\dot{\tilde{\mathbf N}}_\phi^{\mu*\dagger}U^\mu$. Then Eq. (\ref{app:eq:symmetry reduced constraint}) can be reformulated as
\begin{equation} \label{app:eq:symmetry reduced constraint block}
    \lambda \mathds{1}\left(d_\mu\right) \otimes Q^{i,\mu} \ge 4U^{\mu\dagger}\dot{\tilde{\mathbf N}}_\phi^{\mu*}\dot{\tilde{\mathbf N}}_\phi^{\mu*\dagger}U^\mu,\ |\mu|=N.
\end{equation}
Furthermore, both sides of the inequality in Eq. (\ref{app:eq:symmetry reduced constraint block}) are block diagonal with the same repeating blocks and we only need to compare one of the blocks. Without loss of generality, we choose the first block for comparison and obtain
\begin{equation}
    \lambda Q^{i,\mu} \ge 4U^{\mu,1\dagger} \left(\dot{\tilde{\mathbf N}}_\phi^\mu\dot{\tilde{\mathbf N}}_\phi^{\mu\dagger}\right)^T U^{\mu,1},
\end{equation}
which holds for all $i=1,\dots,K$ and $|\mu|=N$.
\end{proof}
We have not characterized each dual affine space $\overline{\mathsf S}^i$ for $i=1,\dots,K$ yet. If we choose an affine basis $\{S^{i,j}\}_{j=1}^{M_i}$ for each $\mathsf S^i$, then the constraint $Q^i\in\overline{\mathsf S}^i$ can be reformulated as a set of linear constraints $\Tr(Q^iS^{i,j})=1$ for all $j=1,\dots,M_i$. If each affine space $\mathsf S^i$ is permutation invariant, then by the proof of Lemma \ref{app:lem:qfi Q symmetry} we have, for a feasible solution of $Q^i \in \overline{\mathsf S}^i$, $G_\pi Q^i G_\pi^\dagger$ is also a feasible solution for any $\pi\in S_N$. Then by defining 
\begin{equation} \label{app:eq:tilde Sij}
    \tilde S^{i,j}:=\frac{1}{N!}\sum_{\pi\in S_N} G_\pi S^{i,j}G_\pi^\dagger, 
\end{equation}
each linear constraint $\Tr(Q^iS^{i,j})=1$ can be replaced by $\Tr(Q^i\tilde S^{i,j})=1$ without changing the problem. Since $\tilde S^{i,j}$ is permutation invariant, similar to Eq. (\ref{app:eq:Qi}) it can be decomposed as
\begin{equation} \label{app:eq:tilde Sij decomposition}
    \tilde S^{i,j} =U \left[\bigoplus_{|\mu|=N} \mathds{1}\left(d_{\mu}\right) \otimes \tilde S^{i,j,\mu} \right] U^\dagger,
\end{equation}
where $\tilde S^{i,j,\mu}$ is an $m_\mu \times m_\mu$ matrix. Combining Eqs. (\ref{app:eq:Qi}) and (\ref{app:eq:tilde Sij decomposition}), we can reformulate the constraint $Q^i\in\overline{\mathsf S}^i$ with reduced matrix sizes as
\begin{equation} \label{app:eq:dual affine space constraint permutation-invariant}
    \sum_{|\mu|=N} d_\mu \Tr(Q^{i,\mu}\tilde S^{i,j,\mu}) = 1,\ j=1,\dots,M_i.
\end{equation}
 
Finally, it remains to be seen how to find the unitary transformation $U$ and $U'$ for block diagonalization. Known rigorous numerical algorithms for identifying the transformation are fairly expensive \cite{Montealegre-Mora21Certifying}. Fortunately, RepLAB \cite{Rosset_2021}, a numerical approach to decomposing representations based on randomized heuristics works very well in practice and is thus adopted here.

%The original constraint $\lambda Q^i \ge \Omega_\phi(h)$ on the permutation-invariant space is reformulated as
%\begin{equation}
%    \lambda \bigoplus_{|\mu|=N} \mathds{1}\left(d_\mu\right) \otimes Q^{i,\mu} \ge \left(\dot{\mathbf N}_\phi-\mathrm i \mathbf N_\phi h\right) \left[\dot{\mathbf N}_\phi-\mathrm i \mathbf N_\phi \left[\bigoplus_{|\mu|=N} \mathds{1}\left(d_\mu\right) \otimes h_\mu\right]\right]^\dagger
%\end{equation}

\subsection{Symmetry reduced algorithm for optimal strategies}
The idea is similar to the symmetry reduced evaluation of QFI. Since the first step of Algorithm \ref{alg:optimal probe} is simply solving the optimization problem in Theorem \ref{thm:qfi}, now we only consider its second step, where we need to solve for the optimal value of $\tilde P$ in
\begin{equation}
    \max_{\tilde{P} \in \tilde{\mathsf P}} \Tr\left[\tilde{P}\Omega_\phi\left(h^{(\mathrm{opt})}\right)\right],
\end{equation}
where $\tilde{\mathsf P} = \Conv\left\{\bigcup_{i=1}^{K}\left\{S^i\ge0\middle|S^i\in\mathsf S^i\right\}\right\}$ and
\begin{equation}
    \Re \left\{\Tr \left\{\tilde P\left[-\mathrm i \mathbf N_\phi \mathscr H\left(\dot{\mathbf N}_\phi-\mathrm i \mathbf N_\phi h^{(\mathrm{opt})}\right)^{\dagger}\right]^T\right\} \right\}
        =0\ \mathrm{for}\ \mathrm{all}\ \mathscr H \in \mathbb H_r.
\end{equation}
We have the following:
\begin{lemma} \label{app:lem:optimal strategy P symmetry}
If, for any $\pi\in S_N$ and any $i$, there exists a $j$ such that the mapping $S\mapsto G_\pi SG_\pi^\dagger$ on $\mathsf S^i$ is a bijective function from $\mathsf S^i$ to $\mathsf S^j$, then there must exist a permutation-invariant $\tilde P$ as an optimal solution in Algorithm \ref{alg:optimal probe}.
\end{lemma}
\begin{proof}
By definition $\tilde{\mathsf P}$ is permutation invariant. By Lemma \ref{app:lem:qfi h symmetry} we can choose a permutation-invariant $h^{(\mathrm{opt})}$ and thus $\Omega_\phi\left(h^{(\mathrm{opt})}\right)$ is also permutation invariant. If $\tilde P^{(\mathrm{opt})}$ is an optimal solution of $\tilde P$, then $G_\pi \tilde P^{(\mathrm{opt})} G_\pi^\dagger$ for any $\pi\in S_N$ is also an optimal solution, since
\begin{equation}
    \Tr\left[G_\pi\tilde P^{(\mathrm{opt})}G_\pi^\dagger\Omega_\phi\left(h^{(\mathrm{opt})}\right)\right]=\Tr\left[\tilde P^{(\mathrm{opt})}G_\pi^\dagger\Omega_\phi\left(h^{(\mathrm{opt})}\right)G_\pi\right]=\Tr\left[\tilde P^{(\mathrm{opt})}\Omega_\phi\left(h^{(\mathrm{opt})}\right)\right]
\end{equation}
and for all $\mathscr H \in \mathbb H_r$
\begin{equation} \label{app:eq:optimal strategy constraint permutation-invariant}
    \begin{aligned}
    &\Re \left\{\Tr \left\{G_\pi\tilde P^{(\mathrm{opt})}G_\pi^\dagger\left[-\mathrm i \mathbf N_\phi G_\pi'\mathscr HG_\pi^{\prime\dagger}\left(\dot{\mathbf N}_\phi-\mathrm i \mathbf N_\phi h^{(\mathrm{opt})}\right)^{\dagger}\right]^T\right\} \right\} \\
    =&\Re \left\{\Tr \left\{\tilde P^{(\mathrm{opt})}G_\pi^\dagger\left[-\mathrm i \mathbf N_\phi G_\pi'\mathscr HG_\pi^{\prime\dagger}\left(\dot{\mathbf N}_\phi-\mathrm i \mathbf N_\phi h^{(\mathrm{opt})}\right)^{\dagger}\right]^TG_\pi\right\} \right\} \\
    =&\Re \left\{\Tr \left\{\tilde P^{(\mathrm{opt})}\left[-\mathrm i \mathbf N_\phi \mathscr H\left(\dot{\mathbf N}_\phi-\mathrm i \mathbf N_\phi h^{(\mathrm{opt})}\right)^{\dagger}\right]^T\right\} \right\}\\
    =&\,0,
    \end{aligned}
\end{equation}
having used $G_\pi^\dagger\mathbf N_\phi = \mathbf N_\phi G_\pi^{\prime\dagger}$ and $G_\pi^\dagger\dot{\mathbf N}_\phi = \dot{\mathbf N}_\phi G_\pi^{\prime\dagger}$ for any $\pi\in S_N$ in the second equality of Eq. (\ref{app:eq:optimal strategy constraint permutation-invariant}). Therefore, there exists a permutation-invariant solution $\tilde P = \frac{1}{N!}\sum_{\pi\in S_N}G_\pi \tilde P^{(\mathrm{opt})} G_\pi^\dagger$.
\end{proof}
Now by following the same line of arguments as used from Eqs. (\ref{app:eq:tilde Sij}) to (\ref{app:eq:dual affine space constraint permutation-invariant}), we define 
\begin{equation}
    \begin{aligned}
    O:=&\,\frac{1}{N!}\sum_{\pi\in S_N} G_\pi \left[-\mathrm i \mathbf N_\phi \mathscr H\left(\dot{\mathbf N}_\phi-\mathrm i \mathbf N_\phi h^{(\mathrm{opt})}\right)^{\dagger}\right]^TG_\pi^\dagger \\ =&\,\left[-\mathrm i \mathbf N_\phi \tilde{\mathscr H}\left(\dot{\mathbf N}_\phi-\mathrm i \mathbf N_\phi h^{(\mathrm{opt})}\right)^{\dagger}\right]^T,
    \end{aligned}
\end{equation}
having defined the permutation-invariant $\tilde{\mathscr H}:=\frac{1}{N!}\sum_{\pi\in S_N}G_\pi'\mathscr HG_\pi^{\prime\dagger}$. Similar to the arguments in Appendix \ref{app:alg optimal probe}, by choosing a basis $\left\{\tilde{\mathscr H}^i\right\}_{i=1}^J$ for the permutation-invariant subspace of $\mathbb H_{r\times r}$, where $J=\binom{N+s^2-1}{s^2-1}$, Eq. (\ref{app:eq:optimal strategy constraint}) can be reformulated as a set of $J$ linear constraints. We further define 
\begin{equation}
    O^i := \left[-\mathrm i \mathbf N_\phi \tilde{\mathscr H}^i\left(\dot{\mathbf N}_\phi-\mathrm i \mathbf N_\phi h^{(\mathrm{opt})}\right)^{\dagger}\right]^T
\end{equation}
for $i=1,\dots,J$. Now we can decompose $O^i =U \left[\bigoplus_{|\mu|=N} \mathds{1}\left(d_{\mu}\right) \otimes \tilde O^{i,\mu} \right] U^\dagger$, $\tilde P =U \left[\bigoplus_{|\mu|=N} \mathds{1}\left(d_{\mu}\right) \otimes \tilde P^\mu \right] U^\dagger$, $\Omega_\phi\left(h^{(\mathrm{opt})}\right) =U \left[\bigoplus_{|\mu|=N} \mathds{1}\left(d_{\mu}\right) \otimes \Omega_\phi^\mu\left(h^{(\mathrm{opt})}\right) \right] U^\dagger$, and then reformulate the optimization problem as the following reduced form:
\begin{equation}
    \begin{aligned}
    \max_{\tilde P^\mu}&\sum_{|\mu|=N} d_\mu \Tr\left[\tilde P^\mu\Omega_\phi^\mu\left(h^{(\mathrm{opt})}\right)\right], \\
    \mathrm{s.t.}\ &\sum_{|\mu|=N} d_\mu\Re \left[\Tr \left(\tilde P^\mu O^{i,\mu}\right) \right]
        =0\ \mathrm{for}\ \mathrm{all}\ i=1,\dots,J.
    \end{aligned}
\end{equation}
Recall that $\tilde P=\sum_{i=1}^K q^iS^i$ for $S^i\in\mathsf S^i$. By choosing an affine basis $\{Q^{i,j}\}_{j=1}^{L_i}$ for $\overline{\mathsf S}^i$ we can also characterize $S^i\in\mathsf S^i$ by a set of linear constraints $\Tr\left(S^iQ^{i,j}\right)=1$ for $i=1,\dots,K$ and $j=1,\dots,L_i$, and follow the same routine to tackle the constraints on the permutation-invariant subspace. Thus both the number of variables and the number of constraints are polynomial with respect to $N$. 

\section{Evaluation of QFI using different strategies} \label{app:qfi all strategies}
In this section we provide explicit formulas of the QFI for all strategy sets considered in the main text, in the forms which can be numerically solved by SDP. Without the positivity constraints, parallel, sequential and general indefinite-causal-order strategy sets are affine spaces themselves, while quantum SWITCH and causal superposition strategy sets are convex hulls of affine spaces. In some cases the result of Theorem \ref{thm:qfi} can be simplified a bit, as it is possible to trace over certain subspace while formulating the primal problem at the beginning. 
\subsection{Parallel strategies}
When definite causal order is obeyed, a strategy can be described by a \emph{quantum comb} \cite{Chiribella2008PRL,Chiribella2008EPL,Chiribella2009PRA}. The dual affine space is the set of dual combs without the positivity constraint \cite{Chiribella_2016}. For parallel strategies the primal problem can be written as
\begin{equation} 
    \begin{aligned}
    J^{(\mathsf{Par})}(N_\phi) =&\min_{h \in \mathbb H_r} \max_{\tilde{P}}  \Tr\left[\tilde{P}\Omega_\phi(h)\right], \\
    \mathrm{s.t.}\ &\,\tilde P \ge 0, \\
    &\,\tilde P= \mathds{1}_{2,4,\dots,2N} \otimes \tilde P^{(1)}, \\
    &\Tr \tilde P^{(1)}=1. 
    \end{aligned}
\end{equation}
Equivalently, the problem can be formulated as
\begin{equation} \label{app:eq:par primal min max}
    \begin{aligned}
    \min_{h \in \mathbb H_r} \max_{P} &\Tr\left[P\Tr_{2,4,\dots,2N}\Omega_\phi(h)\right], \\
    \mathrm{s.t.}\ &\,P \ge 0, \\
    &\Tr P=1.
    \end{aligned}
\end{equation}
The dual problem is given by
\begin{equation} \label{app:eq:par dual}
    \begin{aligned}
    \min_{\lambda, h}&\,\lambda, \\
    \mathrm{s.t.}\ &\,\lambda\mathds{1}_{1,3,\dots,2N-1} \ge \Tr_{2,4,\dots,2N}\Omega_\phi(h),
    \end{aligned}
\end{equation}
which simplifies the result directly obtained from Theorem \ref{thm:qfi} a bit. To solve the problem via SDP, we define a block matrix
\begin{equation}
    A := \left( 
\begin{array}{c | c} 
  \frac \lambda 4 \mathds{1}\left(r\prod_{i=1}^N d_{2i}\right) & 
  \begin{array}{c}
       \bra{n_{1,1}}  \\
       \vdots \\
       \bra{n_{r,\prod_{i=1}^N d_{2i}}}
  \end{array}\\ 
  \hline 
  \begin{array}{ccc}
       \ket{n_{1,1}} & \hdots & \ket{n_{r,\prod_{i=1}^N d_{2i}}}  
  \end{array} & \mathds{1}_{1,3,\dots,2N-1}
 \end{array} 
\right),
\end{equation}
wherein $\mathds{1}(d)$ denotes a $d$-dimensional identity matrix, and
\begin{equation} \label{app:eq:nij par}
    \ket{n_{i,j}} := \braket{j}{\dot{\tilde{N}}_{\phi,i}^{*}},
\end{equation}
where $\ket{\dot{\tilde{N}}_{\phi,j}} = \ket{\dot{N}_{\phi,j}}-\mathrm i \sum_k\ket{N_{\phi,k}}h_{kj}$ and $\left\{\ket{j}, \ j=1,\dots,\prod_{k=1}^Nd_{2k}\right\}$ forms an orthonormal basis of $\otimes_{k=1}^N \mathcal H_{2k}$, having assumed that the identity map trivially acts on the subspace where the dual vector $\bra{j}$ does not affect. Note that 
\begin{equation}
    \sum_{i,j}\dyad{n_{i,j}} = \frac 1 4 \Tr_{2,4,\dots,2N}\Omega_\phi(h).
\end{equation}

By Schur complement lemma \cite[Theorem 1.12]{Horn2005}, the constraint in Eq. (\ref{app:eq:par dual}) is equivalent to the requirement that $A \ge 0$. Hence, the QFI for parallel strategies is solved by
\begin{equation}
    \begin{aligned}
    \min_{\lambda, h}&\,\lambda, \\
    \mathrm{s.t.}\ &\,A\ge 0, \\
    &\,h\in \mathbb H_r.
    \end{aligned}
\end{equation}
The problem can be solved by SDP since $h$ is incorporated linearly in the blocks of $A$.

\emph{Symmetry reduction}.—We can reduce the problem using permutation symmetry. For $\mathsf{Par}$, the set $\tilde{\mathsf P}$ is given by $\tilde{\mathsf P} = \left\{S\ge0\middle|S\in\mathsf S\right\}$, and the affine space $\mathsf S$ is permutation invariant. As explained in Appendix \ref{app:symmetry optimization} we can decompose the permutation-invariant $h=U' \left[\bigoplus_{|\mu|=N} \mathds{1}\left(d_\mu\right) \otimes h^\mu \right] U^{\prime\dagger}$ with $h^\mu$ as an $m_\mu'\times m_\mu'$ matrix. $Q\in \overline{\mathsf S}$ is characterized by the constraint $\Tr_{2,4,\dots,2N} Q = \mathds{1}_{1,3,\dots,2N-1}$, and we can decompose $Q = U \left[\bigoplus_{|\mu|=N} \mathds{1}\left(d_\mu\right) \otimes Q^\mu \right] U^\dagger$ with $Q^\mu$ as an $m_\mu\times m_\mu$ matrix. If we define 
\begin{equation}
    A^\mu := \left( 
    \begin{array}{c | c} 
  \frac \lambda 4 \mathds{1}\left(d_\mu m_\mu'\right) & 
  \left(U^{\mu,1\dagger} \dot{\tilde{\mathbf N}}_\phi^{\mu*}\right)^\dagger \\ 
  \hline 
  U^{\mu,1\dagger} \dot{\tilde{\mathbf N}}_\phi^{\mu*}& Q^\mu
 \end{array} 
\right),
\end{equation}
where $\dot{\tilde{\mathbf N}}_\phi^\mu$ is given by Eq. (\ref{app:eq:N mu}), then by Theorem \ref{app:thm:qfi Q symmetry} we can reformulate the optimization problem as
\begin{equation}
    \begin{aligned}
    \min_{\lambda,Q^\mu,h}&\,\lambda, \\
    \mathrm{s.t.}\ &\,A^\mu\ge 0,\ h^\mu\in \mathbb H_{m_\mu'},\ |\mu|=N, \\
    &\Tr_{2,4,\dots,2N} Q = \mathds{1}_{1,3,\dots,2N-1}.
    \end{aligned}
\end{equation}
The constraint $\Tr_{2,4,\dots,2N} Q = \mathds{1}_{1,3,\dots,2N-1}$ only requires to be explicitly characterized on the permutation-invariant subspace, since $\Tr_{2,4,\dots,2N} Q$ is permutation invariant. Therefore, not only the number of scalar variables but also the number of constraints in terms of scalar variables are polynomial with respect to $N$.

\subsection{Sequential strategies}
For sequential strategies the problem can be written as (having traced over $\mathcal H_{2N}$)
\begin{equation}
    \begin{aligned}
    J^{(\mathsf{Seq})}(N_\phi) =&\min_{h \in \mathbb H_r} \max_{P^{(k)}}\Tr\left[P^{(N)}\Tr_{2N}\Omega_\phi(h) \right], \\
    \mathrm{s.t.}\ &\,P^{(N)}\ge 0, \\
    &\Tr_{2k-1} P^{(k)} =  \mathds{1}_{2k-2} \otimes P^{(k-1)}, \ k=2,\dots, N, \\
    &\Tr P^{(1)}=1,
    \end{aligned}
\end{equation}
from which it follows that the dual problem is
\begin{equation}
    \begin{aligned}
    \min_{\lambda,Q^{(k)},h}&\,\lambda, \\
    \mathrm{s.t.}\ &\,\lambda\mathds{1}_{2N-1} \otimes Q^{(N-1)} \ge \Tr_{2N}\Omega_\phi(h), \\
    &\Tr_{2k}Q^{(k)} = \mathds{1}_{2k-1} \otimes Q^{(k-1)}, \ k=2,\dots, N-1, \\
    &\Tr_2 Q^{(1)} = \mathds{1}_1,
    \end{aligned}
\end{equation}
where $Q^{(N-1)}$ is Hermitian. 

Similarly, in order to solve the problem via SDP we rewrite it as
\begin{equation}
    \begin{aligned}
    \min_{\lambda, Q^{(k)}, h}&\,\lambda, \\
    \mathrm{s.t.}\ &\,A\ge 0, \\
    &\Tr_{2k}Q^{(k)} = \mathds{1}_{2k-1} \otimes Q^{(k-1)}, \ k=2,\dots, N-1, \\
    &\Tr_2 Q^{(1)} = \mathds{1}_1, \\
    &\,h\in \mathbb H_r,
    \end{aligned}
\end{equation}
for
\begin{equation}
    A := \left( 
\begin{array}{c | c} 
  \frac \lambda 4 \mathds{1}\left(rd_{2N}\right) & 
  \begin{array}{c}
       \bra{n_{1,1}}  \\
       \vdots \\
       \bra{n_{r,d_{2N}}}
  \end{array}\\ 
  \hline 
  \begin{array}{ccc}
       \ket{n_{1,1}} & \hdots & \ket{n_{r,d_{2N}}}  
  \end{array} & \mathds{1}_{2N-1} \otimes Q^{(N-1)}
 \end{array} 
\right),
\end{equation}
having defined
\begin{equation}
    \ket{n_{i,j}} := \braket{j}{\dot{\tilde{N}}_{\phi,i}^{*}},
\end{equation}
where $\ket{\dot{\tilde{N}}_{\phi,j}} = \ket{\dot{N}_{\phi,j}}-\mathrm i \sum_k\ket{N_{\phi,k}}h_{kj}$ and $\left\{\ket{j}, \ j=1,\dots,d_{2N}\right\}$ forms an orthonormal basis of $\mathcal H_{2N}$.

\subsection{Quantum SWITCH strategies}
We first formally define a quantum SWITCH strategy set $\mathsf{SWI}$ as the collection of $P \in \mathsf{Strat}$ such that
\begin{equation}
    P = \left(\rho_{T,A,C} \right) * \dyad{P^{(\mathrm{SW})}},\ \rho_{T,A,C} \ge 0,\ \Tr \rho_{T,A,C} = 1,
\end{equation}
where $\ket{P^{(\mathrm{SW})}}
:= \lvert I \rrangle_{A,F_A}\sum_{\pi\in S_N}\left[\ket{\pi}_C \lvert I \rrangle_{T,2\pi(1)-1}\left(\otimes_{i=1}^{N-1}\lvert I \rrangle_{2\pi(i),2\pi(i+1)-1}\right) \lvert I \rrangle_{2\pi(N),F_T} \ket{\pi}_{F_C}\right]$ corresponds to a (generalized) quantum SWITCH for $N$ operations, each permutation $\pi$ is an element of the symmetric group $S_N$ whose order is $N!$, and $\{\ket{\pi}_C\}$ forms an orthonormal basis. We suppose each $\mathcal H_i$ for $i=1,\dots,2N$ has the same dimension $d_1$. $\mathcal H_T\simeq\mathcal H_i$ denotes the input space of the target system, $\mathcal H_A$ the ancillary space, and $\mathcal H_C$ the space of the control system. Correspondingly, $\mathcal H_{F_T}$, $\mathcal H_{F_A}$ and $\mathcal H_{F_C}$ denote the future output spaces of each part. The global future space $\mathcal H_F = \mathcal H_{F_T} \otimes \mathcal H_{F_A} \otimes \mathcal H_{F_C}$.

Using the quantum SWITCH strategy set, after tracing over the global future space $\mathcal H_F$ the QFI evaluation problem is written as
\begin{equation}
    \begin{aligned}
    J^{(\mathsf{SWI})}(N_\phi) =&\min_{h \in \mathbb H_r} \max_{\tilde{P}}  \Tr\left[\tilde{P}\Omega_\phi(h)\right], \\
    \mathrm{s.t.}\ &\,\tilde P = \sum_{\pi\in S_N}q^{\pi}\rho_{2\pi(1)-1}^{\pi}\left(\otimes_{i=1}^{N-1}\lvert I \rrangle_{2\pi(i),2\pi(i+1)-1}\llangle I \rvert_{2\pi(i),2\pi(i+1)-1}\right) \otimes\mathds{1}_{2\pi(N)}, \\
    &\sum_{\pi\in S_N} q^\pi=1, \\
    &\,\rho_{2\pi(1)-1}^\pi \ge 0,\ \Tr\rho_{2\pi(1)-1}^\pi = 1,\ q^\pi\ge0,\ \pi\in S_N,
    \end{aligned}
\end{equation}
where the superscript $\pi$ of an operator denotes a permutation label, and the subscript denotes the subspace it lies in. Note that the primal set of $\tilde P$ is a convex hull of affine spaces. Equivalently the problem can be rewritten as
\begin{equation}
    \begin{aligned}
    \min_{h \in \mathbb H_r} \max_{\rho_{2\pi(1)-1}^\pi,q^\pi}&\sum_{\pi\in S_N}\Tr\left[q^\pi\rho_{2\pi(1)-1}^\pi\left(\otimes_{i=1}^{N-1}\llangle I \rvert_{2\pi(i),2\pi(i+1)-1}\right)\Tr_{2\pi(N)}\Omega_\phi(h)\left(\otimes_{j=1}^{N-1}\lvert I \rrangle_{2\pi(j),2\pi(j+1)-1}\right)\right], \\
    \mathrm{s.t.}\ &\sum_{\pi\in S_N} q^\pi=1, \\
    &\,\rho_{2\pi(1)-1}^\pi \ge 0,\ \Tr\rho_{2\pi(1)-1}^\pi = 1,\ q^\pi\ge0,\ \pi\in S_N.
    \end{aligned}
\end{equation}
Following the method in the proof of Theorem \ref{thm:qfi}, the dual problem is given by
\begin{equation}
    \begin{aligned}
    \min&\,\lambda, \\
    \mathrm{s.t.}\ &\,\lambda\mathds{1}_{2\pi(1)-1} \ge \Omega_\phi^\pi(h),\ \Omega_\phi^\pi(h):=\left(\otimes_{i=1}^{N-1}\llangle I \rvert_{2\pi(i),2\pi(i+1)-1}\right)\Tr_{2\pi(N)}\Omega_\phi(h)\left(\otimes_{j=1}^{N-1}\lvert I \rrangle_{2\pi(j),2\pi(j+1)-1}\right),\ \pi\in S_N.
    \end{aligned}
\end{equation}
Equivalently in an SDP form the problem is written as
\begin{equation} 
    \begin{aligned}
    \min_{\lambda,h}&\,\lambda, \\
    \mathrm{s.t.}\ &\,A^\pi \ge 0,\ \pi\in S_N, \\
    &\,h\in \mathbb{H}_r,
    \end{aligned}
\end{equation}
having defined
\begin{equation}
    A^\pi := \left( 
\begin{array}{c | c} 
  \frac \lambda 4 \mathds{1}(rd_1) & 
  \begin{array}{c}
       \bra{n_{1,1}^\pi}  \\
       \vdots \\
       \bra{n_{r,d_1}^\pi}
  \end{array}\\ 
  \hline 
  \begin{array}{ccc}
       \ket{n_{1,1}^\pi} & \hdots & \ket{n_{r,d_1}^\pi}  
  \end{array} & \mathds{1}_{2\pi(1)-1}
 \end{array} 
\right),
\end{equation}
for
\begin{equation}
    \ket{n_{i,j}^\pi} := \bra{j^\pi}\left(\otimes_{k=1}^{N-1}\llangle I \rvert_{2\pi(k),2\pi(k+1)-1}\right)\ket{\dot{\tilde{N}}_{\phi,i}^{*}},
\end{equation}
where $\ket{\dot{\tilde{N}}_{\phi,j}} = \ket{\dot{N}_{\phi,j}}-\mathrm i \sum_k\ket{N_{\phi,k}}h_{kj}$ and $\{\ket{j^\pi}\}$ forms an orthonormal basis of $\mathcal H_{2\pi(N)}$. 

\emph{Symmetry reduction}.—For $\mathsf{SWI}$, by Theorem \ref{app:thm:qfi h symmetry}, $h$ can be taken to be permutation invariant and the constraint corresponding to one permutation $\pi$ (e.g., the identity element of $S_N$) is sufficient. As explained in Appendix \ref{app:symmetry optimization} we can decompose the permutation-invariant $h=U' \left[\bigoplus_{|\mu|=N} \mathds{1}\left(d_\mu\right) \otimes h^\mu \right] U^{\prime\dagger}$ with $h^\mu$ as an $m_\mu'\times m_\mu'$ matrix. We define 
\begin{equation}
    \mathbf n_{j_{2N}}^\mu := \bra{j_{2N}}\left(\otimes_{i=1}^{N-1}\llangle I \rvert_{2i,2i+1}\right)\dot{\tilde{\mathbf N}}_\phi^{\mu*},
\end{equation}
where $\{\ket{j_{2N}}\}$ forms an orthonormal basis of $\mathcal H_{2N}$ and $\dot{\tilde{\mathbf N}}_\phi^\mu$ is given by Eq. (\ref{app:eq:N mu}). If we define 
\begin{equation}
    A^{(\mathrm{inv})} := \left( 
    \begin{array}{c | c} 
  \frac \lambda 4 \mathds{1}(rd_1) & 
  \begin{array}{c}
       \mathbf n_1^{\mu^1\dagger} \\
       \vdots \\
       \mathbf n_{d_1}^{\mu^I\dagger}
  \end{array} \\
  \hline 
  \begin{array}{ccc}
       \mathbf n_1^{\mu^1} & \hdots & \mathbf n_{d_1}^{\mu^I} 
  \end{array}
 & \mathds{1}_1
 \end{array} 
\right),
\end{equation}
then by Theorem \ref{app:thm:qfi h symmetry} we can reformulate the optimization problem as
\begin{equation}
    \begin{aligned}
    \min_{\lambda,h^\mu}&\,\lambda, \\
    \mathrm{s.t.}\ &\,A^{(\mathrm{inv})}\ge 0,\\
    &\,h^\mu\in \mathbb H_{m_\mu'},\ |\mu|=N.
    \end{aligned}
\end{equation}

\subsection{Causal superposition strategies}
Following a similar route for the causal superposition strategy set the problem can be written as
\begin{equation}
    \begin{aligned}
    J^{(\mathsf{Sup})}(N_\phi) =&\min_{h \in \mathbb H_r} \max_{P^{\pi,(N)},q^\pi} \sum_{\pi\in S_N}\Tr\left(q^\pi P^{\pi,(N)}\Tr_{2\pi(N)}\Omega_\phi(h)\right), \\
    \mathrm{s.t.}\ &\sum_{\pi\in S_N} q^\pi=1, \\
    &\,q^\pi\ge0,\ P^{\pi,(N)}\ge 0,\ \Tr P^{\pi,(1)}=1,\ \Tr_{2k-1} P^{\pi,(k)} =  \mathds{1}_{2k-2} \otimes P^{\pi,(k-1)}\ \mathrm{for}\ k=2,\dots, N,\ \pi\in S_N.
    \end{aligned}
\end{equation}
For each causal order in the superposition the dual affine space is the set of dual combs without the positivity constraint. Thus the dual problem is given by
\begin{equation}
    \begin{aligned}
    \min&\,\lambda, \\
    \mathrm{s.t.}\ &\,\lambda\mathds{1}_{2\pi(N)-1} \otimes Q^{\pi,(N-1)} \ge \Tr_{2\pi(N)}\Omega_\phi(h),\ \Tr_{2\pi(1)} Q^{\pi,(1)} = \mathds{1}_{2\pi(1)-1},\ \pi\in S_N, \\
    &\Tr_{2\pi(k)}Q^{\pi,(k)} = \mathds{1}_{2\pi(k)-1} \otimes Q^{\pi,(k-1)}\ \mathrm{for}\ k=2,\dots, N-1,\ \pi\in S_N,
    \end{aligned}
\end{equation}
where the constraints hold for any $\pi\in S_N$. To solve the problem via SDP we can formulate it as
\begin{equation} 
    \begin{aligned}
    \min_{\lambda,h}&\,\lambda, \\
    \mathrm{s.t.}\ &\,A^\pi \ge 0,\ \Tr_{2\pi(1)} Q^{\pi,(1)} = \mathds{1}_{2\pi(1)-1},\ \Tr_{2\pi(k)}Q^{\pi,(k)} = \mathds{1}_{2\pi(k)-1} \otimes Q^{\pi,(k-1)}\ \mathrm{for}\ k=2,\dots, N-1,\ \pi\in S_N, \\
    &\,h\in \mathbb{H}_r,
    \end{aligned}
\end{equation}
having defined
\begin{equation}
    A^\pi := \left( 
\begin{array}{c | c} 
  \frac \lambda 4 \mathds{1}\left(rd_{2\pi(N)}\right) & 
  \begin{array}{c}
       \bra{n_{1,1}^\pi}  \\
       \vdots \\
       \bra{n_{r,d_{2\pi(N)}}^\pi}
  \end{array}\\ 
  \hline 
  \begin{array}{ccc}
       \ket{n_{1,1}^\pi} & \hdots & \ket{n_{r,d_{2\pi(N)}}^\pi}  
  \end{array} & \mathds{1}_{2\pi(N)-1} \otimes Q^{\pi,(N-1)}
 \end{array} 
\right),
\end{equation}
for
\begin{equation}
    \ket{n_{i,j}^\pi} := \bra{j^\pi}\ket{\dot{\tilde{N}}_{\phi,i}^{*}},
\end{equation}
where $\ket{\dot{\tilde{N}}_{\phi,j}} = \ket{\dot{N}_{\phi,j}}-\mathrm i \sum_k\ket{N_{\phi,k}}h_{kj}$ and $\{\ket{j^\pi}\}$ forms an orthonormal basis of $\mathcal H_{2\pi(N)}$. 

\emph{Symmetry reduction}.—Similar to $\mathsf{SWI}$, by Theorem \ref{app:thm:qfi h symmetry}, for $\mathsf{Sup}$ we take a permutation-invariant $h$ and only need the constraint corresponding to the identity element of $S_N$. As explained in Appendix \ref{app:symmetry optimization} we can decompose the permutation-invariant $h=U' \left[\bigoplus_{|\mu|=N} \mathds{1}\left(d_\mu\right) \otimes h^\mu \right] U^{\prime\dagger}$ with $h^\mu$ as an $m_\mu'\times m_\mu'$ matrix. We define 
\begin{equation}
    \mathbf n_{j_{2N}}^\mu := \bra{j_{2N}}\dot{\tilde{\mathbf N}}_\phi^{\mu*},
\end{equation}
where $\{\ket{j_{2N}}\}$ forms an orthonormal basis of $\mathcal H_{2N}$ and $\dot{\tilde{\mathbf N}}_\phi^\mu$ is given by Eq. (\ref{app:eq:N mu}). If we define 
\begin{equation}
    A^{(\mathrm{inv})} := \left( 
    \begin{array}{c | c} 
  \frac \lambda 4 \mathds{1}\left(rd_{2N}\right) & 
  \begin{array}{c}
       \mathbf n_1^{\mu^1\dagger} \\
       \vdots \\
       \mathbf n_{d_1}^{\mu^I\dagger}
  \end{array} \\
  \hline 
  \begin{array}{ccc}
       \mathbf n_1^{\mu^1} & \hdots & \mathbf n_{d_1}^{\mu^I} 
  \end{array}
 & \mathds{1}_{2N-1} \otimes Q^{(N-1)}
 \end{array} 
\right),
\end{equation}
then by Theorem \ref{app:thm:qfi h symmetry} we can reformulate the optimization problem as
\begin{equation}
    \begin{aligned}
    \min_{\lambda, Q^{(k)}, h}&\,\lambda, \\
    \mathrm{s.t.}\ &\,A^{(\mathrm{inv})}\ge 0, \\
    &\Tr_{2k}Q^{(k)} = \mathds{1}_{2k-1} \otimes Q^{(k-1)}, \ k=2,\dots, N-1, \\
    &\Tr_2 Q^{(1)} = \mathds{1}_1, \\
    &\,h^\mu\in \mathbb H_{m_\mu'},\ |\mu|=N.
    \end{aligned}
\end{equation}

\subsection{General indefinite-causal-order strategies}
In this case the explicit linear constraints on strategies have been derived in \cite{Araujo2015}, and the dual affine space turns out to be the set of CJ operators of $N$-partite no-signaling quantum channels without the positivity constraint \cite{Chiribella2013PRA,Chiribella_2016}, mathematically defined by
\begin{equation}
    \begin{aligned}
    \Tr_{2k}Q &= \frac{\mathds{1}_{2k-1}}{d_{2k-1}}\otimes\Tr_{2k-1,2k}Q,\ k=1,\dots,N, \\
    \Tr Q&=\prod_{i=1}^Nd_{2i-1}.
    \end{aligned}
\end{equation}
The intuitive interpretation for no-signaling channels is that locally the input of each channel only affects the output of this single channel, but cannot transmit any information to $N-1$ other channels. To solve the QFI evaluation problem via SDP we can write it in the form
\begin{equation}
    \begin{aligned}
    \min_{\lambda,Q,h}&\,\lambda, \\
    \mathrm{s.t.}\ &\,A\ge0, \\
    &\Tr_{2k}Q = \frac{\mathds{1}_{2k-1}}{d_{2k-1}}\otimes\Tr_{2k-1,2k}Q,\ k=1,\dots,N, \\
    &\Tr Q=\prod_{i=1}^Nd_{2i-1}, \\
    &\,h\in\mathbb H_r,
    \end{aligned}
\end{equation}
having defined
\begin{equation}
    A = \left( 
\begin{array}{c | c} 
  \frac \lambda 4 \mathds{1}\left(r\right) & 
  \begin{array}{c}
       \bra{\dot{\tilde{N}}_{\phi,1}^{*}} \\
       \vdots \\
       \bra{\dot{\tilde{N}}_{\phi,r}^{*}}
  \end{array}\\ 
  \hline 
  \begin{array}{ccc}
       \ket{\dot{\tilde{N}}_{\phi,1}^{*}} & \hdots & \ket{\dot{\tilde{N}}_{\phi,r}^{*}}
  \end{array} & Q
 \end{array} 
\right).
\end{equation}

\emph{Symmetry reduction}.—For $\mathsf{ICO}$, by Lemmas \ref{app:lem:qfi h symmetry} and \ref{app:lem:qfi Q symmetry}, both $h$ and $Q$ can be taken to be permutation invariant. As explained in Appendix \ref{app:symmetry optimization} we can decompose the permutation-invariant $h=U' \left[\bigoplus_{|\mu|=N} \mathds{1}\left(d_\mu\right) \otimes h^\mu \right] U^{\prime\dagger}$ with $h^\mu$ as an $m_\mu'\times m_\mu'$ matrix, and decompose the permutation-invariant $Q=U \left[\bigoplus_{|\mu|=N} \mathds{1}\left(d_\mu\right) \otimes Q^\mu \right] U^\dagger$ with $Q^\mu$ as an $m_\mu \times m_\mu$ matrix. If we define 
\begin{equation}
    A^\mu := \left( 
    \begin{array}{c | c} 
  \frac \lambda 4 \mathds{1}\left(d_\mu m_\mu'\right) & 
  \left(U^{\mu,1\dagger} \dot{\tilde{\mathbf N}}_\phi^{\mu*}\right)^\dagger \\ 
  \hline 
  U^{\mu,1\dagger} \dot{\tilde{\mathbf N}}_\phi^{\mu*}& Q^\mu
 \end{array} 
\right),
\end{equation}
where $\dot{\tilde{\mathbf N}}_\phi^\mu$ is given by Eq. (\ref{app:eq:N mu}), then by Theorem \ref{app:thm:qfi Q symmetry} we can reformulate the optimization problem as
\begin{equation}
    \begin{aligned}
    \min_{\lambda,Q^\mu,h}&\,\lambda, \\
    \mathrm{s.t.}\ &\,A^\mu\ge 0,\ h^\mu\in \mathbb H_{m_\mu'},\ |\mu|=N, \\
    &\Tr_{2N} Q = \frac{\mathds{1}_{2N-1}}{d_{2N-1}} \otimes \Tr_{2N-1,2N}Q, \\
    &\sum_{|\mu|=N} d_\mu\Tr Q^\mu=\prod_{i=1}^Nd_{2i-1},
    \end{aligned}
\end{equation}
having removed the redundant constraints on $Q$ by permutation symmetry. Similar to the case of $\mathsf{Par}$, we only need to consider the constraint on $Q$ on the permutation-invariant subspace, since both $Q$ and $\Tr_{2N-1,2N} Q$ are permutation invariant. Therefore, both the number of scalar variables and the number of constraints in terms of scalar variables are polynomial with respect to $N$.

\section{Complexity analysis} \label{app:complexity}
Here we refer to the number of real scalar variables concerned in optimization as the complexity. In Appendix \ref{app:qfi all strategies} we have presented both the original and the symmetry reduced QFI evaluation for all families of strategies considered in this work as applicable, from which we can obtain Table \ref{tab:complexity} in the main text.

Now we consider the algorithm for identifying an optimal strategy. Since the algorithm is based on Theorem \ref{thm:qfi}, its complexity is no less than that of the QFI evaluation. For the second step of the algorithm, by Lemma \ref{app:lem:optimal strategy P symmetry} there exists a permutation-invariant optimal strategy for all families of strategies except for the sequential one, and thus we can apply the group-invariant SDP to achieve the permutation-invariant solution. Recall that $\tilde P=\sum_{i=1}^K q^iS^i$ for $S^i\in\mathsf S^i$. In fact, for $\mathsf{SWI}$ and $\mathsf{Sup}$, we only need to characterize $S^1\in\mathsf S^1$ and obtain $\tilde P=\sum_{i=1}^K \frac 1K S^1$ due to the permutation invariance. In summary, taking both steps of the algorithm into account, we obtain Table \ref{tab:complexity of algorithm 1}:

\begin{table}[!htbp]
\caption{\label{tab:complexity of algorithm 1}\textbf{Complexity of Algorithm \ref{alg:optimal probe} for each family of strategies} (with respect to $N$). The asymptotic numbers of variables in optimization are compared between the original (Ori.) and group-invariant (Inv.) SDP. We denote $d:=d_1 d_2$ for $d_i:=\dim(\mathcal H_i)$ and $s:=\max_\phi\rank(E_\phi)\le d$.}
\begin{ruledtabular}
\begin{tabular}{cccccc}
SDP & $\mathsf{Par}$ & $\mathsf{Seq}$ & $\mathsf{SWI}$ & $\mathsf{Sup}$ & $\mathsf{ICO}$ \\
\colrule
Ori. & $O\left(\max(s,d_1)^N\right)$ & $O\left(d^N\right)$ & $O\left(N!\right)$ & $O\left(N!\,d^N\right)$ & $O\left(d^N\right)$ \\
Inv. & $O\left(N^{d^2-1}\right)$ & $O\left(d^N\right)$ & $O(N^{s^2-1})$ & $O\left(d^N\right)$ & $O\left(N^{d^2-1}\right)$\\
\end{tabular}
\end{ruledtabular}
\end{table}

As a concrete example, we apply the group-invariant SDP to the evaluation of the QFI $J^{(\mathsf{ICO})}$ for the general indefinite-causal-order strategies up to $N=5$. We consider the amplitude damping noise and take $\phi=1.0$, $t=1.0$ and $p=0.5$. As illustrated in Fig. \ref{fig:qfi ico growth}, the growth of $J^{(\mathsf{ICO})}$ is faster than linear growth but slower than quadratic growth. Table \ref{tab:Complexity growth ICO} compares the complexity between the original and the group-invariant SDP and indicates that the symmetry reduced approach can save the computational resources dramatically.

\begin{figure} [!htbp]
    \centering
    \includegraphics[width=\linewidth]{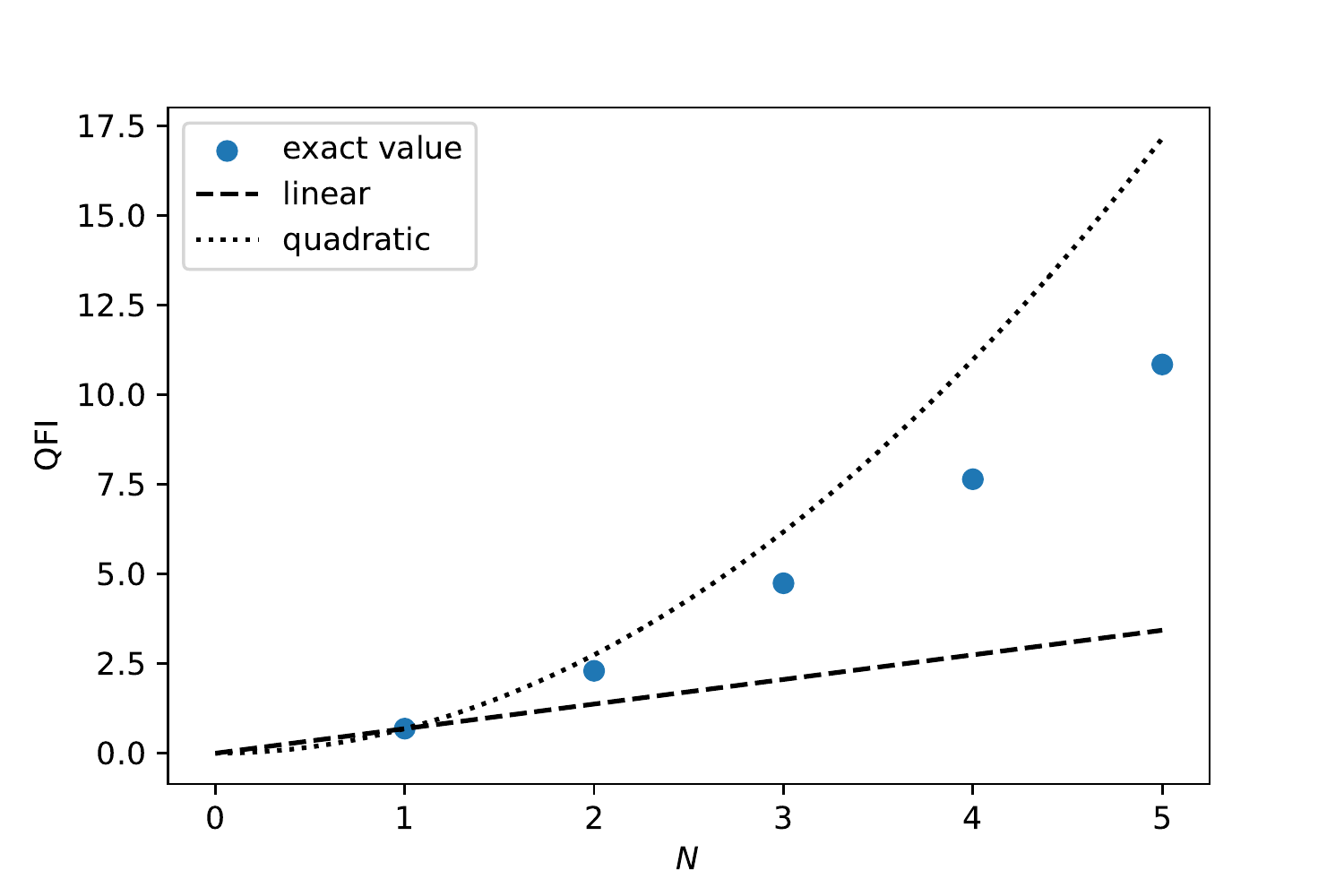}
    \caption{\textbf{Growth of QFI $J^{(\mathsf{ICO})}$ as $N$ increases.} For the amplitude damping noise, we take $\phi=1.0$, $t=1.0$ and $p=0.5$. The dashed and dotted lines illustrate the linear and quadratic growth with respect to $N$ respectively, while matching the QFI at $N=1$.}
    \label{fig:qfi ico growth}
\end{figure}

\begin{table}[!htbp]
\caption{\label{tab:Complexity growth ICO}\textbf{Number of real scalar variables concerned in the evaluation of $J^{(\mathsf{ICO})}$}. For the original SDP the total number of real scalar variables for $Q$, $h$ and $\lambda$ is $16^N+4^N+1$, while for the group-invariant SDP the total number of real scalar variables for $Q^\mu$, $h^\mu$ and $\lambda$ for all Young diagram labels $|\mu|=N$ is $\binom{N+15}{15}+\binom{N+3}{3}+1$.}
\begin{ruledtabular}
\begin{tabular}{cccccc}
$N$ & 1 & 2 & 3 & 4 & 5 \\
\colrule
No. of variables in the original SDP & 21 & 273 & 4161 & 65793 & 1049601 \\
No. of variables in the group-invariant SDP & 21 & 147 & 837 & 3912 & 15561\\
\end{tabular}
\end{ruledtabular}
\end{table}

\section{Supplementary numerical results} \label{app:supplementary numerical results}
\subsection{Hierarchy for the $N=3$ case}
Numerical results in this work are obtained by implementing SDP using the open source Python package CVXPY \cite{diamond2016cvxpy,agrawal2018rewriting} with the solver MOSEK \cite{mosek}. We plot the QFI for $N=3$, amplitude damping noise in Fig. \ref{fig:qfi all strategies tripartite} and observe a similar hierarchy of the estimation performance using parallel, sequential and indefinite-causal-order strategies. Different from the $N=2$ case presented in the main text, general indefinite-causal-order strategies indeed provide a small advantage over causal superposition strategies, which is presented in Table. \ref{tab:ICO Sup}.

\begin{figure} [!htbp]
    \centering
    \includegraphics[width=\linewidth]{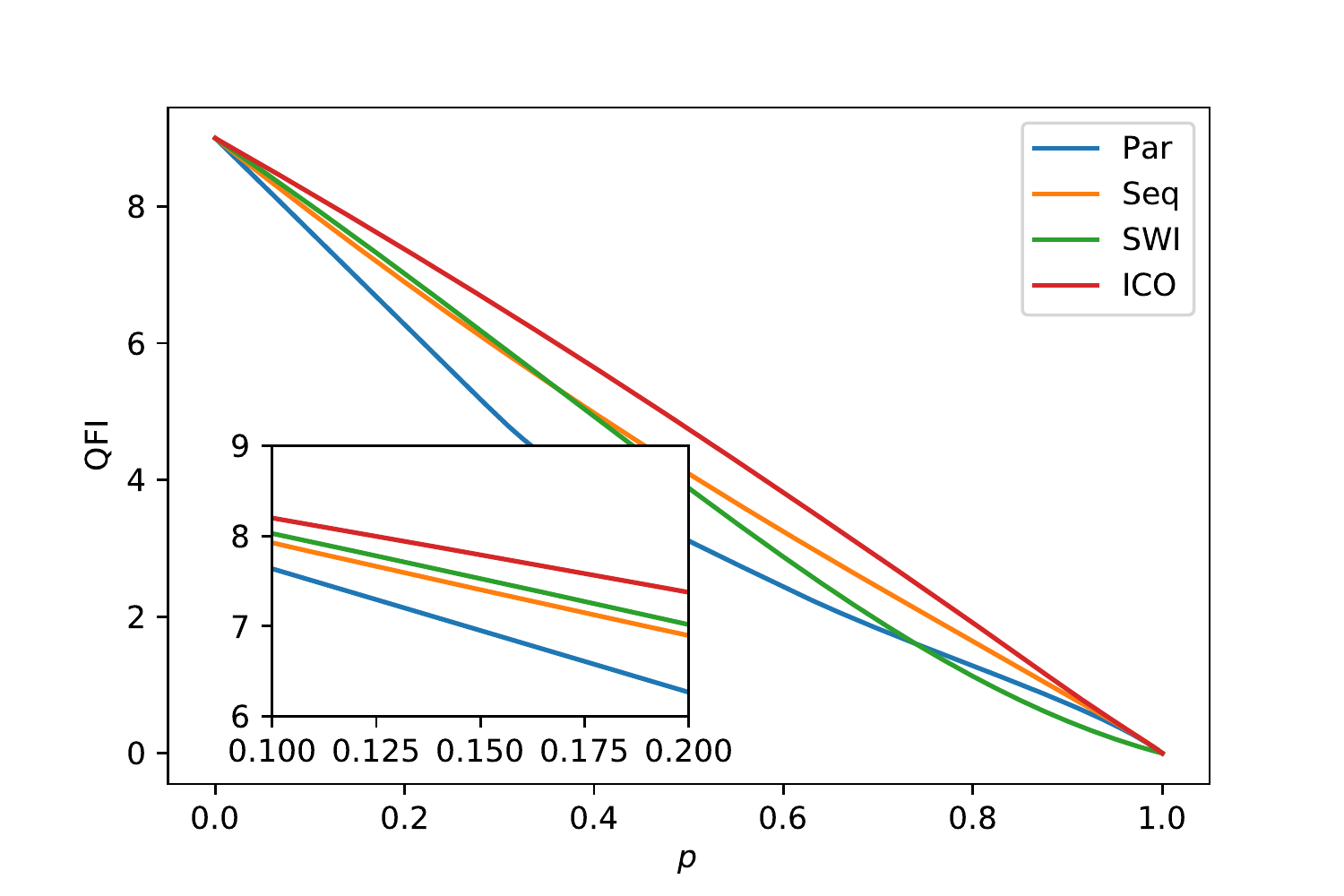}
    \caption{\textbf{Hierarchy of QFI using parallel, sequential, and indefinite-causal-order strategies} for $N=3$, the amplitude damping noise. We take $\phi=1.0$, fix the evolution time $t=1.0$ and vary the decay parameter $p$. To see the gaps between different strategies more clearly we insert an inset plot where the decay parameter $p$ ranges from 0.1 to 0.2.}
    \label{fig:qfi all strategies tripartite}
\end{figure}

\begin{table}[!htbp]
\caption{\label{tab:ICO Sup}\textbf{Hierarchy of QFI using $\mathsf{ICO}$ and $\mathsf{Sup}$.}}
\begin{tabular}{ c@{\qquad}c@{\qquad}c }
\toprule
$p$ & $J^{(\mathsf{Sup})}(N_\phi)$ & $J^{(\mathsf{ICO})}(N_\phi)$ \\
\colrule
0.1 & 8.185 & 8.200 \\
0.2 & 7.364 & 7.375 \\
0.3 & 6.523 & 6.524 \\
0.4 & 5.642 & 5.647 \\
0.5 & 4.725 & 4.743 \\
0.6 & 3.786 & 3.815 \\
0.7 & 2.832 & 2.870 \\
0.8 & 1.871 & 1.909 \\
0.9 & 0.918 & 0.930 \\
\botrule
\end{tabular}
\end{table}

On the other hand, analogous to the $N=2$ case, simple quantum SWITCH strategies without any additional intermediate control operations could have advantage over any definite-causal-order strategies when the decay parameter $p$ is small, but this advantage becomes more insignificant. This should not be surprising since the control can make a bigger difference as $N$ grows.

\subsection{Estimation of randomly sampled channels}
To demonstrate the universality of the hierarchy of different families of strategies considered in the main text, we randomly sample noise channels drawn from an ensemble of CPTP maps defined by Bruzda et al. in \cite{BRUZDA2009320}. In this work we only sample rank-2 qubit channels for $N=2$, which is enough to show the hierarchy. The sampling process is implemented via an open source Python package QuTiP \cite{JOHANSSON20121760,JOHANSSON20131234}. We set an error tolerance of $10^{-8}$, i.e., we claim $J_1>J_2$ only if the gap is no smaller than $10^{-8}$. We find that for 984 of 1000 random channels, a strict hierarchy $J^{(\mathsf{Par})}<J^{(\mathsf{Seq})}<J^{(\mathsf{Sup})}<J^{(\mathsf{ICO})}$ holds, implying that general indefinite-causal-order strategies can provide advantage over causal superposition strategies. In addition, we find that of the same 1000 channels $J^{(\mathsf{Par})}<J^{(\mathsf{SWI})}$ for 34 channels and $J^{(\mathsf{Seq})}<J^{(\mathsf{SWI})}$ only for 1 channel, so with a high probability quantum SWITCH strategies cannot outperform strategies following definite causal order for a random noise channel, which highlights the estimation enhancement from intermediate control in the general case.

\section{Comparison with asymptotic results} \label{app:compar asymp}
In this section we focus on strategies following definite causal order, i.e., parallel and sequential ones, and compare our results and those of the extensively studied asymptotic theory.

\subsection{Preliminaries}
We first introduce some basic notions. If we write the operation-sum representation of the channel $\mathcal E_\phi(\rho)=\sum_{i=1}^r K_{\phi,i}^\dagger \rho K_{\phi,i}$, where $\{K_{\phi,i}\}$ are a set of Kraus operators and $r$ is the rank of the channel, the channel QFI can be evaluated by optimization:
\begin{equation}
    J_Q^{(\mathrm{chan})}(\mathcal E_\phi) = 4\min_{h \in \mathbb H_r} \lVert \alpha \rVert,
\end{equation}
where $\lVert \cdot \rVert$ denotes the operator norm and $\alpha = \sum_i \dot{\tilde{K}}^\dagger_{\phi,i} \dot{\tilde{K}}_{\phi,i}$. Here $\dot{\tilde{K}}_{\phi,i} = \dot K_{\phi,i} - \mathrm i \sum_{j=1}^r h_{ij}K_{\phi,j}$ is nothing but the derivative of an equivalent Kraus representation, given an $r\times r$ Hermitian matrix $h$.

The upper bounds on QFI of $N$ quantum channels have been derived for both sequential and parallel strategies. For parallel strategies an asymptotically tight upper bound is \cite{Fujiwara2008,Demkowicz-Dobrzanski2012}
\begin{equation} \label{eq:parallel bound}
    J^{(\mathsf{Par})}(N_\phi) \le 4 \min_{h \in \mathbb H_r} [N\lVert \alpha \rVert+N(N-1)\lVert\beta\rVert^2],
\end{equation}
where $\beta = \mathrm i \sum_i K_{\phi,i}^{\dagger} \dot {\tilde K}_{\phi,i} $. An upper bound was also derived for sequential strategies \cite{Demkowicz-Dobrzanski14PRL,PRXQuantum.2.010343}:
\begin{equation} \label{eq:sequential bound}
    J^{(\mathsf{Seq})}(N_\phi) \le 4 \min_{h \in \mathbb H_r} [N\lVert \alpha \rVert+N(N-1)\lVert\beta\rVert(\lVert\beta\rVert + 2 \sqrt{\lVert \alpha \rVert})].
\end{equation}
It has been shown that the QFI follows the standard quantum limit if and only if there exists an $h$ such that $\beta=0$ \cite{PRXQuantum.2.010343}. In this case sequential strategies provide no advantage asymptotically, and we have
\begin{equation}
    \lim_{N \rightarrow \infty} \frac{1}{N} J^{(\mathsf{Par})}(N_\phi) = \lim_{N \rightarrow \infty} \frac{1}{N} J^{(\mathsf{Seq})}(N_\phi) = 4\min_{h\in \mathbb H_r\ \mathrm{s.t.}\ \beta=0} \lVert \alpha \rVert.
\end{equation}
We remark that the minimization in Eq. (\ref{eq:parallel bound}) can be efficiently evaluated via SDP \cite{Demkowicz-Dobrzanski2012}.

\subsection{Tightness of QFI bounds in nonasymptotic channel estimation}
We compare our nonasymptotic results and existing asymptotically tight bounds for two types of quantum channels. Apart from the amplitude damping noise described by $K_1^{(\mathrm{AD})} = \dyad{0} + \sqrt{1-p} \dyad{1}$ and $K_2^{(\mathrm{AD})} = \sqrt{p}\ketbra{0}{1}$ considered in the main text, here we also present a second example, where the noise is described by a SWAP-type interaction $V_{\mathrm{int}} = e^{-\mathrm i g\tau H_{\mathrm{SWAP}}}$ between a qubit system $S$ and a qubit environment $E$, where $g$ is the interaction strength, $\tau$ is the interaction time and the Hamiltonian is given by $H_{\mathrm{SWAP}}(\ket{i}_S \ket{j}_E) = \ket{j}_S \ket{i}_E$. The initial environment state is $\ket{0}$, and the Kraus operators can be written as $K_1^{(\mathrm{SWAP})} = \bra{0}_E V_{\mathrm{int}} \ket{0}_E = e^{-\mathrm i g \tau} \dyad{0} + \cos(g\tau) \dyad{1}$, and $K_2^{(\mathrm{SWAP})} = \bra{1}_E V_{\mathrm{int}} \ket{0}_E = -\mathrm i \sin(g\tau)\ketbra{0}{1}$.

We plot the QFI for the two examples in Fig. \ref{fig:compar asymp}. Both of them show the advantage of sequential strategies over parallel ones, and the gaps between exact results of QFI for sequential strategies and the parallel upper bounds given by Eq. (\ref{eq:parallel bound}). 

\begin{figure}[!htbp]
    \captionsetup{position=t}
    \centering
    \sidesubfloat[]{\includegraphics[width=0.45\linewidth]{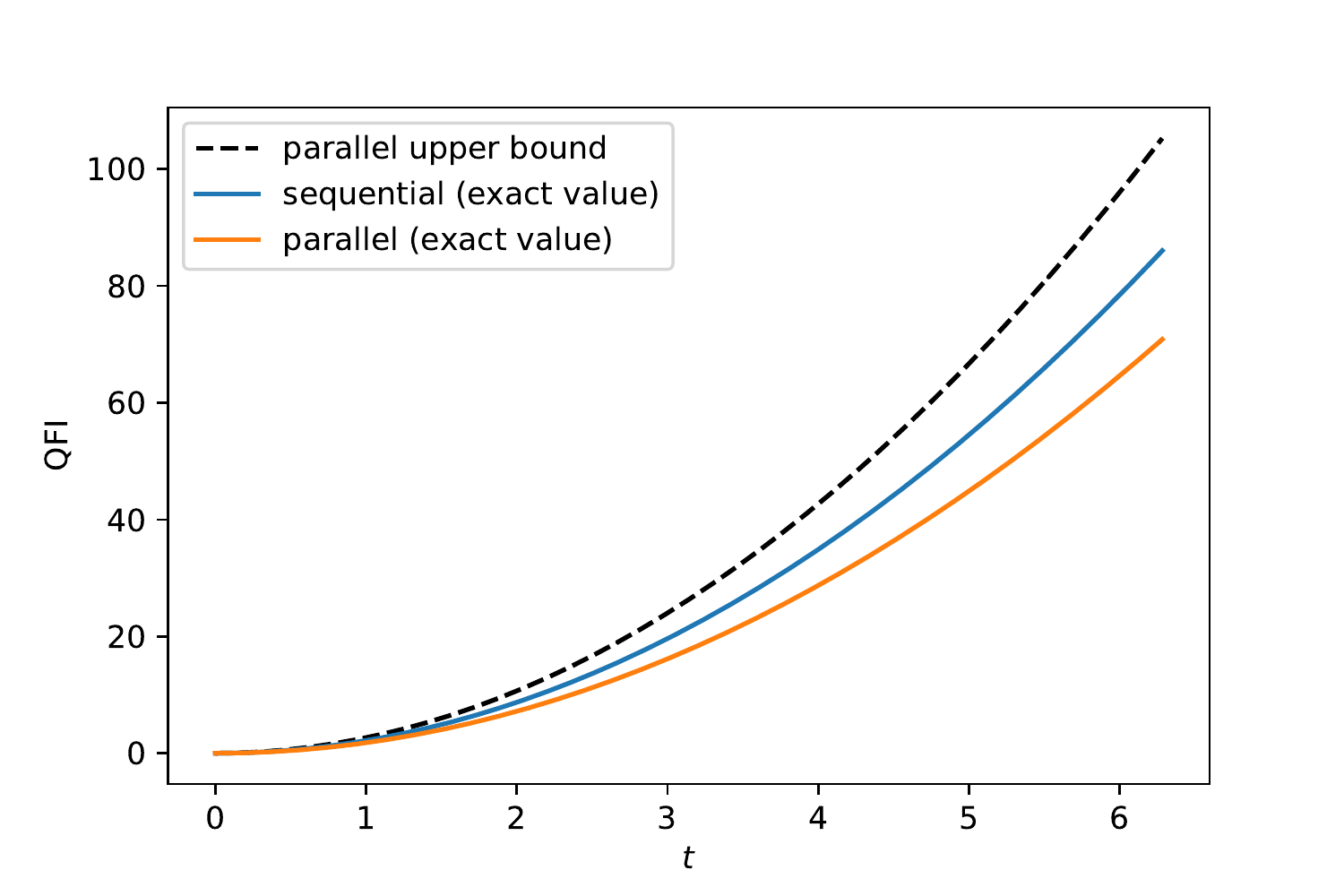}\label{subfig:N_2 ad}} 
    \sidesubfloat[]{\includegraphics[width=0.45\linewidth]{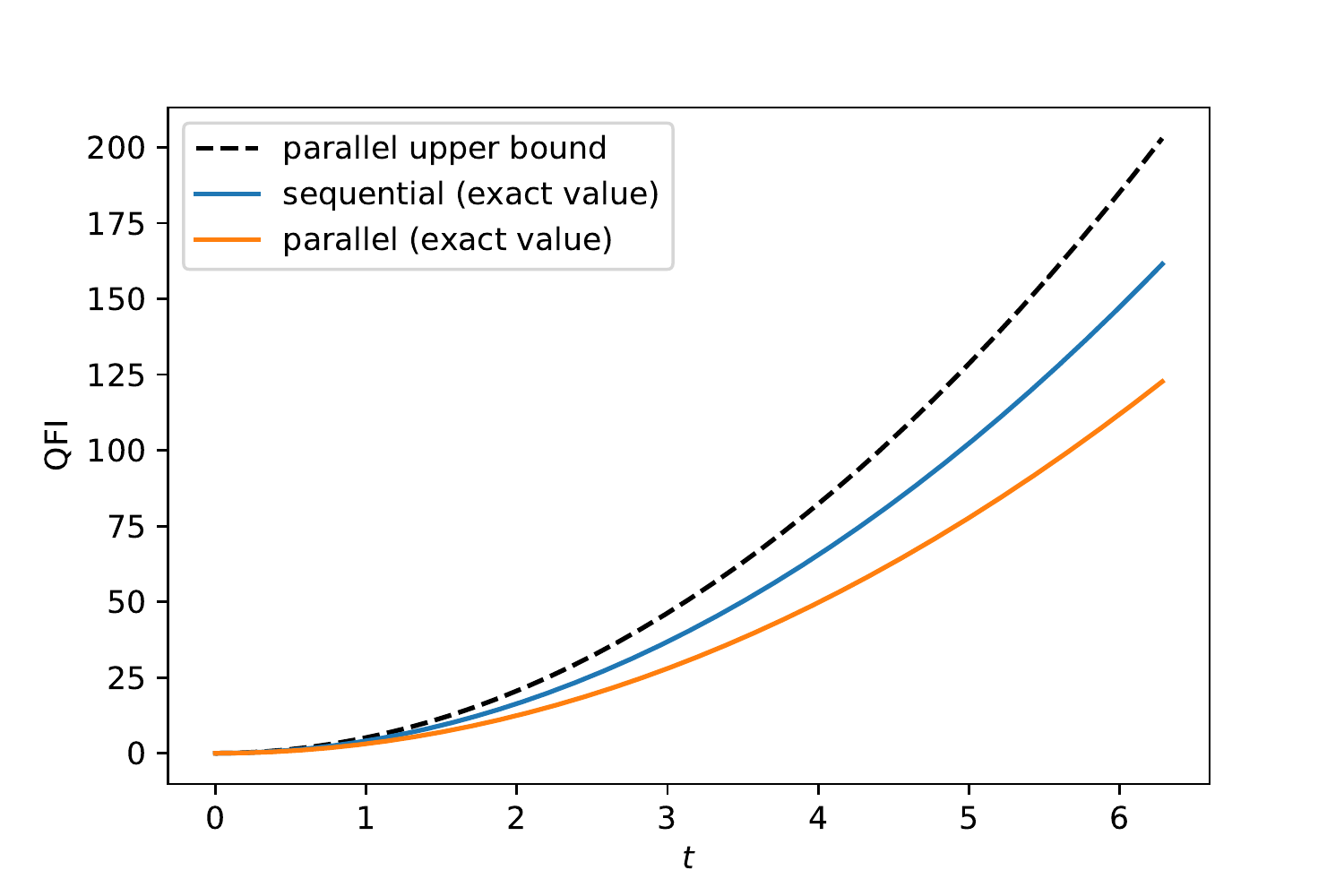}\label{subfig:N_3 ad}} \\
    \sidesubfloat[]{\includegraphics[width=0.45\linewidth]{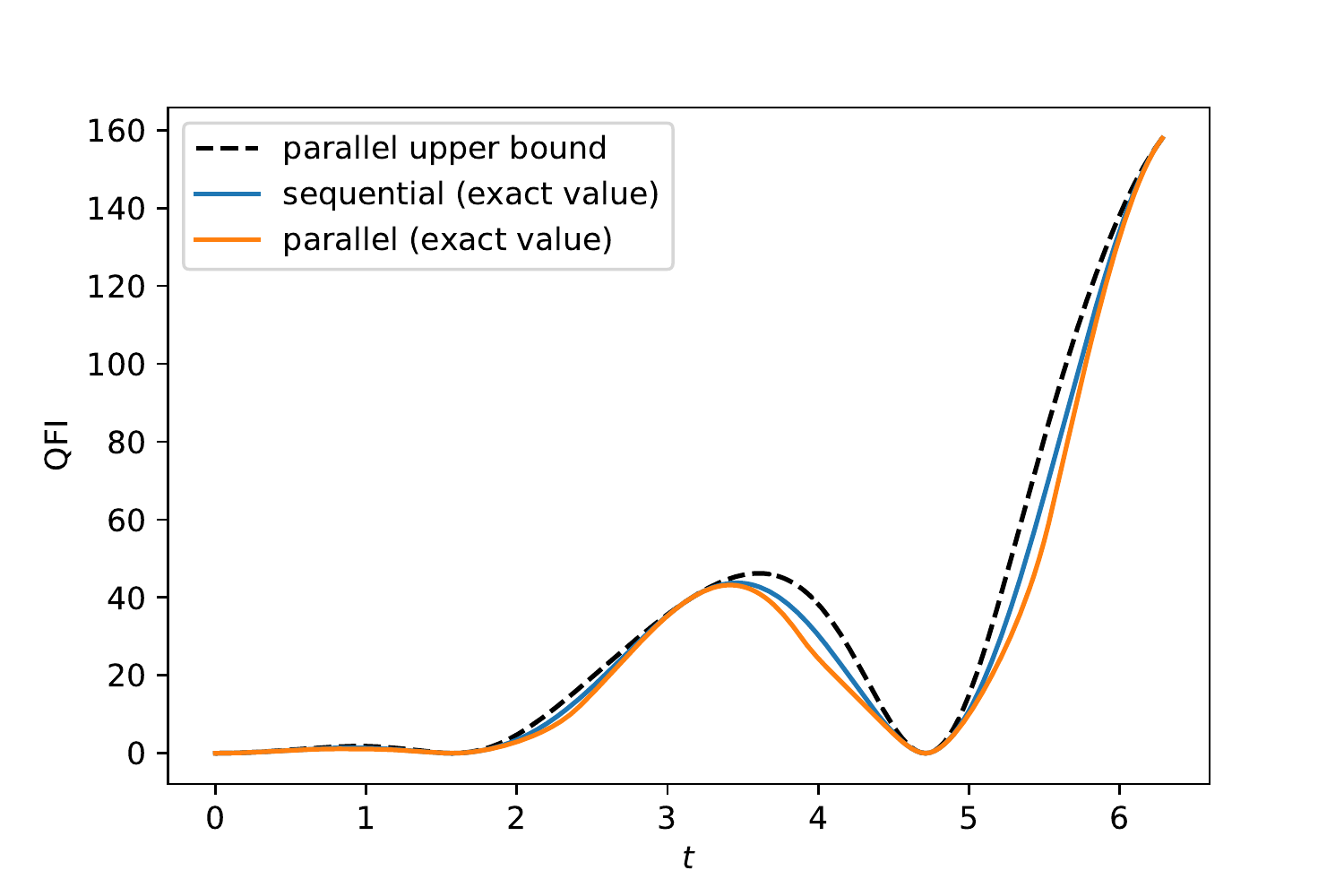}\label{subfig:N_2 swap}} 
    \sidesubfloat[]{\includegraphics[width=0.45\linewidth]{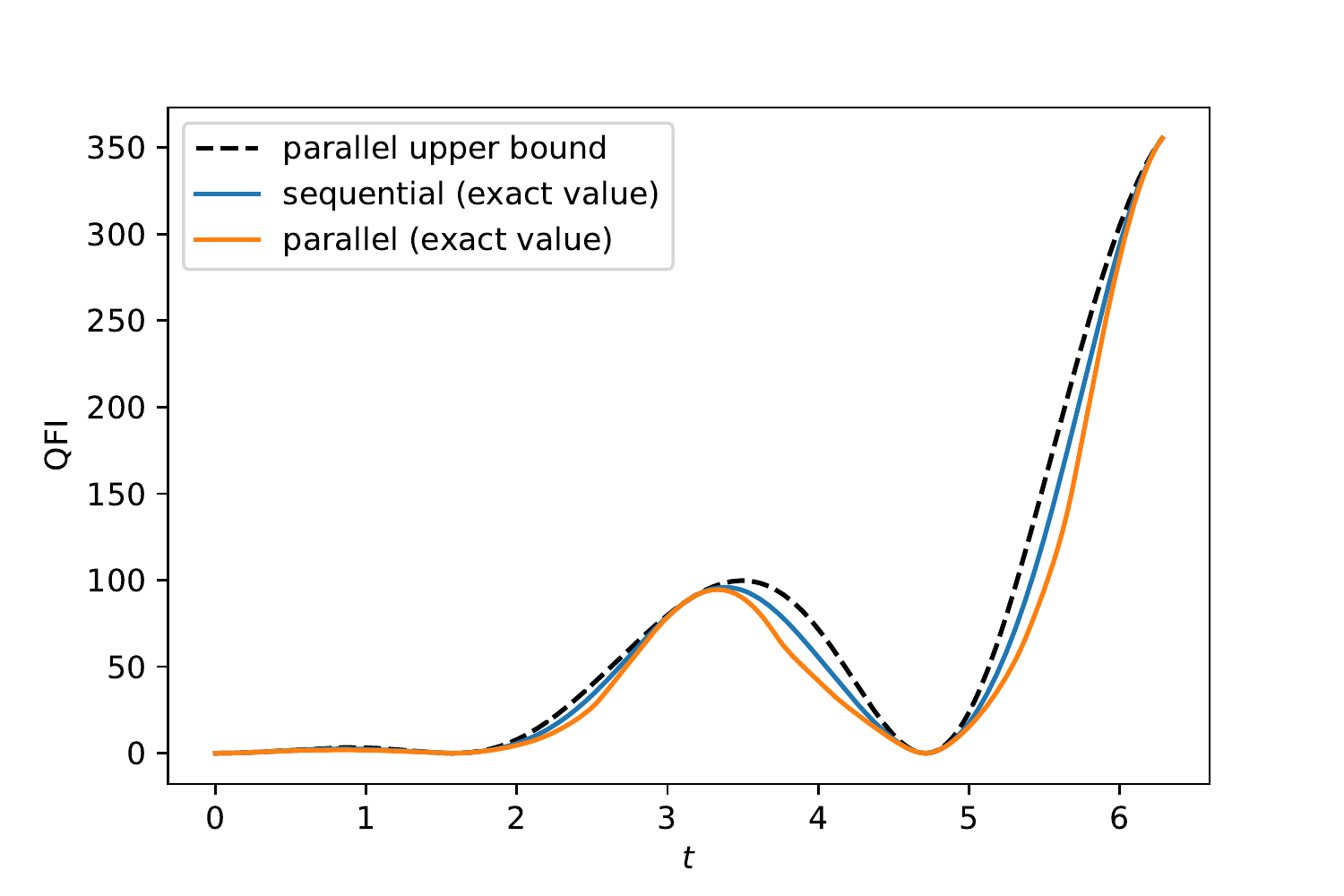}\label{subfig:N_3 swap}}
    \caption{\label{fig:compar asymp}\textbf{Comparison of our results with the existing asymptotically tight QFI bound.} We compare our results to the asymptotically tight bound on the maximal QFI of parallel strategies. We plot the QFI versus the evolution time $t$ for $N=2$, $\phi=1.0$, $g=1.0$, $\tau=t$ and $p=0.5$. \protect\subref{subfig:N_2 ad} $N=2$ and \protect\subref{subfig:N_3 ad} $N=3$ for the amplitude damping noise. \protect\subref{subfig:N_2 swap} $N=2$ and \protect\subref{subfig:N_3 swap} $N=3$ for the SWAP-type noise.} 
\end{figure}

\subsection{Elusive advantage of sequential strategies in the asymptotic limit}
We observe a gap between parallel and sequential strategies for amplitude damping channels and SWAP-type interactions for $N =2$ and $3$. Now we show that for both examples there is no advantage of sequential strategies asymptotically since there exists an $h$ such that $\beta=0$.

For the amplitude damping channel, $K_{\phi,i}^{(\mathrm{AD})} = K_i^{(\mathrm{AD})}U_z(\phi),\ i=1,2$. Direct calculation leads to
\begin{equation}
    \beta^{(\mathrm{AD})} = \left(\frac t 2 + h^{(\mathrm{AD})}_{11}\right) \dyad{0} + \left[h^{(\mathrm{AD})}_{11}-\frac t 2 + \left(h^{(\mathrm{AD})}_{22}-h^{(\mathrm{AD})}_{11}\right)p\right]\dyad{1}.
\end{equation}
To obtain $\beta^{(\mathrm{AD})}=0$ we just need to take $h^{(\mathrm{AD})}_{11}=-t/2$ and $h^{(\mathrm{AD})}_{22} = (2-p)t/2p$.

Similarly, for the SWAP-type interaction we have
\begin{equation}
    \beta^{(\mathrm{SWAP})} = \left(\frac t 2 + h^{(\mathrm{SWAP})}_{11}\right) \dyad{0} + \left[h^{(\mathrm{SWAP})}_{11}-\frac t 2 +  \left(h^{(\mathrm{SWAP})}_{22}-h^{(\mathrm{SWAP})}_{11}\right)\sin^2(g\tau)\right]\dyad{1}.
\end{equation}
Thus there exists $h^{(\mathrm{SWAP})}_{11}=-t/2$ and $h^{(\mathrm{SWAP})}_{22} = \left[2-\sin^2(g\tau)\right]t/2\sin^2(g\tau)$ such that $\beta^{(\mathrm{SWAP})}=0$.

\section{Implementation of optimal strategies with universal quantum gates} \label{app:implementation of strategies}
In this section we apply Algorithm \ref{alg:optimal probe} in the main text to numerically solve for an optimal strategy in the set of sequential and causal superposition ones respectively. The CJ operator of an optimal sequential strategy corresponds to a sequence of isometries with a minimal ancilla-space implementation provided by \cite{Bisio2011PRA}, and can then be decomposed into single-qubit gates and CNOT gates \cite{Iten2016PRA}. By taking advantage of the freedom of choosing a parameter-independent unitary on the final output state, we can adjust the strategy to reduce the CNOT count without affecting the QFI. In terms of an optimal causal superposition strategy we follow the same routine for each sequential strategy branch respectively in the superposition. 

\subsection{Optimal sequential strategy} \label{app:implementation of seq}

Let $\mathcal H_0 = \mathbb C$ and $\mathcal H_{2N+1}=\mathcal H_F$, and we have the CJ operator of a sequential strategy $P \in \mathcal L\left(\mathcal H_F\otimes_{i=0}^{2N} \mathcal H_i\right)$ (an $(N+1)$-step quantum comb), as illustrated in Fig. \ref{fig:sequential strategy with H0}. In this way of relabeling Hilbert spaces we have 
\begin{equation}
P=P^{(N+1)}\ge 0,\ \Tr_{2k-1} P^{(k)} =  \mathds{1}_{2k-2} \otimes P^{(k-1)}\ \mathrm{for}\ k=2,\dots, N+1,\ \Tr P^{(1)}=P^{(0)}=1.
\end{equation}

\begin{figure} [!htbp]
    \centering
    \includegraphics[width=.8\linewidth]{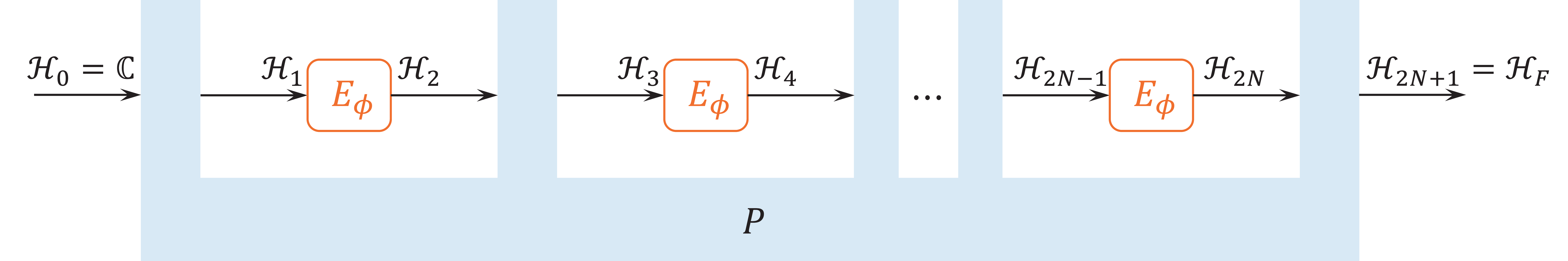}
    \caption{\textbf{Concatenation of a sequential strategy and $N$ quantum channels.} $P$ is the CJ operator of a sequential strategy and $E_\phi$ is the CJ operator of a parametrized quantum channel. $\mathcal \mathcal H_0 = \mathbb C$ is a trivial one-dimensional Hilbert space, and therefore the first step of the strategy is the process of state preparation.}
    \label{fig:sequential strategy with H0}
\end{figure}

According to Theorems 1 and 2 in \cite{Bisio2011PRA},  $P$ corresponds to a sequence of isometries $\{V^{(k)}\}$ for $k=1,\dots,N+1$ by Stinespring dilation:
\begin{equation}
    \mathcal P(\rho) = \Tr_{A_{N+1}} \left[\left(V^{(N+1)}\otimes \mathds{1}_{1,3,\dots,2N-1}\right)\cdots \left(V^{(1)}\otimes \mathds{1}_{2,4,\dots,2N}\right)\rho \left(V^{(1)}\otimes \mathds{1}_{2,4,\dots,2N}\right)^\dagger \cdots \left(V^{(N+1)}\otimes \mathds{1}_{1,3,\dots,2N-1}\right)^\dagger \right]
\end{equation}
for any input state $\rho\in \mathcal L\left(\otimes_{i=0}^{N}\mathcal H_{2i}\right)$, where the whole process corresponding to $P$ is described by an isometry $\mathcal P \in \mathcal L\left(\otimes_{i=0}^{N}\mathcal H_{2i}, \otimes_{i=0}^{N}\mathcal H_{2i+1}\right)$, and in each step a choice of isometry $V^{(k)}\in\mathcal L\left(\mathcal H_{2k-2} \otimes \mathcal H_{A_{k-1}}, \mathcal H_{2k-1} \otimes \mathcal H_{A_{k}}\right)$ with minimal ancilla space is given by
\begin{equation}
    V^{(k)} = \mathds{1}_{2k-1} \otimes \left(P^{(k)*}\right)^{\frac 12}\left[\kett{I}_{2k-1,(2k-1)'}\mathds{1}_{2k-2\rightarrow(2k-2)'}\otimes\left(P^{(k-1)*}\right)^{-\frac 12}\right],
\end{equation}
where $\mathcal H_{A_{k}}=\Supp(P^{(k)*})$ is an ancillary space given by the support of $P^{(k)*}$ with $\mathcal H_{A_{0}}= \mathbb C$, $\mathcal H_{i'}$ is a copy of the Hilbert space $\mathcal H_i$, $\mathds{1}_{2k-2\rightarrow(2k-2)'}:=\sum_i\ket{i}_{(2k-2)'}\bra{i}_{2k-2}$ is an identity map from $\mathcal H_{2k-2}$ to $\mathcal H_{(2k-2)'}$, and $\left(P^{(k-1)*}\right)^{-\frac 12}$ denotes the Moore–Penrose pseudoinverse of $\left(P^{(k-1)*}\right)^{\frac 12}$ with its support on $\mathcal H_{A_{k-1}}$.

As the last isometry $V^{(N+1)}$ preserves the QFI, it is therefore only necessary to consider the implementation of $P^{(N)}$ instead of the full strategy $P^{(N+1)}$. From this explicit construction it follows that the minimal dimension of the ancilla space for implementing the sequential strategy $P$ is $\dim(\mathcal H_{A_{N}})=\rank(P^{(N)})$. In the case of $N=2$ for the amplitude damping noise considered in the main text, it is easy to see that $\dim(\mathcal H_{A_{1}})\le 2$ and $\dim(\mathcal H_{A_{2}}) \le 8$, so $V^{(1)}$ is an isometry from $0$ to (at most) $2$ qubits and $V^{(2)}$ is an isometry from $2$ to (at most) $4$ qubits, as illustrated in Fig. \ref{fig:V1 and V2 seq}.  

\begin{figure} [!htbp]
    \centering
    \mbox{
    \Qcircuit @C=1em @R=.7em {
\lstick{\ket{0}} & \multigate{1}{V^{(1)}} & \gate{\mathcal E_\phi} & \multigate{3}{V^{(2)}} & \gate{\mathcal E_\phi} & \qw \\
\lstick{\ket{0}} & \ghost{V^{(1)}}& \qw & \ghost{V^{(2)}} & \qw & \qw \\
\lstick{\ket{0}} & \qw & \qw & \ghost{V^{(2)}} & \qw & \qw \\
\lstick{\ket{0}} & \qw & \qw & \ghost{V^{(2)}} & \qw & \qw 
}
}
    \caption{\textbf{A sequence of isometries corresponding to a sequential strategy for $N=2$.} The first qubit is the system qubit going through the channel $\mathcal E_\phi$ twice, while the three other qubits are ancillary.}
    \label{fig:V1 and V2 seq}
\end{figure}
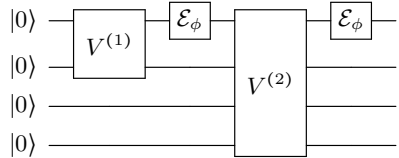

Next, we apply a circuit decomposition of each isometry into single-qubit gates and CNOT gates. In practice it is often desirable to achieve a CNOT count as low as possible. Note that $V^{(1)}$ actually corresponds to the preparation of a $2$-qubit state, which in general requires only one CNOT gate \cite{Znidaric08PRA}. In terms of $V^{(2)}$, an isometry from $2$ to $4$ qubits, the state-of-the-art decomposition scheme is the column-by-column approach which requires at most 54 CNOT gates \cite{Iten2016PRA}. However, as an arbitrary parameter-independent unitary on $3$ ancillae can always be deferred and regrouped into $V^{(3)}$ and therefore does not affect the QFI, in fact we can choose a proper $V^{(2)}$ to further reduce the worst CNOT count to $47$ without changing the QFI. To see this we need to briefly introduce the main ideas behind the column-by-column decomposition scheme.

An isometry $V$ from $m$ to $n$ qubits ($m\le n$) can be represented in the matrix form by $V=U^{\dagger} \mathds{1}(2^n\times 2^m)$, where $U^{\dagger}$ is a $2^n\times 2^n$ unitary matrix and $\mathds{1}(2^n\times 2^m)$ is the first $2^m$ columns of the $2^n \times 2^n$ identity matrix. If we obtain a decomposition of $U^\dagger$, then we can simply initialize the state of the first $n-m$ qubits to $\ket{0}$ to implement $V$. Equivalently we can find a decomposition of $U$ such that $UV=\mathds{1}(2^n\times 2^m)$ and then inverse the circuit representing $U$. The idea is to find a sequence of unitary operations such that $U=U_{2^m-1}\cdots U_0$ transforms $V$ into $\mathds{1}(2^n\times 2^m)$ column by column. More specifically, we first choose a proper $U_0$ to map the first column of $V$ to the first column of $\mathds{1}(2^n\times 2^m)$, i.e., $U_0 V\ket{0}_m=\ket{0}_n$, then choose $U_1$ satisfying $U_1U_0 V\ket{1}_m=\ket{1}_n$ as well as $U_1U_0 V\ket{0}_m=\ket{0}_n$ \ldots until we determine $U_{2^m-1}$.

Here we only focus on $U_0$, the inverse of which can be seen as a process preparing a state $V\ket{0}_m$ from $\ket{0}_n$. In terms of decomposing $V^{(2)}$ from $m=2$ to $n=4$ qubits, preparing a $4$-qubit state in general requires $8$ CNOT gates \cite{Plesch11PRA}. Fortunately, without changing the QFI, we have the freedom to choose a unitary $U_{\mathrm{anc}}$ on the ancillae after applying $V^{(2)}$ such that the state $V^{(2)\prime}\ket{0}_2=U_{\mathrm{anc}}V^{(2)}\ket{0}_2$ can be prepared using only one CNOT gate. This can be seen by dividing the $4$ qubits into two parties, including the single system qubit (in the space $\mathcal H_S$) and the three ancillary qubits (in the space $\mathcal H_A$), and taking the Schimidt decomposition of the $4$-qubit state $V^{(2)}\ket{0}_2$
\begin{equation}
    \ket{\psi}_{SA} := V^{(2)}\ket{0}_2 = \sum_{i=0}^1 \lambda_i \ket{e_i}_S \ket{f_i}_A,
\end{equation}
where $\{\ket{e_i/f_i}_{S/A}\}$ forms an orthonormal basis of $\mathcal H_{S/A}$, and $\{\lambda_i\}$ is a set of nonnegative real numbers satisfying $\sum_i\lambda_i^2=1$. Therefore, to prepare $V^{(2)}\ket{0}_2$, we only need a local unitary on $\mathcal H_S$ to generate $\sum_{i=0}^1 \lambda_i \ket{i}_S\ket{0}_A$, then apply one CNOT gate taking the system qubit as the control to obtain $\sum_{i=0}^1 \lambda_i \ket{i}_S\ket{i}_A$, and finally apply local unitary operations $U_S=\sum_i\ket{e_i}_S\bra{i}_S$ on the system and $U_A=\sum_i\ket{f_i}_A\bra{i}_A$ on the ancillae respectively. If we take $U_{\mathrm{anc}} = U_A^{\dagger}$, then it is easy to see that $V^{(2)\prime}\ket{0}_2=U_{\mathrm{anc}}V^{(2)}\ket{0}_2=\sum_{i=0}^1 \lambda_i \ket{e_i}_S \ket{i}_A$ can thus be prepared using one CNOT gate. This choice of $V^{(2)\prime}$ saves $7$ CNOT gates compared to the general state preparation scheme, and leads to a worst CNOT count of $47$ in total.

Now we present numerical results of the circuit implementation of an optimal sequential strategy. The decomposition of ismometries is implemented using the Mathematica package UniversalQCompiler \cite{Iten19Universal} based on the method described above. As in the main text, we consider the amplitude damping noise and take $N=2$, $\phi=1.0$, $p=0.5$ and $t=1.0$. The circuits implementing $V^{(1)}$ and $V^{(2)\prime}$ are illustrated in Fig. \ref{fig:decomp seq}. The state preparation $V^{(1)}$ requires $1$ CNOT gate and the intermediate control operation $V^{(2)\prime}$ requires $33$ CNOT gates.

\begin{figure} [!htbp]
    \captionsetup{position=bottom}
    \centering
    \subfloat[Decomposition of $V^{(1)}$. For simplicity the angles of single-qubit rotation gates are not depicted.]{\mbox{
    \Qcircuit @C=1em @R=.7em {
\push{\rule{20em}{0em}} & \lstick{\ket{0}} & \gate{R_y} & \gate{R_z} & \ctrl{1} & \gate{R_y} & \qw & \push{\rule{20em}{0em}}  \\
\push{\rule{20em}{0em}} & \lstick{\ket{0}} & \gate{R_y} & \qw & \targ & \qw  & \qw & \push{\rule{20em}{0em}}
}
}\label{subfig:seq V1 decomp}} \\
    \subfloat[Decomposition of $V^{(2)\prime}=U_{\mathrm{anc}}V^{(2)}$. For simplicity single-qubit gates, which might be required in addition to CNOT gates, are not depicted.]{\mbox{
    \Qcircuit @C=.6em @R=.7em {
& \ctrl{1} & \ctrl{1} & \ctrl{2} & \qw  & \ctrl{2} & \ctrl{3} 
& \qw  & \ctrl{3} & \qw  & \qw  & \qw  & \qw  & \qw  & \qw  & \targ & 
\targ & \targ & \targ & \targ & \targ & \targ & \ctrl{2} & \qw  & 
\ctrl{3} & \qw  & \qw  & \qw  & \targ & \targ & \targ & \targ & \targ & \ctrl{1} & \qw \\
& \targ & \targ & \qw  & \ctrl{1} & \qw  & \qw  & \ctrl{2} & 
\qw  & \ctrl{1} & \qw  & \ctrl{2} & \targ & \targ & \targ & \ctrl{-1} 
& \qw  & \ctrl{-1} & \qw  & \ctrl{-1} & \qw  & \ctrl{-1} & \qw  & \qw 
 & \qw  & \targ & \targ & \targ & \qw  & \qw  & \ctrl{-1} & \qw  & 
\ctrl{-1} & \targ & \qw \\
\lstick{\ket{0}} & \qw  & \qw  & \targ & \targ & \targ & \qw  & \qw  & \qw  & \targ & \ctrl{1} & \qw  & \qw  & \ctrl{-1} & \qw  & \qw  & \qw  & \qw  
& \ctrl{-2} & \qw  & \qw  & \qw  & \targ & \ctrl{1} & \qw  & \qw  & 
\ctrl{-1} & \qw  & \qw  & \ctrl{-2} & \qw  & \qw  & \qw & \qw & \qw \\
\lstick{\ket{0}} & \qw  & \qw  & \qw  & \qw  & \qw  & \targ & \targ & \targ & \qw  & \targ & \targ & \ctrl{-2} & \qw  & \ctrl{-2} & \qw  & 
\ctrl{-3} & \qw  & \qw  & \qw  & \ctrl{-3} & \qw  & \qw  & \targ & \targ & \ctrl{-2} & \qw  & \ctrl{-2} & \ctrl{-3} & \qw  & \qw  &
\ctrl{-3} & \qw  & \qw & \qw
}
}\label{subfig:seq V2 decomp}} 
    \caption{\label{fig:decomp seq}\textbf{Decomposition of isometries corresponding to an optimal sequential strategy for $N=2$.} We apply $V^{(2)\prime}$ instead of $V^{(2)}$ to achieve the maximal QFI with fewer CNOT gates.} 
\end{figure}
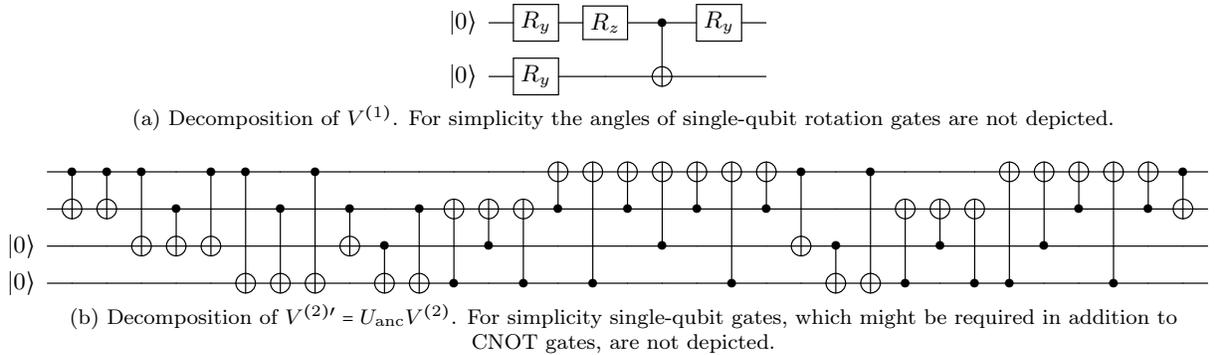

\subsection{Optimal causal superposition strategy} \label{app:implementation of sup}
A causal superposition strategy for estimating $N$ channels can be implemented by an $N!$-dim quantum control system entangled with $N!$ sequential strategies of applying the channels:
\begin{equation} \label{app:eq:causal superposition strategy}
    P = \dyad{P}\ \mathrm{for}\ \ket{P}=\sum_{\pi\in S_N} \ket{P^{\pi}}\ket{\pi}_C,
\end{equation}
where $\{\ket{\pi}_C\}$ forms an orthonormal basis of the Hilbert space $\mathcal H_C$ of the control system, and each $P^{\pi}=\dyad{P^\pi}$ is a sequential strategy. Once we obtain an optimal causal superposition strategy by applying Algorithm \ref{alg:optimal probe}, we can apply the circuit decomposition for each sequential strategy in the superposition, following the method described in Appendix \ref{app:implementation of seq}. Taking account of the permutation symmetry, we can choose an optimal causal superposition strategy such that apart from the execution order of two channels, sequential strategies in the decomposition of the strategy are the same, containing the same state preparation and intermediate control.  

As a concrete example, we again take $N=2$, $\phi=1.0$, $p=0.5$ and $t=1.0$ for the amplitude damping noise and present numerical results of the circuit implementation of an optimal causal superposition strategy. As illustrated in Fig. \ref{fig:V1 and V2 sup}, we use the qubit $\ket{\psi}_C$ to coherently control which sequential order is executed. Due to the permutation invariance of the optimal strategy, we can simply control the query order of the identical channels while fixing $V^{(1)}$ and $V^{(2)}$ for all sequential orders. In view of this, generally we can use a $(2N-1)$-quantum SWITCH to control the order of $N$ channels $\mathcal E_\phi$ and $N-1$ intermediate control operations $V^{(i)}$.

\begin{figure} [!htbp]
    \centering
\mbox{
\Qcircuit @C=1em @R=.7em {
& \mbox{If $\ket{\psi}_C = \ket{0}_C$,}  & & & & \\
\lstick{\ket{0}} & \multigate{1}{V^{(1)}} & \gate{\mathcal E_\phi^{(1)}} & \multigate{3}{V^{(2)}} & \gate{\mathcal E_\phi^{(2)}} & \qw \\
\lstick{\ket{0}} & \ghost{V^{(1)}}& \qw & \ghost{V^{(2)}} & \qw & \qw \\
\lstick{\ket{0}} & \qw & \qw & \ghost{V^{(2)}} & \qw & \qw \\
\lstick{\ket{0}} & \qw & \qw & \ghost{V^{(2)}} & \qw & \qw 
}
} \hspace{6em}
\mbox{
\Qcircuit @C=1em @R=.7em {
& \mbox{If $\ket{\psi}_C = \ket{1}_C$,} & & & & \\
\lstick{\ket{0}} & \multigate{1}{V^{(1)}} & \gate{\mathcal E_\phi^{(2)}} & \multigate{3}{V^{(2)}} & \gate{\mathcal E_\phi^{(1)}} & \qw \\
\lstick{\ket{0}} & \ghost{V^{(1)}}& \qw & \ghost{V^{(2)}} & \qw & \qw \\
\lstick{\ket{0}} & \qw & \qw & \ghost{V^{(2)}} & \qw & \qw \\
\lstick{\ket{0}} & \qw & \qw & \ghost{V^{(2)}} & \qw & \qw 
}
}
    \caption{\textbf{Sequences of isometries corresponding to each sequential order in the causal superposition for $N=2$.} The first qubit of the circuit is the system qubit, and the query order of two identical channels $\mathcal E_\phi^{(1)}$ and $\mathcal E_\phi^{(2)}$ is entangled with the state of the control qubit $\ket{\psi}_C$. When $\ket{\psi}_C$ is a superposition of the two states shown in the figure, the causal order is also in a superposition given by Eq. (\ref{app:eq:causal superposition strategy}).}
    \label{fig:V1 and V2 sup}
\end{figure}
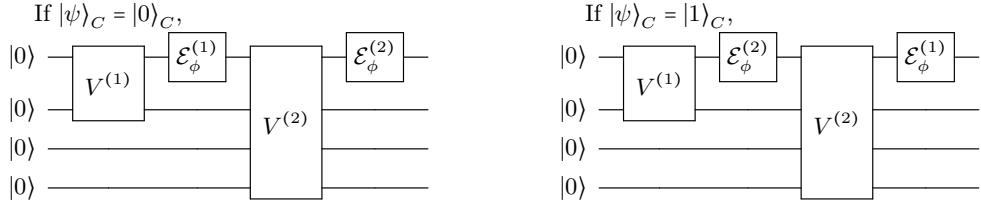

Further decomposition shows that each sequential branch requires one CNOT gate for state preparation $V^{(1)}$ and 36 CNOT gates for intermediate control $V^{(2)}$, as illustrated in Fig. \ref{fig:decomp sup}.

\begin{figure} [!htbp]
    \captionsetup{position=bottom}
    \centering
    \subfloat[Decomposition of $V^{(1)}$. For simplicity the angles of single-qubit rotation gates are not depicted.]{\mbox{
    \Qcircuit @C=1em @R=.7em {
\push{\rule{20em}{0em}} & \lstick{\ket{0}} & \gate{R_y} & \gate{R_z} & \ctrl{1} & \gate{R_y} & \qw & \push{\rule{20em}{0em}}  \\
\push{\rule{20em}{0em}} & \lstick{\ket{0}} & \gate{R_y} & \qw & \targ & \qw  & \qw & \push{\rule{20em}{0em}}
}
}\label{subfig:sup V1 decomp}} \\
    \subfloat[Decomposition of $V^{(2)}$. For simplicity single-qubit gates, which might be required in addition to CNOT gates, are not depicted.]{\mbox{
    \Qcircuit @C=.6em @R=.7em {
& \ctrl{1} & \ctrl{1} & \ctrl{2} & \qw  & \ctrl{2} & \qw  & 
\ctrl{3} & \qw  & \ctrl{3} & \qw  & \qw  & \qw  & \qw  & \qw  & \qw  
& \qw  & \qw  & \qw  & \targ & \targ & \targ & \targ & \targ & \targ 
& \targ & \ctrl{2} & \qw  & \ctrl{3} & \qw  & \qw  & \qw  & \ctrl{1} 
& \targ & \targ & \targ & \ctrl{1} & \qw  \\
& \targ & \targ & \qw  & \ctrl{1} & \qw  & \qw  & \qw  & 
\ctrl{2} & \qw  & \targ & \targ & \targ & \ctrl{1} & \qw  & \ctrl{2} 
& \targ & \targ & \targ & \ctrl{-1} & \qw  & \ctrl{-1} & \qw  & 
\ctrl{-1} & \qw  & \ctrl{-1} & \qw  & \qw  & \qw  & \targ & \targ & \targ & \targ & \ctrl{-1} & \qw  & \ctrl{-1} & \targ & \qw  \\
\lstick{\ket{0}}  & \qw  & \qw  & \targ & \targ & \targ & \ctrl{1} & \qw  & \qw  & \qw  & \qw  & \ctrl{-1} & \qw  & \targ 
& \ctrl{1} & \qw  & \qw  & \ctrl{-1} & \qw  & \qw  & \qw  & \qw  & 
\ctrl{-2} & \qw  & \qw  & \qw  & \targ & \ctrl{1} & \qw  & \qw  & 
\ctrl{-1} & \qw  & \qw  & \qw  & \ctrl{-2} & \qw  & \qw & \qw  \\
\lstick{\ket{0}} & \qw  & \qw  & \qw  & \qw  & \qw  & 
\targ & \targ & \targ & \targ & \ctrl{-2} & \qw  & \ctrl{-2} & \qw  & 
\targ & \targ & \ctrl{-2} & \qw  & \ctrl{-2} & \qw  & \ctrl{-3} & \qw 
 & \qw  & \qw  & \ctrl{-3} & \qw  & \qw  & \targ & \targ & \ctrl{-2} 
& \qw  & \ctrl{-2} & \qw  & \qw  & \qw  & \qw & \qw & \qw  
}
}\label{subfig:sup V2 decomp}} 
    \caption{\label{fig:decomp sup}\textbf{Decomposition of isometries corresponding to one causal order in the decomposition of an optimal causal superposition strategy for $N=2$.} We have already taken advantage of the freedom to choose a $V^{(2)}$ implemented by fewer CNOT gates, as explained in Appendix \ref{app:implementation of seq}.} 
\end{figure}
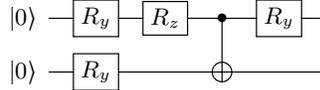
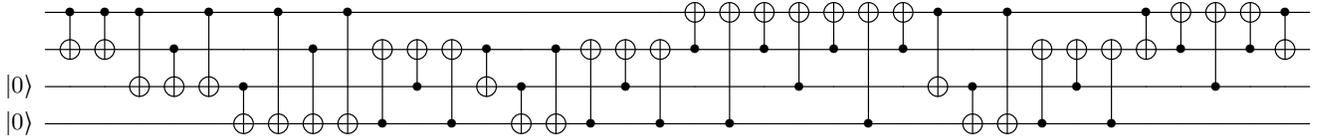


\begin{thebibliography}{75}
\bibitem [{Asy()}]{Asymptotic_note}%
  \BibitemOpen
  \href@noop {} {}\bibinfo {note} {For $\nu \rightarrow \infty$, the maximum
  likelihood estimator is unbiased and the quantum Cramér-Rao bound is
  achieved. This should not be confused with the nonasymptotic regime (where
  $N$ is finite).}\BibitemShut {Stop}%
\bibitem [{pur()}]{purification_note}%
  \BibitemOpen
  \href@noop {} {}\bibinfo {note} {For any CJ operator $P$ of rank $r_P$, we
  can always choose an $r_P$-dimensional ancillary space to purify $P$, which
  only extends the global future space $\mathcal H_F$ without changing the
  causal relations between queries to the channel $\mathcal E_\phi$. As QFI is
  monotonic under CPTP maps, the optimal QFI can be achieved by a purified
  process, and thus it is sufficient to only consider pure
  processes.}\BibitemShut {Stop}%
\bibitem [{min()}]{minimax_note}%
  \BibitemOpen
  \href@noop {} {}\bibinfo {note} {We also suppose that the strategy set
  $\mathsf P$ is compact in order to apply Fan's minimax theorem.}\BibitemShut
  {Stop}%
\bibitem [{inv()}]{invariant_strategy_note}%
  \BibitemOpen
  \href@noop {} {}\bibinfo {note} {Concretely, $\tilde P^{(\mathrm{opt})}=
  \Tr_F P^{(\mathrm{opt})}$ is permutation invariant in this case.}\BibitemShut
  {Stop}%
\bibitem [{seq()}]{seq_bound_note}%
  \BibitemOpen
  \href@noop {} {}\bibinfo {note} {Note that the best existing upper bound on
  the QFI of sequential strategies
  \cite{Demkowicz-Dobrzanski14PRL,PRXQuantum.2.010343} is greater than or equal
  to that of parallel strategies and is thus omitted here. The sequential upper
  bound is also harder to evaluate numerically because it cannot be readily
  formulated as a convex optimization problem.}\BibitemShut {Stop}%
\end{thebibliography}
\end{document}